\documentclass[a4paper]{amsart}
\usepackage{fullpage}
\usepackage{amsmath,amsthm,amssymb}
\usepackage{epic,eepic}
\usepackage{epsfig,cite,oldgerm}
\usepackage{float}

\newtheorem{lemma}{Lemma}

\newtheorem{conditions}{ }

\def\lprod{\mathop{\prod{\mkern-29.5mu}{\mathbf\longleftarrow}}}

\begin{document}

\title[Scalar products in six-vertex models]
{An Izergin-Korepin procedure for calculating scalar products in six-vertex models}

\author{M Wheeler}

\address{Department of Mathematics and Statistics,
                 University of Melbourne,
                 Parkville, Victoria, Australia.}
\email{m.wheeler@ms.unimelb.edu.au}


\begin{abstract}

Using the framework of the algebraic Bethe Ansatz, we study the scalar product of the inhomogeneous XXZ spin-$\frac{1}{2}$ chain. Inspired by the Izergin-Korepin procedure for evaluating the domain wall partition function, we obtain a set of conditions which uniquely determine the scalar product. Assuming the Bethe equations for one set of variables within the scalar product, these conditions may be solved to produce a determinant expression originally found by Slavnov. We also consider the inhomogeneous XX spin-$\frac{1}{2}$ chain in an external magnetic field. Repeating our earlier procedure, we find a set of conditions on the scalar product of this model and solve them in the presence of the Bethe equations. The expression obtained is in factorized form.

\end{abstract}

\maketitle

\setcounter{section}{0}

\section{Introduction}\label{introduction}

Scalar products are objects of essential interest in models solvable by the quantum inverse scattering method/algebraic Bethe Ansatz \cite{kbi}. In the case of two-dimensional integrable models, the definition of the scalar product may be phrased as follows: Let $V_a,V_b$ denote copies of the vector space $\mathbb{C}^2$, and consider a generic two-dimensional model with $R$-matrix $R_{ab}(u,v) \in {\rm End}(V_a \otimes V_b)$, monodromy matrix 
\begin{align}
T_a(u) 
=
\left(
\begin{array}{cc}
A(u) & B(u)
\\ 
C(u) & D(u)
\end{array}
\right)_a
\in 
{\rm End}(V_a),
\end{align}
and pseudo-vacuum states $\langle 0|, | 0\rangle$ satisfying
\begin{align}
&\langle 0| A(u) = a(u) \langle 0|,& 
&\langle 0| B(u) = 0,& 
&\langle 0| D(u) = d(u) \langle 0|,&
\\
&A(u) |0\rangle = a(u) |0\rangle,&
&C(u) |0\rangle = 0,&
&D(u) |0\rangle = d(u) |0\rangle,&
\\
&\langle 0| 0\rangle = 1.&
\end{align}
The scalar product of the model is then given by
\begin{align}
S_N\Big(\{u\},\{v\}\Big)
=
\langle 0| C(u_1)\ldots C(u_N)
B(v_1)\ldots B(v_N) |0\rangle,
\label{SP}
\end{align}
where the operators $C(u)$ and $B(v)$ obey a non-trivial commutation relation which is recovered as a single component of the bilinear equation
\begin{align}
R_{ab}(u,v) T_a(u) T_b(v) = T_b(v) T_a(u) R_{ab}(u,v).
\label{YBalg}
\end{align}
Finding simple and exact expressions for scalar products is important for at least two reasons: 1. In the limit when the variables $\{u\}$ and $\{v\}$ are equal to the same solution of the Bethe equations, they reduce to the norm-squared of a Bethe eigenstate; 2. Algebraically, they are closely related to correlation functions and constitute the first step in the calculation of such quantities. Nevertheless, the explicit calculation of the scalar product is a difficult task, owing to the complicated nature of the algebraic relations contained in (\ref{YBalg}).  

A major development in the theory came in \cite{sla}, when N~A~Slavnov found a determinant expression for the scalar product (\ref{SP}) of models with an $A_1^{(1)}$ type $R$-matrix \cite{jim} in the case when one set of variables $\{v\}$ satisfies the Bethe equations, whilst the other set $\{u\}$ is left free. Aside from its intrinsic beauty, the Slavnov scalar product formula has many practical uses. The simplest of these is seen by setting the free set of variables $\{u\}$ equal to the set $\{v\}$, when it specializes to a formula for the norm-squared of a Bethe eigenstate, originally proposed in \cite{gau} and proved in \cite{kor}. The Slavnov formula was also indispensable to the work of \cite{kmst}, in which a compact expression was found for a generating functional of correlation functions in generic $A_1^{(1)}$ type models. More recently, by virtue of its determinant form, it was shown that the Slavnov scalar product can be viewed as a specialization of a $\tau$-function of the KP hierarchy \cite{fwz5} or more simply as a solution of discrete Hirota equations \cite{fs}.

The main purpose of this paper is to provide a recursive, diagrammatic proof of the Slavnov formula in the particular case of the inhomogeneous XXZ spin-$\frac{1}{2}$ model. In order to do this, we resort to the relationship between the XXZ model and the six-vertex model of statistical mechanics. We write the XXZ scalar product as a sum over lattice states in the six-vertex model, subject to certain boundary conditions. The lattice sum is shown to satisfy a set of four conditions which determine it uniquely, and then the Slavnov formula is proposed as a solution. This type of procedure is motivated by the work of V~E~Korepin and A~G~Izergin on the domain wall partition function of the six-vertex model, whereby a set of conditions was found in \cite{kor} and subsequently solved in \cite{ize}.  

In \cite{fel2}, B~U~Felderhof defined a transfer matrix which commutes with the Hamiltonian of the XY spin-$\frac{1}{2}$ chain in an external magnetic field. The transfer matrix thus constructed was parametrized in terms of Jacobi elliptic functions. By specializing to the trigonometric limit, one obtains a transfer matrix which commutes with the Hamiltonian of the XX spin-$\frac{1}{2}$ chain in an external field. We shall consider this latter case in the second part of the paper, and refer to it as the trigonometric Felderhof model. It previously appeared in \cite{da1} as the first in a hierarchy of vertex models, and its domain wall partition function was calculated in \cite{fwz1} using an Izergin-Korepin type of procedure. In reflection of the free-fermionic nature of this model the domain wall partition function has a product form, in which all zeros clearly factorize. In this work we calculate the scalar product of the trigonometric Felderhof model, once again assuming the Bethe equations for one set of variables. We find that the scalar product also has a product form, with all zeros factorized. Some of the zeros are of the type found in the domain wall partition function, while the rest are roots of the Bethe equations.

The outline of the paper is as follows: In \S 2 we introduce basic facts related to the XXZ spin-$\frac{1}{2}$ chain, describing its solution via the algebraic Bethe Ansatz. In \S 3 we define the domain wall partition function of the six-vertex model, and review the method of its evaluation due to \cite{kor,ize}. This serves as a useful introduction to \S 4, where we present our Izergin-Korepin procedure for calculating the scalar product. We set up a sequence of intermediate scalar products $\{S_n\}_{0 \leq n \leq N}$ which were originally considered in \cite{kmt}, and write down conditions which uniquely determine them, before proposing a solution and showing that it is valid. In \S 5 we introduce the XX spin-$\frac{1}{2}$ chain in an external magnetic field, and describe the method of diagonalizing its Hamiltonian due to Felderhof. The calculation of the trigonometric Felderhof domain wall partition function is reviewed in \S 6, before we evaluate its scalar product in \S 7. Once again, an Izergin-Korepin approach is employed in calculating these quantities. 

\section{XXZ spin-$\frac{1}{2}$ chain}
\label{xxz-intro}

In this section we introduce the basics of the XXZ spin-$\frac{1}{2}$ chain, describing its solution via the quantum inverse scattering method/algebraic Bethe Ansatz. For a more general introduction, the reader is referred to chapters VI and VII of the book \cite{kbi}.

\subsection{Space of states $V$}

The finite length XXZ spin-$\frac{1}{2}$ chain consists of a one-dimensional lattice with $M$ sites. Each site $m$ is occupied by a spin-$\frac{1}{2}$ fermion and has a corresponding two-dimensional vector space $V_m$ with the basis
\begin{align}
{\rm Basis}(V_m)
=
\Big\{ \uparrow_m,\downarrow_m \Big\},
\end{align}
where for convenience we have adopted the notations 
\begin{align}
\uparrow_m\  
=
\Big(
\begin{array}{c}
1 \\ 0
\end{array}
\Big)_m,
\quad 
\downarrow_m\ 
=
\Big(
\begin{array}{c}
0 \\ 1
\end{array}
\Big)_m
\label{V_i}.
\end{align}
Physically speaking, $\uparrow_m$ and $\downarrow_m$ represent the spin eigenstates of a spin-$\frac{1}{2}$ fermion at site $m$. The space of states $V$ is defined as
\begin{align}
V = V_1 \otimes \cdots \otimes V_M
\end{align}
and it is the goal of the algebraic Bethe Ansatz to find states within $V$ which are eigenvectors of the Hamiltonian (see \S \ref{s-ham}). The two simplest states in $V$ are those for which all spins are up/down, and we prescribe them the notation
\begin{align}
|\Uparrow_M\rangle = \bigotimes_{m=1}^{M} \uparrow_m,
\quad
|\Downarrow_M\rangle = \bigotimes_{m=1}^{M} \downarrow_m.
\end{align}
We will also make use of the definitions
\begin{align}
|\Uparrow_{N/M}\rangle
=
\bigotimes_{1 \leq m \leq N}
\uparrow_m
\bigotimes_{N < m \leq M}
\downarrow_m,
\quad\quad
|\Downarrow_{N/M}\rangle
=
\bigotimes_{1 \leq m \leq N}
\downarrow_m
\bigotimes_{N < m \leq M}
\uparrow_m
\end{align}
for states whose first $N$ spins are up/down, with all the remaining spins being down/up, respectively. 

Analogous definitions apply in the construction of dual spaces. To each site $m$ we associate the dual vector space $V_m^{*}$ with the basis
\begin{align}
{\rm Basis}(V_m^{*})
=
\Big\{ 
\uparrow_m^{*},\downarrow_m^{*}
\Big\},
\end{align}
where we have adopted the notations
\begin{align}
\uparrow^{*}_m\ 
=
\left(
\begin{array}{cc}
1 & 0 
\end{array}
\right)_m,
\quad 
\downarrow^{*}_m\ 
=
\left(
\begin{array}{cc}
0 & 1
\end{array} 
\right)_m
\label{V_i*},
\end{align}
and from this we construct the dual space of states 
\begin{align}
V^{*} = V_1^{*} \otimes \cdots \otimes V_M^{*}.
\end{align}
Similarly to before we shall write
\begin{align}
\langle \Uparrow_M| = \bigotimes_{m=1}^{M} \uparrow^*_m,
\quad 
\langle \Downarrow_M| = \bigotimes_{m=1}^{M} \downarrow^*_m
\end{align}
for the dual states whose spins are all up/down, and
\begin{align}
\langle \Uparrow_{N/M}|
=
\bigotimes_{1 \leq m \leq N}
\uparrow^{*}_m
\bigotimes_{N < m \leq M}
\downarrow^{*}_m,
\quad\quad
\langle \Downarrow_{N/M}|
=
\bigotimes_{1 \leq m \leq N}
\downarrow^{*}_m
\bigotimes_{N < m \leq M}
\uparrow^{*}_m
\end{align}
for dual states whose first $N$ spins are up/down, with all remaining spins being down/up, respectively.

\subsection{Pauli matrices}

We define the Pauli matrices
\begin{align}
\sigma_m^{x}
=
\left(
\begin{array}{rr}
0 & 1 
\\
1 & 0 
\end{array}
\right)_m,
\quad
\sigma_m^{y}
=
\left(
\begin{array}{rr}
0 & -i 
\\
i & 0 
\end{array}
\right)_m,
\quad
\sigma_m^{z}
=
\left(
\begin{array}{rr}
1 & 0 
\\
0 & -1
\end{array}
\right)_m
\label{pauli}
\end{align}
with $i = \sqrt{-1}$, and the spin raising/lowering matrices
\begin{align}
\sigma_m^{+}
=
\frac{1}{2}
(\sigma_m^{x} + i \sigma_m^{y})
=
\left(
\begin{array}{cc}
0 & 1
\\
0 & 0 
\end{array}
\right)_m,
\quad
\sigma_m^{-}
=
\frac{1}{2}
(\sigma_m^{x}-i\sigma_m^{y})
=
\left(
\begin{array}{cc}
0 & 0
\\
1 & 0
\end{array}
\right)_m
\label{spin-chang} 
\end{align}
where in all cases the subscript $m$ is used to indicate that the matrices act in the space $V_m$.

\subsection{Hamiltonian $H$}
\label{s-ham}

The Hamiltonian of the finite length XXZ spin-$\frac{1}{2}$ chain is given by 
\begin{align}
H
=
\frac{1}{2}
\sum_{m=1}^{M} 
\Big(
\sigma_m^x \sigma_{m+1}^x
+
\sigma_m^y \sigma_{m+1}^y
+
\Delta (\sigma_m^z \sigma_{m+1}^z+1)
\Big),
\label{hamiltonian}
\end{align}
where $\Delta=\cosh\gamma$ is the anisotropy parameter of the model, and the periodicity conditions $\sigma_{M+1}^{x}=\sigma_1^{x},\sigma_{M+1}^{y}=\sigma_1^{y},\sigma_{M+1}^{z}=\sigma_1^{z}$ 
are assumed. In the coming subsections, we shall review the algebraic Bethe Ansatz procedure for finding the eigenvectors $|\Psi\rangle \in V$ of $H$.

\subsection{$R$-matrix, crossing symmetry and Yang-Baxter equation}

The $R$-matrix corresponding to the XXZ spin-$\frac{1}{2}$ chain is given by
\begin{align}
R_{ab}(u,v)
=
\left(
\begin{array}{cccc}
[u-v+\gamma] & 0                          & 0              & 0    \\
0                                        & [u-v]    & [\gamma]            & 0    \\
0                                        & [\gamma]             & [u-v]  & 0    \\
0                            & 0                          & 0              & [u-v+\gamma] 
\end{array}
\right)_{ab},
\label{Rmat1}
\end{align}
where we have defined $[u]=\sinh u$. The $R$-matrix is an element of ${\rm End}(V_a\otimes V_b)$, where $V_a,V_b$ are copies of $\mathbb{C}^2$. The variables $u,v$ are rapidities associated to the respective vector spaces $V_a,V_b$, and $\gamma$ is the crossing parameter. We identify the entries of (\ref{Rmat1}) with vertices, as shown in figure \ref{R6v}.

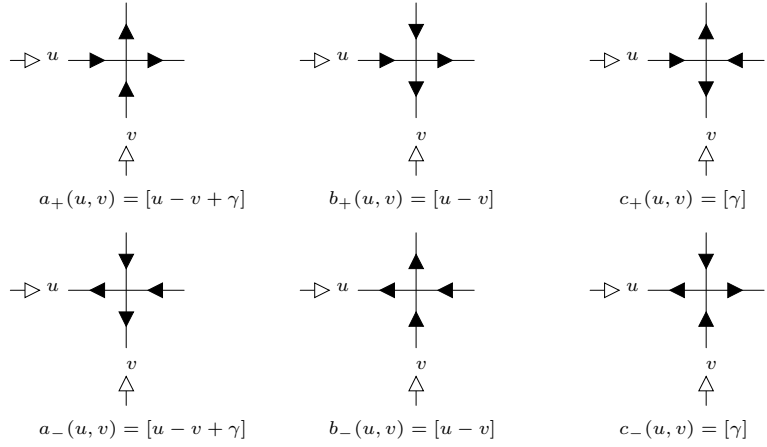
\begin{figure}[H]
\begin{center}
\begin{minipage}{4.3in}

\setlength{\unitlength}{0.00038cm}
\begin{picture}(20000,14000)(-4500,-13000)

%%%Vertex a_{+}%%%
\path(-2000,0000)(2000,0000)
\blacken\path(-1250,250)(-1250,-250)(-750,0)(-1250,250)
\blacken\path(750,250)(750,-250)(1250,0)(750,250)
%%%
\put(-2750,0){\scriptsize{$u$}}
\path(-4000,0)(-3000,0)
\whiten\path(-3500,250)(-3500,-250)(-3000,0)(-3500,250)
%%%
\path(0000,-2000)(0000,2000)
\blacken\path(-250,-1250)(250,-1250)(0,-750)(-250,-1250)
\blacken\path(-250,750)(250,750)(0,1250)(-250,750)
%%%
\put(-3000,-5000){\scriptsize$a_{+}(u,v)=[u-v+\gamma]$}
\put(0,-2750){\scriptsize{$v$}}
\path(0,-4000)(0,-3000)
\whiten\path(-250,-3500)(250,-3500)(0,-3000)(-250,-3500)

%%%Vertex b_{+}%%%
\path(8000,0000)(12000,0000)
\blacken\path(8750,250)(8750,-250)(9250,0)(8750,250)
\blacken\path(10750,250)(10750,-250)(11250,0)(10750,250)
%%%
\put(7250,0){\scriptsize{$u$}}
\path(6000,0)(7000,0)
\whiten\path(6500,250)(6500,-250)(7000,0)(6500,250)
%%%
\path(10000,-2000)(10000,2000)
\blacken\path(9750,-750)(10250,-750)(10000,-1250)(9750,-750)
\blacken\path(9750,1250)(10250,1250)(10000,750)(9750,1250)
%%%
\put(7000,-5000){\scriptsize$b_{+}(u,v)=[u-v]$}
\put(10000,-2750){\scriptsize{$v$}}
\path(10000,-4000)(10000,-3000)
\whiten\path(9750,-3500)(10250,-3500)(10000,-3000)(9750,-3500)

%%%Vertex c_{+}%%%
\path(18000,0000)(22000,0000)
\blacken\path(18750,250)(18750,-250)(19250,0)(18750,250)
\blacken\path(21250,250)(21250,-250)(20750,0)(21250,250)
%%%
\put(17250,0){\scriptsize{$u$}}
\path(16000,0)(17000,0)
\whiten\path(16500,250)(16500,-250)(17000,0)(16500,250)
%%%
\path(20000,-2000)(20000,2000)
\blacken\path(19750,-750)(20250,-750)(20000,-1250)(19750,-750)
\blacken\path(19750,750)(20250,750)(20000,1250)(19750,750)
%%%
\put(17000,-5000){\scriptsize$c_{+}(u,v)=[\gamma]$}
\put(20000,-2750){\scriptsize{$v$}}
\path(20000,-4000)(20000,-3000)
\whiten\path(19750,-3500)(20250,-3500)(20000,-3000)(19750,-3500)

%%%Vertex a_{-}%%%
\path(-2000,-8000)(2000,-8000)
\blacken\path(-750,-7750)(-750,-8250)(-1250,-8000)(-750,-7750)
\blacken\path(1250,-7750)(1250,-8250)(750,-8000)(1250,-7750)
%%%
\put(-2750,-8000){\scriptsize{$u$}}
\path(-4000,-8000)(-3000,-8000)
\whiten\path(-3500,-7750)(-3500,-8250)(-3000,-8000)(-3500,-7750)
%%%
\path(0000,-10000)(0000,-6000)
\blacken\path(-250,-8750)(250,-8750)(0,-9250)(-250,-8750)
\blacken\path(-250,-6750)(250,-6750)(0,-7250)(-250,-6750)
%%%
\put(-3000,-13000){\scriptsize$a_{-}(u,v)=[u-v+\gamma]$}
\put(0,-10750){\scriptsize{$v$}}
\path(0,-12000)(0,-11000)
\whiten\path(-250,-11500)(250,-11500)(0,-11000)(-250,-11500)

%%%Vertex b_{-}%%%
\path(8000,-8000)(12000,-8000)
\blacken\path(9250,-7750)(9250,-8250)(8750,-8000)(9250,-7750)
\blacken\path(11250,-7750)(11250,-8250)(10750,-8000)(11250,-7750)
%%%
\put(7250,-8000){\scriptsize{$u$}}
\path(6000,-8000)(7000,-8000)
\whiten\path(6500,-7750)(6500,-8250)(7000,-8000)(6500,-7750)
%%%
\path(10000,-10000)(10000,-6000)
\blacken\path(9750,-9250)(10250,-9250)(10000,-8750)(9750,-9250)
\blacken\path(9750,-7250)(10250,-7250)(10000,-6750)(9750,-7250)
%%%
\put(7000,-13000){\scriptsize$b_{-}(u,v)=[u-v]$}
\put(10000,-10750){\scriptsize{$v$}}
\path(10000,-12000)(10000,-11000)
\whiten\path(9750,-11500)(10250,-11500)(10000,-11000)(9750,-11500)

%%%Vertex c_{-}%%%
\path(18000,-8000)(22000,-8000)
\blacken\path(19250,-7750)(19250,-8250)(18750,-8000)(19250,-7750)
\blacken\path(20750,-7750)(20750,-8250)(21250,-8000)(20750,-7750)
%%%
\put(17250,-8000){\scriptsize{$u$}}
\path(16000,-8000)(17000,-8000)
\whiten\path(16500,-7750)(16500,-8250)(17000,-8000)(16500,-7750)
%%%
\path(20000,-10000)(20000,-6000)
\blacken\path(19750,-9250)(20250,-9250)(20000,-8750)(19750,-9250)
\blacken\path(19750,-6750)(20250,-6750)(20000,-7250)(19750,-6750)
%%%
\put(17000,-13000){\scriptsize$c_{-}(u,v)=[\gamma]$}
\put(20000,-10750){\scriptsize{$v$}}
\path(20000,-12000)(20000,-11000)
\whiten\path(19750,-11500)(20250,-11500)(20000,-11000)(19750,-11500)

\end{picture}

\end{minipage}
\end{center}

\caption[Six vertices associated to the XXZ $R$-matrix]{Six vertices associated to the XXZ $R$-matrix. Each entry of the $R$-matrix (\ref{Rmat1}) is matched with a vertex. All types of vertex not shown are by definition weighted to zero.} 

\label{R6v}
\end{figure} 

\begin{lemma} 
{\rm
Define $\bar{u} = u+\gamma$ for all rapidities $u$. The $R$-matrix has the {\it crossing symmetry} property
\begin{align}
R_{ab}(u,v)
=
-\sigma_b^y R_{ba}(v,\bar{u})^{{\rm t}_b} \sigma_b^y
\label{cross}
\end{align}
where $\sigma_b^y$ is the second of the Pauli matrices (\ref{pauli}) acting in $V_b$, and ${\rm t}_b$ denotes transposition in the space ${\rm End}(V_b)$. 
}
\end{lemma}

\begin{proof} 
We have
\begin{align}
R_{ba}(v,\bar{u})^{{\rm t}_b}
=
\left(
\begin{array}{cccc}
[v-\bar{u}+\gamma] & 0 & 0 & [\gamma]
\\
0 & [v-\bar{u}] & 0 & 0
\\
0 & 0 & [v-\bar{u}] & 0
\\ 
\phantom{.}[\gamma] & 0 & 0 & [v-\bar{u}+\gamma]
\end{array}
\right)_{ba},
\end{align}
which leads to the equation
\begin{align}
-\sigma_b^{y} R_{ba}(v,\bar{u})^{{\rm t}_b} \sigma_b^{y}
=
\left(
\begin{array}{cccc}
[\bar{u}-v] & 0 & 0 & 0
\\
0 & [\bar{u}-v-\gamma] & [\gamma] & 0
\\
0 & [\gamma] & [\bar{u}-v-\gamma] & 0
\\
0 & 0 & 0 & [\bar{u}-v]
\end{array}
\right)_{ba}.
\end{align}
Finally, using the definition $\bar{u} = u+\gamma$ and the fact that $R_{ba}(u,v) = R_{ab}(u,v)$, we prove (\ref{cross}).

\end{proof}

\begin{lemma}
{\rm The $R$-matrix obeys the {\it Yang-Baxter equation} 
\begin{align}
R_{ab}(u,v) R_{ac}(u,w) R_{bc}(v,w)
=
R_{bc}(v,w) R_{ac}(u,w) R_{ab}(u,v),
\label{YB}
\end{align}
which holds in ${\rm End}(V_a \otimes V_b \otimes V_c)$ for all $u,v,w$.
}
\end{lemma}

\subsection{$L$-matrix and local intertwining equation}

The $L$-matrix for the XXZ model depends on a single indeterminate $u$, and acts in the space $V_a$. Its entries are operators acting at the $m^{\rm th}$ lattice site, and identically everywhere else. It has the form
\begin{align}
L_{am}(u)
=
\left(
\begin{array}{cc}
[u+\frac{\gamma}{2}\sigma_m^{z}]
&
[\gamma] \sigma_m^{-}
\\
\phantom{.}
[\gamma] \sigma_m^{+}
&
[u-\frac{\gamma}{2} \sigma_m^{z}]
\end{array}
\right)_a,
\label{xxz-Lmat}
\end{align}
where we have defined, as before, $[u] = \sinh u$. 

\begin{lemma}
{\rm
Using the definition of the $R$-matrix (\ref{Rmat1}) and the $L$-matrix (\ref{xxz-Lmat}), the local intertwining equation is given by
\begin{align}
R_{ab}(u,v) L_{am}(u) L_{bm}(v)
=
L_{bm}(v) L_{am}(u) R_{ab}(u,v).
\label{xxz-Lint}
\end{align}
}
\end{lemma}

\begin{proof}

Using the matrix representations of $\{\sigma_m^{+},\sigma_m^{-},\sigma_m^{z}\}$ given by equations (\ref{pauli}) and (\ref{spin-chang}), we find that the $L$-matrix (\ref{xxz-Lmat}) takes the form
\begin{align}
L_{am}(u)
=
\left(
\begin{array}{cccc}
[u+\frac{\gamma}{2}] & 0 & 0 & 0
\\
0 & [u-\frac{\gamma}{2}] & [\gamma] & 0
\\
0 & [\gamma] & [u-\frac{\gamma}{2}] & 0
\\
0 & 0 & 0 & [u+\frac{\gamma}{2}]
\end{array}
\right)_{am}
=
R_{am}(u,\gamma/2),
\end{align}
from which we see that it is equal to the $R$-matrix $R_{am}(u,w_m)$ with $w_m = \frac{\gamma}{2}$. The local intertwining equation (\ref{xxz-Lint}) becomes
\begin{align}
R_{ab}(u,v) R_{am}(u,\gamma/2) R_{bm}(v,\gamma/2)
=
R_{bm}(v,\gamma/2) R_{am}(u,\gamma/2) R_{ab}(u,v),
\end{align}
which is simply a corollary of the Yang-Baxter equation (\ref{YB}).

\end{proof}

\subsection{Monodromy matrix and global intertwining equation}
\label{s-xxzmon}

The monodromy matrix is an ordered product of $L$-matrices, taken across all sites of the chain. It is given by
\begin{align}
T_a(u)
=
L_{a1}(u)\ldots L_{aM}(u).
\end{align}
In the last subsection we observed that the $L$-matrix $L_{am}(u)$ is equal to the $R$-matrix $R_{am}(u,w_m)$ under the specialization $w_m =\frac{\gamma}{2}$. Using this observation, it is convenient to construct an inhomogeneous monodromy matrix as an ordered product of the $R$-matrices $R_{am}(u,w_m)$, without restricting the variables $w_m$. That is, we define
\begin{align}
T_{a}(u,\{w\}_M)
=
R_{a1}(u,w_1) \ldots R_{aM}(u,w_M).
\label{Tmat}
\end{align}
The variables $\{w_1,\ldots,w_M\}$ are called {\it inhomogeneities} and the usual monodromy matrix is recovered by setting $w_m = \frac{\gamma}{2}$ for all $1\leq m \leq M$. It turns out that the inclusion of the variables $\{w_1,\ldots,w_M\}$ simplifies many later calculations. The contribution from the space ${\rm End}(V_a)$ can be exhibited explicitly by defining
\begin{align}
T_{a}(u,\{w\}_M)
=
\left(
\begin{array}{cc}
A(u,\{w\}_M) & B(u,\{w\}_M) \\
C(u,\{w\}_M) & D(u,\{w\}_M)
\end{array}
\right)_a,
\label{Tmat2}
\end{align}
where the matrix entries are all operators acting in $V=V_1\otimes\cdots\otimes V_M$. The diagrammatic version of these operators is given in figure \ref{6vT}.

\begin{figure}[H]

\begin{center}
\begin{minipage}{4.3in}

\setlength{\unitlength}{0.0004cm}
\begin{picture}(20000,15000)(-2500,-13000)

%%String 11%%%
\path(-2000,0000)(10000,0000)
\blacken\path(-1250,250)(-1250,-250)(-750,0)(-1250,250)
\blacken\path(8750,250)(8750,-250)(9250,0)(8750,250)
\put(-2750,0){$u$}
\path(-4000,0)(-3000,0)
\whiten\path(-3500,250)(-3500,-250)(-3000,0)(-3500,250)
\path(0000,-2000)(0000,2000)
\put(-400,-2700){\scriptsize$w_1$}
\path(0,-4000)(0,-3000)
\whiten\path(-250,-3500)(250,-3500)(0,-3000)(-250,-3500)
\path(2000,-2000)(2000,2000)
\path(4000,-2000)(4000,2000)
\path(6000,-2000)(6000,2000)
\path(8000,-2000)(8000,2000)
\put(7800,-2700){\scriptsize$w_M$}
\path(8000,-4000)(8000,-3000)
\whiten\path(7750,-3500)(8250,-3500)(8000,-3000)(7750,-3500)
\put(1500,-4500)
{$A(u,\{w\}_M)$}
%%%%%%%%%%%%%%%%

%%String 12%%%
\path(14000,0000)(26000,0000)
\blacken\path(14750,250)(14750,-250)(15250,0)(14750,250)
\blacken\path(25250,250)(25250,-250)(24750,0)(25250,250)
\put(13250,0){$u$}
\path(12000,0)(13000,0)
\whiten\path(12500,250)(12500,-250)(13000,0)(12500,250)
\path(16000,-2000)(16000,2000)
\put(15600,-2700){\scriptsize$w_1$}
\path(16000,-4000)(16000,-3000)
\whiten\path(15750,-3500)(16250,-3500)(16000,-3000)(15750,-3500)
\path(18000,-2000)(18000,2000)
\path(20000,-2000)(20000,2000)
\path(22000,-2000)(22000,2000)
\path(24000,-2000)(24000,2000)
\put(23800,-2700){\scriptsize$w_M$}
\path(24000,-4000)(24000,-3000)
\whiten\path(23750,-3500)(24250,-3500)(24000,-3000)(23750,-3500)
\put(17500,-4500)
{$B(u,\{w\}_M)$}
%%%%%%%%%

%%String 21%%%
\path(-2000,-8000)(10000,-8000)
\blacken\path(-750,-7750)(-750,-8250)(-1250,-8000)(-750,-7750)
\blacken\path(8750,-7750)(8750,-8250)(9250,-8000)(8750,-7750)
\put(-2750,-8000){$u$}
\path(-4000,-8000)(-3000,-8000)
\whiten\path(-3500,-7750)(-3500,-8250)(-3000,-8000)(-3500,-7750)
\path(0000,-10000)(0000,-6000)
\put(-400,-10700){\scriptsize$w_1$}
\path(0,-12000)(0,-11000)
\whiten\path(-250,-11500)(250,-11500)(0,-11000)(-250,-11500)
\path(2000,-10000)(2000,-6000)
\path(4000,-10000)(4000,-6000)
\path(6000,-10000)(6000,-6000)
\path(8000,-10000)(8000,-6000)
\put(7800,-10700){\scriptsize$w_M$}
\path(8000,-12000)(8000,-11000)
\whiten\path(7750,-11500)(8250,-11500)(8000,-11000)(7750,-11500)
\put(1500,-12500)
{$C(u,\{w\}_M)$}
%%%%%%%%%

%%String 22%%%
\path(14000,-8000)(26000,-8000)
\blacken\path(15250,-7750)(15250,-8250)(14750,-8000)(15250,-7750)
\blacken\path(25250,-7750)(25250,-8250)(24750,-8000)(25250,-7750)
\put(13250,-8000){$u$}
\path(12000,-8000)(13000,-8000)
\whiten\path(12500,-7750)(12500,-8250)(13000,-8000)(12500,-7750)
\path(16000,-10000)(16000,-6000)
\put(15600,-10700){\scriptsize$w_1$}
\path(16000,-12000)(16000,-11000)
\whiten\path(15750,-11500)(16250,-11500)(16000,-11000)(15750,-11500)
\path(18000,-10000)(18000,-6000)
\path(20000,-10000)(20000,-6000)
\path(22000,-10000)(22000,-6000)
\path(24000,-10000)(24000,-6000)
\put(23800,-10700){\scriptsize$w_M$}
\path(24000,-12000)(24000,-11000)
\whiten\path(23750,-11500)(24250,-11500)(24000,-11000)(23750,-11500)
\put(17500,-12500)
{$D(u,\{w\}_M)$}

\end{picture}

\end{minipage}

\end{center}

\caption[Four vertex-strings of the XXZ monodromy matrix]{Four vertex-strings of the XXZ monodromy matrix. Each entry of (\ref{Tmat2}) is matched with a string of $R$-matrix vertices.}
\label{6vT}

\end{figure}
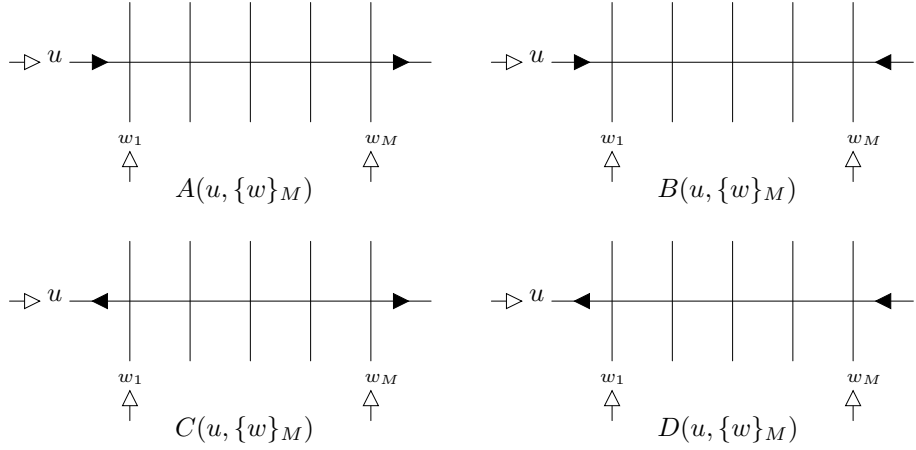 

\noindent Owing to the Yang-Baxter equation (\ref{YB}) and the definition of the monodromy matrix (\ref{Tmat}), we obtain the global intertwining equation 
\begin{align}
R_{ab}(u,v) T_{a}(u,\{w\}_M) T_{b}(v,\{w\}_M)
=
T_{b}(v,\{w\}_M) T_{a}(u,\{w\}_M) R_{ab}(u,v)
\label{int1}
\end{align}
This equation contains sixteen commutation relations between the entries of the monodromy matrix. Particular examples of these commutation relations include
\begin{align}
&B(u) B(v) = B(v) B(u), \label{BB} \\ 
&[u-v+\gamma] B(u) A(v) =  [\gamma] B(v) A(u) + [u-v] A(v) B(u), \label{AB} \\
&[\gamma] B(u) D(v) + [u-v] D(u) B(v) = [u-v+\gamma] B(v) D(u), \label{DB} \\ 
&C(u) C(v) = C(v) C(u), \label{CC} \\
&[\gamma] A(u) C(v) + [u-v] C(u) A(v) = [u-v+\gamma] A(v) C(u), \label{CA} \\ 
&[u-v+\gamma] D(u) C(v) = [\gamma] D(v) C(u) + [u-v] C(v) D(u), \label{CD}
\\
& D(u) D(v) = D(v) D(u), \label{DD} 
\end{align}
which are obtained by multiplying the matrices in (\ref{int1}) and equating components.

\begin{lemma}
{\rm 
Let $\{u\}_{L} = \{u_1,\ldots,u_L\}$ and $\{w\}_M = \{w_1,\ldots,w_M\}$ be two sets of variables, with cardinalities $L,M \geq 1$. To each variable $u_l$ we associate an auxiliary space $V_{a_l}$, while to each variable $w_m$ we associate a quantum space $V_{m}$. Defining
\begin{align}
T\Big(\{u\}_{L},\{w\}_{M}\Big)
=
T_{a_L}(u_L,\{w\}_M)
\ldots
T_{a_1}(u_1,\{w\}_M),
\label{doublemon}
\end{align}
we claim that
\begin{align}
T\Big(\{u\}_{L},\{w\}_{M}\Big)
=
(-)^{LM}
\overline{T}_1(w_1,\{\bar{u}\}_L)
\ldots
\overline{T}_M(w_M,\{\bar{u}\}_L),
\label{doublemon2}
\end{align}
where for all $1 \leq m \leq M$ we have defined
\begin{align}
\overline{T}_m(w_m,\{\bar{u}\}_L)
=
\left(
\begin{array}{rr}
D(w_m,\{\bar{u}\}_L) & -B(w_m,\{\bar{u}\}_L)
\\
-C(w_m,\{\bar{u}\}_L) & A(w_m,\{\bar{u}\}_L)
\end{array}
\right)_m
\label{doublemon!}
\end{align}
with $\{\bar{u}\}_L = \{u_1+\gamma,\ldots,u_L+\gamma\}$.

}
\end{lemma}

\begin{proof}
 
By the definition (\ref{doublemon}) we have
\begin{align}
T\Big(\{u\}_{L},\{w\}_M\Big)
=
\Big(
R_{a_L 1}(u_L,w_1)
\ldots
R_{a_L M}(u_L,w_M)
\Big)
\ldots
\Big(
R_{a_1 1}(u_1,w_1)
\ldots
R_{a_1 M}(u_1,w_M)
\Big).
\end{align}
Commuting $R$-matrices which act in different spaces leads to the equation
\begin{align}
T\Big(\{u\}_{L},\{w\}_M\Big)
=
\label{every}
\Big(
R_{a_L 1}(u_L,w_1)
\ldots
R_{a_1 1}(u_1,w_1)
\Big)
\ldots
\Big(
R_{a_L M}(u_L,w_M)
\ldots
R_{a_1 M}(u_1,w_M)
\Big),
\end{align}
and using the crossing symmetry relation (\ref{cross}) on every $R$-matrix in (\ref{every}) we obtain
\begin{align}
T\Big(\{u\}_{L},\{w\}_M\Big)
&=
(-)^{LM}
\sigma_1^{y}
\Big(
R_{1 a_L}(w_1,\bar{u}_L)^{{\rm t}_1}
\ldots
R_{1 a_1}(w_1,\bar{u}_1)^{{\rm t}_1}
\Big)
\sigma_1^{y}
\label{every2}
\\
&
\times
\ldots
\times 
\sigma_M^{y}
\Big(
R_{M a_L}(w_M,\bar{u}_L)^{{\rm t}_M}
\ldots
R_{M a_1}(w_M,\bar{u}_1)^{{\rm t}_M}
\Big)
\sigma_M^{y}.
\nonumber
\end{align}
From the standard identity of matrix transposition $(AB)^{t} = B^t A^t$ we may reverse the order of the $R$-matrices in (\ref{every2}), yielding
\begin{align}
T\Big(\{u\}_{L},\{w\}_M\Big)
&=
(-)^{LM}
\sigma_1^{y}
\Big(
R_{1 a_1}(w_1,\bar{u}_1)
\ldots
R_{1 a_L}(w_1,\bar{u}_L)
\Big)^{{\rm t}_1}
\sigma_1^{y}
\label{every3}
\\
&
\times
\ldots
\times
\sigma_M^{y}
\Big(
R_{M a_1}(w_M,\bar{u}_1)
\ldots
R_{M a_L}(w_M,\bar{u}_L)
\Big)^{{\rm t}_M}
\sigma_M^{y}.
\nonumber
\end{align}
Finally we replace each parenthesized term in (\ref{every3}) with its corresponding monodromy matrix, which gives 
\begin{align}
T\Big(\{u\}_{L},\{w\}_M\Big)
=
(-)^{LM}
\sigma_1^{y}
T_1(w_1,\{\bar{u}\}_L)^{{\rm t}_1}
\sigma_1^{y}
\ldots
\sigma_M^{y}
T_M(w_M,\{\bar{u}\}_L)^{{\rm t}_M}
\sigma_M^{y}.
\label{every4}
\end{align}
Letting the monodromy matrices in (\ref{every4}) be written in the form
\begin{align}
T_m(w_m,\{\bar{u}\}_L)
=
\left(
\begin{array}{cc}
A(w_m,\{\bar{u}\}_L) & B(w_m,\{\bar{u}\}_L)
\\
C(w_m,\{\bar{u}\}_L) & D(w_m,\{\bar{u}\}_L)
\end{array}
\right)_m
\end{align}
for all $1 \leq m \leq M$ and contracting on the quantum spaces $V_1,\ldots,V_M$, we recover the result (\ref{doublemon2}).

\end{proof}

\begin{lemma} 
{\rm 
The spin-up states $|\Uparrow_M\rangle$ and $\langle \Uparrow_M|$ are eigenvectors of the diagonal elements of the monodromy matrix:
\begin{align}
A(u,\{w\}_M)|\Uparrow_M\rangle
&=
\prod_{j=1}^{M} [u-w_j+\gamma] |\Uparrow_M\rangle,
\quad
D(u,\{w\}_M)|\Uparrow_M\rangle
=
\prod_{j=1}^{M} [u-w_j] |\Uparrow_M\rangle,
\label{diagonal1}
\\
\langle \Uparrow_M | A(u,\{w\}_M)
&=
\prod_{j=1}^{M} [u-w_j+\gamma] \langle \Uparrow_M|,
\quad
\langle \Uparrow_M | D(u,\{w\}_M)
=
\prod_{j=1}^{M} [u-w_j] \langle \Uparrow_M |. 
\label{diagonal1*}
\end{align}
In addition, the spin-down states $|\Downarrow_M\rangle$ and $\langle \Downarrow_M|$ are eigenvectors of the diagonal elements of the monodromy matrix:
\begin{align}
A(u,\{w\}_M)|\Downarrow_M\rangle
&=
\prod_{j=1}^{M} [u-w_j] |\Downarrow_M\rangle,
\quad
D(u,\{w\}_M)|\Downarrow_M\rangle
=
\prod_{j=1}^{M} [u-w_j+\gamma] |\Downarrow_M\rangle,
\label{diagonal2}
\\
\langle \Downarrow_M | A(u,\{w\}_M)
&=
\prod_{j=1}^{M} [u-w_j] \langle \Downarrow_M|,
\quad
\langle \Downarrow_M | D(u,\{w\}_M)
=
\prod_{j=1}^{M} [u-w_j+\gamma] \langle \Downarrow_M |. 
\label{diagonal2*}
\end{align}

}
\end{lemma}

\subsection{Recovering $H$ from the transfer matrix}

Let 
\begin{align}
t(u,\{w\}_M) = A(u,\{w\}_M)+D(u,\{w\}_M)
\end{align}
denote the transfer matrix of the XXZ model. The Hamiltonian (\ref{hamiltonian}) is recovered via the quantum trace identity
\begin{align}
H
=
[\gamma]
\frac{d}{du}
\log t(u)
\Big|_{u = \frac{\gamma}{2}},
\quad
{\rm where}\quad
t(u)
=
t(u,\{w\}_M)
\Big|_{w_1=\cdots = w_M =\frac{\gamma}{2}}.
\end{align}
Therefore all eigenvectors of $t(u,\{w\}_M)$ are also eigenvectors of $H$. The attention turns, therefore, to finding vectors $|\Psi\rangle \in V$ satisfying
\begin{align}
\Big(
A(u,\{w\}_M)+D(u,\{w\}_M)
\Big)
|\Psi\rangle = 
\tau_{\Psi}(u,\{w\}_M) |\Psi\rangle
\label{eigenvector}
\end{align}
for some suitable constants $\tau_{\Psi}(u,\{w\}_M)$.

\subsection{Bethe Ansatz for the eigenvectors}
\label{s-betheans}

The eigenvectors of $t(u,\{w\}_M)$ are given by the Ansatz 
\begin{align}
|\Psi\rangle = B(v_1,\{w\}_M)\ldots B(v_N,\{w\}_M)|\Uparrow_M\rangle,
\label{bethe1}
\end{align}
where we assume that $N \leq M$, since we annihilate the state $|\Uparrow_M\rangle$ when acting with more $B$-operators than the number of sites in the spin chain. Similarly, we construct eigenvectors of $t(u,\{w\}_M)$ in the dual space of states via the Ansatz 
\begin{align}
\langle \Psi|
=
\langle \Uparrow_M|
C(v_N,\{w\}_M)
\ldots
C(v_1,\{w\}_M),
\label{bethe3}
\end{align}
where we again restrict $N \leq M$. To ensure that (\ref{bethe1}) and (\ref{bethe3}) are genuine eigenvectors, the variables $\{v_1,\ldots,v_N\}$ are required to satisfy the Bethe equations. Using the commutation relations (\ref{BB})--(\ref{CD}) and the actions (\ref{diagonal1}), (\ref{diagonal1*}) it is possible to show that $|\Psi\rangle,\langle \Psi|$ are eigenvectors of $t(u,\{w\}_M)$ if and only if 
\begin{align}
\prod_{j=1}^{M}\frac{ [v_i-w_j+\gamma]}{[v_i-w_j]}
=
\prod_{j \not= i}^{N}
\frac{[v_i-v_j + \gamma]}{[v_i - v_j - \gamma]}
\label{bethe2}
\end{align}
for all $1\leq i \leq N$. Throughout the rest of this work, we will call the expression (\ref{bethe1}) a {\it Bethe eigenvector} and assume implicitly that the Bethe equations (\ref{bethe2}) hold for the parameters $\{v\}_N$. 

\section{Domain wall partition function $Z_N\Big(\{v\}_N,\{w\}_N\Big)$}
\label{xxz-pf}

In this section we study $Z_N$, the domain wall partition function of the XXZ spin-$\frac{1}{2}$ chain. This quantity acquires its name because it is equal to the partition function of the six-vertex model under domain wall boundary conditions \cite{bax}. The calculation of $Z_N$ is essential for the explicit evaluation of more complicated objects within the XXZ model, such as its scalar product.

\subsection{Definition of $Z_N(\{v\}_N,\{w\}_N)$}

Let $\{v\}_N = \{v_1,\ldots,v_N\}$ and $\{w\}_N = \{w_1,\ldots,w_N\}$ be two sets of variables. The domain wall partition function has the algebraic definition
\begin{align}
Z_N\Big(\{v\}_N,\{w\}_N\Big)
=
\langle \Downarrow_N | \prod_{i=1}^{N} B(v_i,\{w\}_N) |\Uparrow_N \rangle,
\label{pf-xxz}
\end{align}
where the ordering of the $B$-operators is irrelevant, since by (\ref{BB}) they commute.  

\subsection{Graphical representation of partition function}

Using the graphical conventions described in \S \ref{s-xxzmon}, the domain wall partition function may be represented as the $N\times N$ lattice shown in figure \ref{pf6v}.

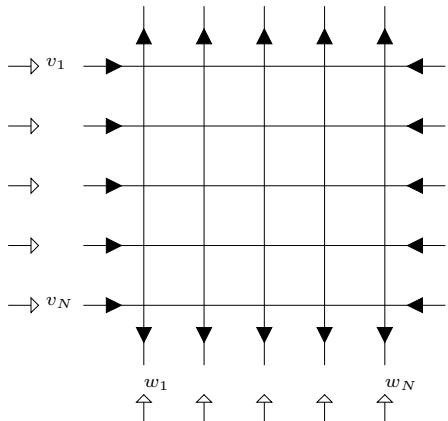
\begin{figure}[H]

\begin{center}
\begin{minipage}{4.3in}

\setlength{\unitlength}{0.0004cm}
\begin{picture}(20000,12500)(-10000,-3000)

\path(-2000,0)(10000,0)
\put(-3250,0){\scriptsize$v_N$}
\path(-4500,0)(-3500,0)
\whiten\path(-3750,250)(-3750,-250)(-3500,0)(-3750,250)
\blacken\path(-1250,250)(-1250,-250)(-750,0)(-1250,250)
\blacken\path(9250,250)(9250,-250)(8750,0)(9250,250)

\path(-2000,2000)(10000,2000)
\path(-4500,2000)(-3500,2000)
\whiten\path(-3750,2250)(-3750,1750)(-3500,2000)(-3750,2250)
\blacken\path(-1250,2250)(-1250,1750)(-750,2000)(-1250,2250)
\blacken\path(9250,2250)(9250,1750)(8750,2000)(9250,2250)

\path(-2000,4000)(10000,4000)
\path(-4500,4000)(-3500,4000)
\whiten\path(-3750,4250)(-3750,3750)(-3500,4000)(-3750,4250)
\blacken\path(-1250,4250)(-1250,3750)(-750,4000)(-1250,4250)
\blacken\path(9250,4250)(9250,3750)(8750,4000)(9250,4250)

\path(-2000,6000)(10000,6000)
\path(-4500,6000)(-3500,6000)
\whiten\path(-3750,6250)(-3750,5750)(-3500,6000)(-3750,6250)
\blacken\path(-1250,6250)(-1250,5750)(-750,6000)(-1250,6250)
\blacken\path(9250,6250)(9250,5750)(8750,6000)(9250,6250)

\path(-2000,8000)(10000,8000)
\put(-3250,8000){\scriptsize$v_1$}
\path(-4500,8000)(-3500,8000)
\whiten\path(-3750,8250)(-3750,7750)(-3500,8000)(-3750,8250)
\blacken\path(-1250,8250)(-1250,7750)(-750,8000)(-1250,8250)
\blacken\path(9250,8250)(9250,7750)(8750,8000)(9250,8250)

%%%%%%%%%%%

\path(0,-2000)(0,10000)
\put(0,-2750){\scriptsize$w_1$}
\path(0,-4000)(0,-3000)
\whiten\path(-250,-3250)(250,-3250)(0,-3000)(-250,-3250)
\blacken\path(-250,-750)(250,-750)(0,-1250)(-250,-750)
\blacken\path(-250,8750)(250,8750)(0,9250)(-250,8750)

\path(2000,-2000)(2000,10000)
\path(2000,-4000)(2000,-3000)
\whiten\path(1750,-3250)(2250,-3250)(2000,-3000)(1750,-3250)
\blacken\path(1750,-750)(2250,-750)(2000,-1250)(1750,-750)
\blacken\path(1750,8750)(2250,8750)(2000,9250)(1750,8750)

\path(4000,-2000)(4000,10000)
\path(4000,-4000)(4000,-3000)
\whiten\path(3750,-3250)(4250,-3250)(4000,-3000)(3750,-3250)
\blacken\path(3750,-750)(4250,-750)(4000,-1250)(3750,-750)
\blacken\path(3750,8750)(4250,8750)(4000,9250)(3750,8750)

\path(6000,-2000)(6000,10000)
\path(6000,-4000)(6000,-3000)
\whiten\path(5750,-3250)(6250,-3250)(6000,-3000)(5750,-3250)
\blacken\path(5750,-750)(6250,-750)(6000,-1250)(5750,-750)
\blacken\path(5750,8750)(6250,8750)(6000,9250)(5750,8750)

\path(8000,-2000)(8000,10000)
\put(8000,-2750){\scriptsize$w_N$}
\path(8000,-4000)(8000,-3000)
\whiten\path(7750,-3250)(8250,-3250)(8000,-3000)(7750,-3250)
\blacken\path(7750,-750)(8250,-750)(8000,-1250)(7750,-750)
\blacken\path(7750,8750)(8250,8750)(8000,9250)(7750,8750)

\end{picture}

\end{minipage}
\end{center}

\caption[Domain wall partition function of the six-vertex model]{Domain wall partition function of the six-vertex model. The top row of upward pointing arrows corresponds with the state vector $|\Uparrow_N\rangle$. The bottom row of downward pointing arrows corresponds with the dual state vector $\langle \Downarrow_N|$. Each horizontal lattice line corresponds to multiplication by a $B$-operator.}

\label{pf6v}
\end{figure}

\subsection{Conditions on $Z_N(\{v\}_N,\{w\}_N)$}
\label{s-pfcond}

In \cite{kor}, Korepin showed that the domain wall partition function satisfies a set of four conditions which determine it uniquely. We reproduce these facts below, with the following two lemmas.

\begin{lemma}
{\rm Let us adopt the shorthand $Z_N = Z_N(\{v\}_N,\{w\}_N)$. For all $N \geq 2$ we claim that 

\setcounter{conditions}{0}
\begin{conditions}
{\rm
$Z_N$ is symmetric in the $\{ w \}_N$ variables.
}
\end{conditions}

\begin{conditions}
{\rm
$Z_N$ is a trigonometric polynomial of degree $N-1$ in the rapidity variable $v_N$.
}
\end{conditions}

\begin{conditions}
{\rm
Setting $v_N = w_N-\gamma$, $Z_{N}$ satisfies the recursion relation 
\begin{align}
Z_N\Big|_{v_N=w_N-\gamma}
&=
[\gamma]
\prod_{i=1}^{N-1}
[v_i-w_N][w_N-w_i-\gamma]
Z_{N-1},
\label{recursion}
\end{align}
where $Z_{N-1}$ is the domain wall partition function on a square lattice of size $N-1$.
}
\end{conditions}

In addition, we have the supplementary condition 

\begin{conditions}
{\rm The partition function on the $1\times 1$ lattice is given by $Z_{1} = [\gamma]$.
}
\end{conditions}
}
\end{lemma}

\begin{proof}

\setcounter{conditions}{0}
\begin{conditions}
{\rm 
We write the domain wall partition function in the form
\begin{align}
Z_N\Big(\{v\}_N,\{w\}_N\Big)
=
\langle \Uparrow_N^a|
\otimes
\langle \Downarrow_N|
T\Big(
\{v\}_N,\{w\}_N
\Big)
|\Uparrow_N\rangle
\otimes
|\Downarrow_N^a\rangle
\end{align}
with $T(\{v\}_N,\{w\}_N)$ given by (\ref{doublemon}), and where we have defined the auxiliary states
\begin{align}
\langle \Uparrow_N^a|
=
\bigotimes_{i=1}^{N} 
\uparrow_{a_i}^{*},
\quad
|\Downarrow_N^a\rangle
=
\bigotimes_{i=1}^{N}
\downarrow_{a_i}.
\end{align}
Using the expressions (\ref{doublemon2}), (\ref{doublemon!}) for $T(\{v\}_N,\{w\}_N)$ and contracting on the quantum spaces $V_1,\ldots,V_N$ gives
\begin{align}
Z_N\Big(
\{v\}_N,\{w\}_N
\Big)
=
\langle \Uparrow_N^a|
C(w_1,\{\bar{v}\}_N)
\ldots
C(w_N,\{\bar{v}\}_N)
|\Downarrow_N^a\rangle.
\label{zequiv}
\end{align}
The diagrammatic interpretation of the equivalence (\ref{zequiv}) is shown in figure \ref{zequiv2}.

%%%%%%%%

\begin{figure}[H]

\begin{center}
\begin{minipage}{4.3in}

\setlength{\unitlength}{0.000325cm}
\begin{picture}(20000,12500)(-4000,-3000)

\path(-2000,0)(10000,0)
\put(-3250,0){\scriptsize$v_N$}
\path(-4500,0)(-3500,0)
\whiten\path(-3750,250)(-3750,-250)(-3500,0)(-3750,250)
\blacken\path(-1250,250)(-1250,-250)(-750,0)(-1250,250)
\blacken\path(9250,250)(9250,-250)(8750,0)(9250,250)

\path(-2000,2000)(10000,2000)
\path(-4500,2000)(-3500,2000)
\whiten\path(-3750,2250)(-3750,1750)(-3500,2000)(-3750,2250)
\blacken\path(-1250,2250)(-1250,1750)(-750,2000)(-1250,2250)
\blacken\path(9250,2250)(9250,1750)(8750,2000)(9250,2250)

\path(-2000,4000)(10000,4000)
\path(-4500,4000)(-3500,4000)
\whiten\path(-3750,4250)(-3750,3750)(-3500,4000)(-3750,4250)
\blacken\path(-1250,4250)(-1250,3750)(-750,4000)(-1250,4250)
\blacken\path(9250,4250)(9250,3750)(8750,4000)(9250,4250)

\path(-2000,6000)(10000,6000)
\path(-4500,6000)(-3500,6000)
\whiten\path(-3750,6250)(-3750,5750)(-3500,6000)(-3750,6250)
\blacken\path(-1250,6250)(-1250,5750)(-750,6000)(-1250,6250)
\blacken\path(9250,6250)(9250,5750)(8750,6000)(9250,6250)

\path(-2000,8000)(10000,8000)
\put(-3250,8000){\scriptsize$v_1$}
\path(-4500,8000)(-3500,8000)
\whiten\path(-3750,8250)(-3750,7750)(-3500,8000)(-3750,8250)
\blacken\path(-1250,8250)(-1250,7750)(-750,8000)(-1250,8250)
\blacken\path(9250,8250)(9250,7750)(8750,8000)(9250,8250)

%%%%%%%%%%%

\path(0,-2000)(0,10000)
\put(0,-2750){\scriptsize$w_1$}
\path(0,-4000)(0,-3000)
\whiten\path(-250,-3250)(250,-3250)(0,-3000)(-250,-3250)
\blacken\path(-250,-750)(250,-750)(0,-1250)(-250,-750)
\blacken\path(-250,8750)(250,8750)(0,9250)(-250,8750)

\path(2000,-2000)(2000,10000)
\path(2000,-4000)(2000,-3000)
\whiten\path(1750,-3250)(2250,-3250)(2000,-3000)(1750,-3250)
\blacken\path(1750,-750)(2250,-750)(2000,-1250)(1750,-750)
\blacken\path(1750,8750)(2250,8750)(2000,9250)(1750,8750)

\path(4000,-2000)(4000,10000)
\path(4000,-4000)(4000,-3000)
\whiten\path(3750,-3250)(4250,-3250)(4000,-3000)(3750,-3250)
\blacken\path(3750,-750)(4250,-750)(4000,-1250)(3750,-750)
\blacken\path(3750,8750)(4250,8750)(4000,9250)(3750,8750)

\path(6000,-2000)(6000,10000)
\path(6000,-4000)(6000,-3000)
\whiten\path(5750,-3250)(6250,-3250)(6000,-3000)(5750,-3250)
\blacken\path(5750,-750)(6250,-750)(6000,-1250)(5750,-750)
\blacken\path(5750,8750)(6250,8750)(6000,9250)(5750,8750)

\path(8000,-2000)(8000,10000)
\put(8000,-2750){\scriptsize$w_N$}
\path(8000,-4000)(8000,-3000)
\whiten\path(7750,-3250)(8250,-3250)(8000,-3000)(7750,-3250)
\blacken\path(7750,-750)(8250,-750)(8000,-1250)(7750,-750)
\blacken\path(7750,8750)(8250,8750)(8000,9250)(7750,8750)

%%%%%%%

\put(12000,3750){$=$}

%%%%%%%2nd lattice %%%%%%

\path(18000,0)(30000,0)
\put(16250,0){\scriptsize$w_1$}
\path(15000,0)(16000,0)
\whiten\path(15750,250)(15750,-250)(16000,0)(15750,250)
\blacken\path(19250,250)(19250,-250)(18750,0)(19250,250)
\blacken\path(28750,250)(28750,-250)(29250,0)(28750,250)

\path(18000,2000)(30000,2000)
\path(15000,2000)(16000,2000)
\whiten\path(15750,2250)(15750,1750)(16000,2000)(15750,2250)
\blacken\path(19250,2250)(19250,1750)(18750,2000)(19250,2250)
\blacken\path(28750,2250)(28750,1750)(29250,2000)(28750,2250)

\path(18000,4000)(30000,4000)
\path(15000,4000)(16000,4000)
\whiten\path(15750,4250)(15750,3750)(16000,4000)(15750,4250)
\blacken\path(19250,4250)(19250,3750)(18750,4000)(19250,4250)
\blacken\path(28750,4250)(28750,3750)(29250,4000)(28750,4250)

\path(18000,6000)(30000,6000)
\path(15000,6000)(16000,6000)
\whiten\path(15750,6250)(15750,5750)(16000,6000)(15750,6250)
\blacken\path(19250,6250)(19250,5750)(18750,6000)(19250,6250)
\blacken\path(28750,6250)(28750,5750)(29250,6000)(28750,6250)

\path(18000,8000)(30000,8000)
\put(16250,8000){\scriptsize$w_N$}
\path(15000,8000)(16000,8000)
\whiten\path(15750,8250)(15750,7750)(16000,8000)(15750,8250)
\blacken\path(19250,8250)(19250,7750)(18750,8000)(19250,8250)
\blacken\path(28750,8250)(28750,7750)(29250,8000)(28750,8250)

%%%%%%%%%%%

\path(20000,-2000)(20000,10000)
\put(20000,-2750){\scriptsize$\bar{v}_1$}
\path(20000,-4000)(20000,-3000)
\whiten\path(19750,-3250)(20250,-3250)(20000,-3000)(19750,-3250)
\blacken\path(19750,-1250)(20250,-1250)(20000,-750)(19750,-1250)
\blacken\path(19750,9250)(20250,9250)(20000,8750)(19750,9250)

\path(22000,-2000)(22000,10000)
\path(22000,-4000)(22000,-3000)
\whiten\path(21750,-3250)(22250,-3250)(22000,-3000)(21750,-3250)
\blacken\path(21750,-1250)(22250,-1250)(22000,-750)(21750,-1250)
\blacken\path(21750,9250)(22250,9250)(22000,8750)(21750,9250)

\path(24000,-2000)(24000,10000)
\path(24000,-4000)(24000,-3000)
\whiten\path(23750,-3250)(24250,-3250)(24000,-3000)(23750,-3250)
\blacken\path(23750,-1250)(24250,-1250)(24000,-750)(23750,-1250)
\blacken\path(23750,9250)(24250,9250)(24000,8750)(23750,9250)

\path(26000,-2000)(26000,10000)
\path(26000,-4000)(26000,-3000)
\whiten\path(25750,-3250)(26250,-3250)(26000,-3000)(25750,-3250)
\blacken\path(25750,-1250)(26250,-1250)(26000,-750)(25750,-1250)
\blacken\path(25750,9250)(26250,9250)(26000,8750)(25750,9250)

\path(28000,-2000)(28000,10000)
\put(28000,-2750){\scriptsize$\bar{v}_N$}
\path(28000,-4000)(28000,-3000)
\whiten\path(27750,-3250)(28250,-3250)(28000,-3000)(27750,-3250)
\blacken\path(27750,-1250)(28250,-1250)(28000,-750)(27750,-1250)
\blacken\path(27750,9250)(28250,9250)(28000,8750)(27750,9250)

\end{picture}

\end{minipage}
\end{center}

\caption[Equivalent expressions for the domain wall partition function]{Equivalent expressions for the domain wall partition function. The diagram on the left represents a stacking of the operators $B(v_i,\{w\}_N)$. The diagram on the right represents a stacking of the operators $C(w_i,\{\bar{v}\}_N)$.}

\label{zequiv2}
\end{figure}
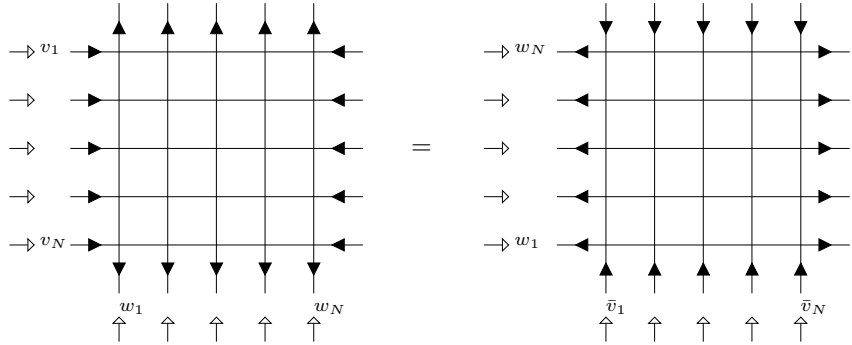

By equation (\ref{CC}) the $C(w_i,\{\bar{v}\}_N)$ operators all commute, proving that $Z_N(\{v\}_N,\{w\}_N)$ is symmetric in $\{w\}_N$. 
}
\end{conditions}

\begin{conditions}
{\rm 
By inserting the set of states $\sum_{n=1}^{N} \sigma_n^{+} |\Downarrow_N\rangle \langle \Downarrow_N| \sigma_n^{-}$ after the first $B$-operator appearing in (\ref{pf-xxz}) we obtain the expansion
\begin{align}
Z_N\Big(\{v\}_N,\{w\}_N\Big)
&=
\sum_{n=1}^{N}
\langle \Downarrow_N|
B(v_N,\{w\}_N)
\sigma^{+}_n |\Downarrow_N\rangle
\langle \Downarrow_N| \sigma^{-}_n
\prod_{i=1}^{N-1} B(v_i,\{w\}_N)
|\Uparrow_N\rangle,
\label{pfexp-xxz}
\end{align}
in which all dependence on $v_N$ appears in the first factor within the sum. Hence we proceed to calculate $\langle \Downarrow_N| B(v_N,\{w\}_N) \sigma_n^{+} |\Downarrow_N\rangle$ for all $1 \leq n \leq N$, as shown below: 

\begin{figure}[H]

\begin{center}
\begin{minipage}{4.3in}

\setlength{\unitlength}{0.00035cm}
\begin{picture}(20000,6000)(-3000,-3000)

\path(-2000,0)(10000,0)
\put(-3250,0){\scriptsize$v_N$}
\path(-4500,0)(-3500,0)
\whiten\path(-3750,250)(-3750,-250)(-3500,0)(-3750,250)
\blacken\path(-1250,250)(-1250,-250)(-750,0)(-1250,250)
\blacken\path(9250,250)(9250,-250)(8750,0)(9250,250)

%%%%%%%%%%%

\path(0,-2000)(0,2000)
\put(0,-2750){\scriptsize$w_1$}
\path(0,-4000)(0,-3000)
\whiten\path(-250,-3250)(250,-3250)(0,-3000)(-250,-3250)
\blacken\path(-250,-750)(250,-750)(0,-1250)(-250,-750)
\blacken\path(-250,1250)(250,1250)(0,750)(-250,1250)

\path(2000,-2000)(2000,2000)
\path(2000,-4000)(2000,-3000)
\whiten\path(1750,-3250)(2250,-3250)(2000,-3000)(1750,-3250)
\blacken\path(1750,-750)(2250,-750)(2000,-1250)(1750,-750)
\blacken\path(1750,1250)(2250,1250)(2000,750)(1750,1250)

\path(4000,-2000)(4000,2000)
\put(4000,-2750){\scriptsize$w_n$}
\path(4000,-4000)(4000,-3000)
\whiten\path(3750,-3250)(4250,-3250)(4000,-3000)(3750,-3250)
\blacken\path(3750,-750)(4250,-750)(4000,-1250)(3750,-750)
\blacken\path(3750,750)(4250,750)(4000,1250)(3750,750)

\path(6000,-2000)(6000,2000)
\path(6000,-4000)(6000,-3000)
\whiten\path(5750,-3250)(6250,-3250)(6000,-3000)(5750,-3250)
\blacken\path(5750,-750)(6250,-750)(6000,-1250)(5750,-750)
\blacken\path(5750,1250)(6250,1250)(6000,750)(5750,1250)

\path(8000,-2000)(8000,2000)
\put(8000,-2750){\scriptsize$w_N$}
\path(8000,-4000)(8000,-3000)
\whiten\path(7750,-3250)(8250,-3250)(8000,-3000)(7750,-3250)
\blacken\path(7750,-750)(8250,-750)(8000,-1250)(7750,-750)
\blacken\path(7750,1250)(8250,1250)(8000,750)(7750,1250)

%%%%%%%

\put(12000,-250){$=$}

%%%%%%%2nd lattice %%%%%%

\path(18000,0)(30000,0)
\put(16750,0){\scriptsize$v_N$}
\path(15500,0)(16500,0)
\whiten\path(16250,250)(16250,-250)(16500,0)(16250,250)
\blacken\path(18750,250)(18750,-250)(19250,0)(18750,250)
\blacken\path(20750,250)(20750,-250)(21250,0)(20750,250)
\blacken\path(22750,250)(22750,-250)(23250,0)(22750,250)
\blacken\path(25250,250)(25250,-250)(24750,0)(25250,250)
\blacken\path(27250,250)(27250,-250)(26750,0)(27250,250)
\blacken\path(29250,250)(29250,-250)(28750,0)(29250,250)

%%%%%%%%%%%

\path(20000,-2000)(20000,2000)
\put(20000,-2750){\scriptsize$w_1$}
\path(20000,-4000)(20000,-3000)
\whiten\path(19750,-3250)(20250,-3250)(20000,-3000)(19750,-3250)
\blacken\path(19750,-750)(20250,-750)(20000,-1250)(19750,-750)
\blacken\path(19750,1250)(20250,1250)(20000,750)(19750,1250)

\path(22000,-2000)(22000,2000)
\path(22000,-4000)(22000,-3000)
\whiten\path(21750,-3250)(22250,-3250)(22000,-3000)(21750,-3250)
\blacken\path(21750,-750)(22250,-750)(22000,-1250)(21750,-750)
\blacken\path(21750,1250)(22250,1250)(22000,750)(21750,1250)

\path(24000,-2000)(24000,2000)
\put(24000,-2750){\scriptsize$w_n$}
\path(24000,-4000)(24000,-3000)
\whiten\path(23750,-3250)(24250,-3250)(24000,-3000)(23750,-3250)
\blacken\path(23750,-750)(24250,-750)(24000,-1250)(23750,-750)
\blacken\path(23750,750)(24250,750)(24000,1250)(23750,750)

\path(26000,-2000)(26000,2000)
\path(26000,-4000)(26000,-3000)
\whiten\path(25750,-3250)(26250,-3250)(26000,-3000)(25750,-3250)
\blacken\path(25750,-750)(26250,-750)(26000,-1250)(25750,-750)
\blacken\path(25750,1250)(26250,1250)(26000,750)(25750,1250)

\path(28000,-2000)(28000,2000)
\put(28000,-2750){\scriptsize$w_N$}
\path(28000,-4000)(28000,-3000)
\whiten\path(27750,-3250)(28250,-3250)(28000,-3000)(27750,-3250)
\blacken\path(27750,-750)(28250,-750)(28000,-1250)(27750,-750)
\blacken\path(27750,1250)(28250,1250)(28000,750)(27750,1250)

\end{picture}

\end{minipage}
\end{center}

\caption[Peeling away the bottom row of the partition function]{Peeling away the bottom row of the partition function. The diagram on the left represents $\langle \Downarrow_N |B(v_N,\{w\}_N) \sigma_n^{+} |\Downarrow_N \rangle$, with the internal black arrows being summed over all configurations. The diagram on the right represents the only surviving configuration.}

\label{zequiv3}
\end{figure}
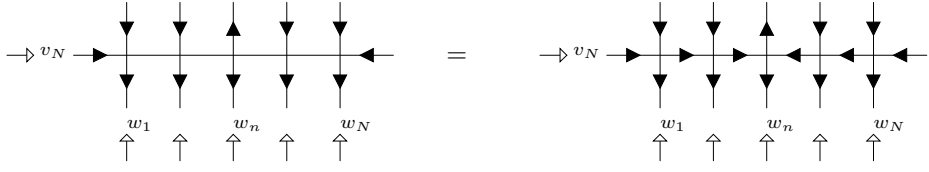

The right hand side of figure \ref{zequiv3} is simply a product of vertices. Replacing each vertex with its corresponding trigonometric weight (see figure 1), we conclude that
\begin{align}
\langle \Downarrow_N|
B(v_N,\{w\}_N)
\sigma_n^{+} |\Downarrow_N\rangle
=
\prod_{1\leq i<n} [v_N-w_i]
[\gamma]
\prod_{n < i \leq N} [v_N-w_i+\gamma].
\label{zequiv4}
\end{align}
Substituting (\ref{zequiv4}) into the expansion (\ref{pfexp-xxz}) gives 
\begin{align}
Z_N
=
\label{PFexp1}
[\gamma]
\sum_{n=1}^{N}
\prod_{1\leq i<n}
[v_N-w_i]
\prod_{n < i \leq N}
[v_N-w_i+\gamma]
\langle \Downarrow_N|\sigma_n^{-}
\prod_{i=1}^{N-1} B(v_i,\{w\}_N)
|\Uparrow_N\rangle.
\end{align}
From this equation we see that every term in $Z_N(\{v\}_N,\{w\}_N)$ contains a product of exactly $N-1$ trigonometric functions with argument $v_N$. Therefore $Z_N(\{v\}_N,\{w\}_N)$ is a trigonometric polynomial of degree $N-1$ in the variable $v_N$. 
}
\end{conditions}

\begin{conditions}
{\rm
We start from the expansion (\ref{PFexp1}) of the domain wall partition function, and set $v_N=w_N-\gamma$. This causes all terms in the summation over $1\leq n \leq N$ to collapse to zero except the $n=N$ term, giving
\begin{align}
Z_N\Big(\{v\}_N,\{w\}_N\Big) \Big|_{v_N=w_N-\gamma}
&=
[\gamma]
\prod_{i=1}^{N-1}
[w_N-w_i-\gamma]
\langle \Downarrow_N|
\sigma_N^{-}
\prod_{i=1}^{N-1} B(v_i,\{w\}_N)
|\Uparrow_N\rangle.
\label{PFexp3}
\end{align}
We then consider the graphical representation of $\langle \Downarrow_N | \sigma_N^{-} \prod_{i=1}^{N-1} B(v_i,\{w\}_N) |\Uparrow_N\rangle$, as shown below.

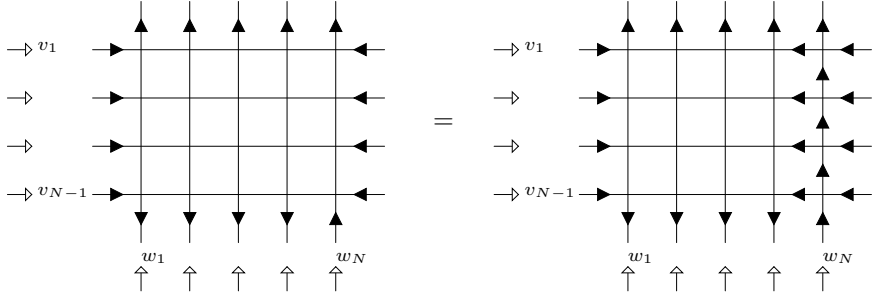
\begin{figure}[H]

\begin{center}
\begin{minipage}{4.3in}

\setlength{\unitlength}{0.000325cm}
\begin{picture}(20000,12500)(-4000,-3000)

\path(-2000,0)(10000,0)
\put(-4250,0){\scriptsize$v_{N-1}$}
\path(-5500,0)(-4500,0)
\whiten\path(-4750,250)(-4750,-250)(-4500,0)(-4750,250)
\blacken\path(-1250,250)(-1250,-250)(-750,0)(-1250,250)
\blacken\path(9250,250)(9250,-250)(8750,0)(9250,250)

\path(-2000,2000)(10000,2000)
\path(-5500,2000)(-4500,2000)
\whiten\path(-4750,2250)(-4750,1750)(-4500,2000)(-4750,2250)
\blacken\path(-1250,2250)(-1250,1750)(-750,2000)(-1250,2250)
\blacken\path(9250,2250)(9250,1750)(8750,2000)(9250,2250)

\path(-2000,4000)(10000,4000)
\path(-5500,4000)(-4500,4000)
\whiten\path(-4750,4250)(-4750,3750)(-4500,4000)(-4750,4250)
\blacken\path(-1250,4250)(-1250,3750)(-750,4000)(-1250,4250)
\blacken\path(9250,4250)(9250,3750)(8750,4000)(9250,4250)

\path(-2000,6000)(10000,6000)
\put(-4250,6000){\scriptsize$v_1$}
\path(-5500,6000)(-4500,6000)
\whiten\path(-4750,6250)(-4750,5750)(-4500,6000)(-4750,6250)
\blacken\path(-1250,6250)(-1250,5750)(-750,6000)(-1250,6250)
\blacken\path(9250,6250)(9250,5750)(8750,6000)(9250,6250)

%%%%%%%%%%%

\path(0,-2000)(0,8000)
\put(0,-2750){\scriptsize$w_1$}
\path(0,-4000)(0,-3000)
\whiten\path(-250,-3250)(250,-3250)(0,-3000)(-250,-3250)
\blacken\path(-250,-750)(250,-750)(0,-1250)(-250,-750)
\blacken\path(-250,6750)(250,6750)(0,7250)(-250,6750)

\path(2000,-2000)(2000,8000)
\path(2000,-4000)(2000,-3000)
\whiten\path(1750,-3250)(2250,-3250)(2000,-3000)(1750,-3250)
\blacken\path(1750,-750)(2250,-750)(2000,-1250)(1750,-750)
\blacken\path(1750,6750)(2250,6750)(2000,7250)(1750,6750)

\path(4000,-2000)(4000,8000)
\path(4000,-4000)(4000,-3000)
\whiten\path(3750,-3250)(4250,-3250)(4000,-3000)(3750,-3250)
\blacken\path(3750,-750)(4250,-750)(4000,-1250)(3750,-750)
\blacken\path(3750,6750)(4250,6750)(4000,7250)(3750,6750)

\path(6000,-2000)(6000,8000)
\path(6000,-4000)(6000,-3000)
\whiten\path(5750,-3250)(6250,-3250)(6000,-3000)(5750,-3250)
\blacken\path(5750,-750)(6250,-750)(6000,-1250)(5750,-750)
\blacken\path(5750,6750)(6250,6750)(6000,7250)(5750,6750)

\path(8000,-2000)(8000,8000)
\put(8000,-2750){\scriptsize$w_N$}
\path(8000,-4000)(8000,-3000)
\whiten\path(7750,-3250)(8250,-3250)(8000,-3000)(7750,-3250)
\blacken\path(7750,-1250)(8250,-1250)(8000,-750)(7750,-1250)
\blacken\path(7750,6750)(8250,6750)(8000,7250)(7750,6750)

%%%%%%%

\put(12000,2750){$=$}

%%%%%%%2nd lattice %%%%%%

\path(18000,0)(30000,0)
\put(15750,0){\scriptsize$v_{N-1}$}
\path(14500,0)(15500,0)
\whiten\path(15250,250)(15250,-250)(15500,0)(15250,250)
\blacken\path(18750,250)(18750,-250)(19250,0)(18750,250)
\blacken\path(27250,250)(27250,-250)(26750,0)(27250,250)
\blacken\path(29250,250)(29250,-250)(28750,0)(29250,250)

\path(18000,2000)(30000,2000)
\path(14500,2000)(15500,2000)
\whiten\path(15250,2250)(15250,1750)(15500,2000)(15250,2250)
\blacken\path(18750,2250)(18750,1750)(19250,2000)(18750,2250)
\blacken\path(27250,2250)(27250,1750)(26750,2000)(27250,2250)
\blacken\path(29250,2250)(29250,1750)(28750,2000)(29250,2250)

\path(18000,4000)(30000,4000)
\path(14500,4000)(15500,4000)
\whiten\path(15250,4250)(15250,3750)(15500,4000)(15250,4250)
\blacken\path(18750,4250)(18750,3750)(19250,4000)(18750,4250)
\blacken\path(27250,4250)(27250,3750)(26750,4000)(27250,4250)
\blacken\path(29250,4250)(29250,3750)(28750,4000)(29250,4250)

\path(18000,6000)(30000,6000)
\put(15750,6000){\scriptsize$v_1$}
\path(14500,6000)(15500,6000)
\whiten\path(15250,6250)(15250,5750)(15500,6000)(15250,6250)
\blacken\path(18750,6250)(18750,5750)(19250,6000)(18750,6250)
\blacken\path(27250,6250)(27250,5750)(26750,6000)(27250,6250)
\blacken\path(29250,6250)(29250,5750)(28750,6000)(29250,6250)

%%%%%%%%%%%

\path(20000,-2000)(20000,8000)
\put(20000,-2750){\scriptsize$w_1$}
\path(20000,-4000)(20000,-3000)
\whiten\path(19750,-3250)(20250,-3250)(20000,-3000)(19750,-3250)
\blacken\path(19750,-750)(20250,-750)(20000,-1250)(19750,-750)
\blacken\path(19750,6750)(20250,6750)(20000,7250)(19750,6750)

\path(22000,-2000)(22000,8000)
\path(22000,-4000)(22000,-3000)
\whiten\path(21750,-3250)(22250,-3250)(22000,-3000)(21750,-3250)
\blacken\path(21750,-750)(22250,-750)(22000,-1250)(21750,-750)
\blacken\path(21750,6750)(22250,6750)(22000,7250)(21750,6750)

\path(24000,-2000)(24000,8000)
\path(24000,-4000)(24000,-3000)
\whiten\path(23750,-3250)(24250,-3250)(24000,-3000)(23750,-3250)
\blacken\path(23750,-750)(24250,-750)(24000,-1250)(23750,-750)
\blacken\path(23750,6750)(24250,6750)(24000,7250)(23750,6750)

\path(26000,-2000)(26000,8000)
\path(26000,-4000)(26000,-3000)
\whiten\path(25750,-3250)(26250,-3250)(26000,-3000)(25750,-3250)
\blacken\path(25750,-750)(26250,-750)(26000,-1250)(25750,-750)
\blacken\path(25750,6750)(26250,6750)(26000,7250)(25750,6750)

\path(28000,-2000)(28000,8000)
\put(28000,-2750){\scriptsize$w_N$}
\path(28000,-4000)(28000,-3000)
\whiten\path(27750,-3250)(28250,-3250)(28000,-3000)(27750,-3250)
\blacken\path(27750,-1250)(28250,-1250)(28000,-750)(27750,-1250)
\blacken\path(27750,750)(28250,750)(28000,1250)(27750,750)
\blacken\path(27750,2750)(28250,2750)(28000,3250)(27750,2750)
\blacken\path(27750,4750)(28250,4750)(28000,5250)(27750,4750)
\blacken\path(27750,6750)(28250,6750)(28000,7250)(27750,6750)

\end{picture}

\end{minipage}
\end{center}

\caption[Peeling away the right-most column of the partition function]{Peeling away the right-most column of the partition function. The left hand side represents the quantity $\langle \Downarrow_N | \sigma_N^{-} \prod_{i=1}^{N-1} B(v_i,\{w\}_N) |\Uparrow_N\rangle$, with the internal black arrows being summed over all configurations. The diagram on the right contains all surviving configurations.}

\label{zequiv5}
\end{figure}

The right hand side of figure \ref{zequiv5} represents the $(N-1)\times (N-1)$ domain wall partition function, multiplied by a column of vertices. Replacing these vertices with their trigonometric weights, we find that 
\begin{align}
\langle \Downarrow_N|
\sigma_N^{-} 
\prod_{i=1}^{N-1} B(v_i,\{w\}_N) 
|\Uparrow_N\rangle
=
\prod_{i=1}^{N-1}
[v_i-w_N]
Z_{N-1}\Big(\{v\}_{N-1},\{w\}_{N-1}\Big).
\label{zequiv6}
\end{align}
Substituting (\ref{zequiv6}) into (\ref{PFexp3}) produces the required recursion relation (\ref{recursion}).

}
\end{conditions}

\begin{conditions}
{\rm
Specializing the definition (\ref{pf-xxz}) to the case $N=1$ gives
\begin{align} 
Z_1(v_1,w_1) 
=
\langle \Downarrow_1 |
B(v_1,w_1)
| \Uparrow_1 \rangle
=
\uparrow_{a_1}^{*} \otimes \downarrow_1^{*}
R_{a_1 1}(v_1,w_1)
\uparrow_1 \otimes \downarrow_{a_1}
=
[\gamma],
\end{align}
as required. Alternatively, the lattice representation of $Z_1$ is simply the top-right vertex in figure \ref{R6v}, whose weight is equal to $[\gamma]$.  
} 
\end{conditions}
\end{proof}

\begin{lemma}
\label{uniqueness}
{\rm 
Let $\{\breve{Z}_N\}_{N \in \mathbb{N}}$ denote a set of functions $\breve{Z}_N(\{v\}_N,\{w\}_N)$ which satisfy the four conditions of the previous lemma. Then $\breve{Z}_N = Z_N$ for all $N \geq 1$. In other words, the conditions imposed on the domain wall partition function determine it uniquely. 
}
\end{lemma}

\begin{proof}
From condition {\bf 4} on $\breve{Z}_1,Z_1$ we know that $\breve{Z}_1=Z_1$. Hence we may assume that $\breve{Z}_{N-1} = Z_{N-1}$ for some $N \geq 2$. Using this assumption together with condition {\bf 3} on $\breve{Z}_N,Z_N$ yields 
\begin{align}
\breve{Z}_N
\Big|_{v_N = w_N-\gamma}
&=
[\gamma] \prod_{i=1}^{N-1} [v_i-w_N] [w_N-w_i-\gamma]
\breve{Z}_{N-1},
\\
&=
[\gamma] \prod_{i=1}^{N-1} [v_i-w_N] [w_N-w_i-\gamma]
Z_{N-1}
=
Z_{N} \Big|_{v_N=w_N-\gamma}.
\nonumber
\end{align}
Condition {\bf 1} on $\breve{Z}_N,Z_N$ states that both are symmetric in the variables $\{w\}_N$. Using this fact in the previous equation, we find that 
\begin{align} 
\breve{Z}_N \Big|_{v_N=w_{i}-\gamma}
=
Z_N \Big|_{v_N=w_{i}-\gamma}
\quad
{\rm for\ all}\ 1 \leq i \leq N,
\end{align}
which proves that $\breve{Z}_N$ and $Z_N$ are equal at $N$ distinct values of $v_N$. By condition {\bf 2}, both functions are trigonometric polynomials of degree $N-1$ in $v_N$, so their equality at $N$ points implies $\breve{Z}_N = Z_N$ everywhere. This completes the proof by induction.

\end{proof}

\subsection{Determinant expression for $Z_N$}

In \cite{ize}, Izergin found a function which satisfies the four conditions of the previous subsection, and is therefore equal to the domain wall partition function. We present this formula below.

\begin{lemma}
{\rm For all $N \geq 1$ we define  
\begin{align}
\breve{Z}_N\Big(\{v\}_N,\{w\}_N\Big)
&=
\frac{\displaystyle{\prod_{i,j=1}^{N}} [v_i-w_j+\gamma] [v_i-w_j]}
{\displaystyle{\prod_{1 \leq i < j \leq N}} [v_i-v_j] [w_j-w_i]}
\det\left(\frac{[\gamma]}{[v_i-w_j+\gamma][v_i-w_j]}\right)_{1\leq i,j \leq N},
\label{dwpf}
\\
&=
\frac{
\det\left(
[\gamma] 
\displaystyle{\prod_{k \not= i}^{N} [v_k-w_j+\gamma] [v_k-w_j]}
\right)_{1\leq i,j \leq N}
}
{
\displaystyle{
\prod_{1 \leq i < j \leq N} [v_i-v_j] [w_j-w_i]
}
}.
\nonumber
\end{align}
The functions $\{\breve{Z}_N\}_{N\in \mathbb{N}}$ satisfy the four conditions of lemma 6. Equivalently, the domain wall partition function $Z_N$ is equal to the right hand side of (\ref{dwpf}).
}
\end{lemma}

\begin{proof}
 
\setcounter{conditions}{0}
\begin{conditions}
{\rm
Permuting $w_m\leftrightarrow w_n$ in $\det\left([\gamma] \prod_{k \not= i}^{N} [v_k-w_j+\gamma] [v_k-w_j]\right)_{1\leq i,j \leq N}$ swaps two columns of the determinant, which introduces a minus sign into the numerator of (\ref{dwpf}). Similarly, permuting $w_m \leftrightarrow w_n$ in $\prod_{1 \leq i<j \leq N} [w_j-w_i]$ introduces a minus sign into the denominator of (\ref{dwpf}). These minus signs cancel, leaving $\breve{Z}_N$ invariant under permutations of the $\{w\}_N$ variables. 
}
\end{conditions}

\begin{conditions}
{\rm
The numerator of (\ref{dwpf}) is a trigonometric polynomial of degree $2N-2$ in $v_N$, with zeros at the points $v_N=v_n$ for all $1 \leq n \leq N-1$, since such a substitution would render two rows of the determinant equal. The denominator of (\ref{dwpf}) is a trigonometric polynomial of degree $N-1$ in $v_N$, with zeros at the same points. Cancelling these common zeros, $\breve{Z}_N$ is a trigonometric polynomial of degree $N-1$ in $v_N$. 
}
\end{conditions}

\begin{conditions}
{\rm 
Expanding the determinant in $\breve{Z}_N$ along its $N^{\rm th}$ row we find
\begin{align}
\breve{Z}_N\Big(\{v\}_N,\{w\}_N\Big)
&=
\sum_{j=1}^{N}
\breve{Z}_{N-1}\Big(\{v\}_{N-1},\{w_1,\ldots, \widehat{w_j},\ldots,w_N\}\Big)
\\
&
\times
\frac{
\displaystyle{
[\gamma]
\prod_{k=1}^{N-1}
[v_k-w_j+\gamma] 
[v_k-w_j]
\prod_{k\not=j}^{N}
[v_N-w_k+\gamma]
[v_N-w_k]
}
}{
\displaystyle{
\prod_{k = 1}^{N-1}
[v_k-v_N]
\prod_{k \not= j}^{N}
[w_j-w_k]
}
},
\nonumber
\end{align}
where $\widehat{w_j}$ denotes the omission of that variable. Setting $v_N=w_N-\gamma$ in the above expression, all terms in the summation over $1 \leq j \leq N$ vanish except the $j=N$ term and we obtain
\begin{align}
\breve{Z}_N
\Big|_{v_N=w_N-\gamma}
&=
[\gamma]
\prod_{k=1}^{N-1}
[v_k-w_N]
[w_N-w_k-\gamma]
\breve{Z}_{N-1}
\Big(\{v\}_{N-1},\{w\}_{N-1}\Big),
\end{align} 
which is the desired recursion relation.
}
\end{conditions}

\begin{conditions}
{\rm 
From the definition (\ref{dwpf}) it is clear that $\breve{Z}_1=[\gamma]$.
} 
\end{conditions}

\end{proof}

\section{Scalar products $S_n\Big(\{u\}_n,\{v\}_N,\{w\}_M\Big)$}
\label{xxz-sp}

In this section we define and calculate a sequence of {\it intermediate scalar products} $S_n$, which interpolate between the domain wall partition function and the full scalar product. The domain wall partition function corresponds to the case $n=0$, whereas the full scalar product is given by the case $n=N$. These functions were originally studied in \cite{kmt}, using Drinfel'd twists in the algebraic Bethe Ansatz setting \cite{ms}. Our approach owes much to the work of\cite{kmt} but is somewhat different, since it does not involve constructing the $F$-basis.

\subsection{Definition of $S_n(\{u\}_n,\{v\}_N,\{w\}_M)$}
\label{xxz-sp-def}

Let $\{u\}_n = \{u_1,\ldots,u_n\}$, $\{v\}_N = \{v_1,\ldots,v_N\}$, $\{w\}_M = \{w_1,\ldots,w_M\}$ be three sets of variables whose cardinalities satisfy $0 \leq n \leq N$ and $1\leq N \leq M$. We proceed to introduce functions $S_n(\{u\}_n,\{v\}_N,\{w\}_M)$ for all $0 \leq n \leq N$. In the case $n=0$ we define
\begin{align}
S_0 \Big( \{v\}_N, \{w\}_M \Big) 
=
\langle \Downarrow_{N/M}|
\prod_{j=1}^{N} 
B(v_j,\{w\}_M)  
|\Uparrow_M\rangle,
\label{s0-xxz}
\end{align}
where for conciseness we have suppressed dependence on the set $\{u\}_0 = \emptyset$. As we will soon show, up to an overall normalization the scalar product $S_0$ is equal to the domain wall partition function $Z_N$. Next, for all $1\leq n \leq N-1$ we define
\begin{align}
S_n \Big( \{u\}_n, \{v\}_N, \{w\}_M \Big)
=
\langle \Downarrow_{\widetilde{N}/M}|
\prod_{i=1}^{n}
C(u_i,\{w\}_M)
\prod_{j=1}^{N}  
B(v_j,\{w\}_M) 
|\Uparrow_M\rangle,
\label{sn-xxz}
\end{align}
where we have adopted the notation $\widetilde{N} = N-n$, which is used frequently hereafter. Finally, for $n=N$ we define
\begin{align}
&
S_N \Big( \{u\}_N, \{v\}_N, \{w\}_M \Big)
=
\langle \Uparrow_M|
\prod_{i=1}^{N}
C(u_i,\{w\}_M)
\prod_{j=1}^{N} 
B(v_j,\{w\}_M) 
|\Uparrow_M\rangle,
\label{sN-xxz}
\end{align}
which represents the usual scalar product. In all cases (\ref{s0-xxz})--(\ref{sN-xxz}) we shall assume that the parameters $\{v\}_N$ obey the Bethe equations (\ref{bethe2}), while the remaining variables $\{u\}_n$ are considered free. Accordingly, we name these objects {\it Bethe scalar products}. It turns out that $\{S_n\}_{0 \leq n \leq N}$ are related by a simple recursion. Hence they provide a convenient way of calculating $S_N$, starting from Izergin's determinant formula (\ref{dwpf}) for $Z_N$.  

\subsection{Graphical representation of scalar products}
\label{xxz-sp-graph}

We now provide lattice representations of the Bethe scalar products $\{S_n\}_{0 \leq n \leq N}$. As in the case of the domain wall partition function, the lattices simplify the calculation of these functions.

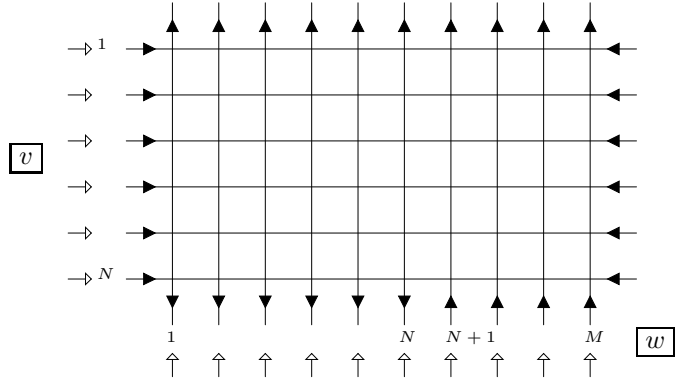
\begin{figure}[H]

\begin{center}
\begin{minipage}{4.3in}

\setlength{\unitlength}{0.0003cm}
\begin{picture}(20000,15000)(-9000,-4000)

\put(-7000,5000){\fbox{$v$}}

%%%%%%%%

\path(-2000,0)(20000,0)
\put(-3250,0){\tiny{$N$}}
\path(-4500,0)(-3500,0)
\whiten\path(-3750,250)(-3750,-250)(-3500,0)(-3750,250)
\blacken\path(-1250,250)(-1250,-250)(-750,0)(-1250,250)
\blacken\path(19250,250)(19250,-250)(18750,0)(19250,250)

\path(-2000,2000)(20000,2000)
\path(-4500,2000)(-3500,2000)
\whiten\path(-3750,2250)(-3750,1750)(-3500,2000)(-3750,2250)
\blacken\path(-1250,2250)(-1250,1750)(-750,2000)(-1250,2250)
\blacken\path(19250,2250)(19250,1750)(18750,2000)(19250,2250)

\path(-2000,4000)(20000,4000)
\path(-4500,4000)(-3500,4000)
\whiten\path(-3750,4250)(-3750,3750)(-3500,4000)(-3750,4250)
\blacken\path(-1250,4250)(-1250,3750)(-750,4000)(-1250,4250)
\blacken\path(19250,4250)(19250,3750)(18750,4000)(19250,4250)

\path(-2000,6000)(20000,6000)
\path(-4500,6000)(-3500,6000)
\whiten\path(-3750,6250)(-3750,5750)(-3500,6000)(-3750,6250)
\blacken\path(-1250,6250)(-1250,5750)(-750,6000)(-1250,6250)
\blacken\path(19250,6250)(19250,5750)(18750,6000)(19250,6250)

\path(-2000,8000)(20000,8000)
\path(-4500,8000)(-3500,8000)
\whiten\path(-3750,8250)(-3750,7750)(-3500,8000)(-3750,8250)
\blacken\path(-1250,8250)(-1250,7750)(-750,8000)(-1250,8250)
\blacken\path(19250,8250)(19250,7750)(18750,8000)(19250,8250)

\path(-2000,10000)(20000,10000)
\put(-3250,10000){\tiny$1$}
\path(-4500,10000)(-3500,10000)
\whiten\path(-3750,10250)(-3750,9750)(-3500,10000)(-3750,10250)
\blacken\path(-1250,10250)(-1250,9750)(-750,10000)(-1250,10250)
\blacken\path(19250,10250)(19250,9750)(18750,10000)(19250,10250)

%%%%%%%%%%%

\put(20000,-3000){\fbox{$w$}}

%%%%%%%%%%

\path(0,-2000)(0,12000)
\put(-250,-2750){\tiny$1$}
\path(0,-4250)(0,-3250)
\whiten\path(-250,-3500)(250,-3500)(0,-3250)(-250,-3500)
\blacken\path(-250,-750)(250,-750)(0,-1250)(-250,-750)
\blacken\path(-250,10750)(250,10750)(0,11250)(-250,10750)

\path(2000,-2000)(2000,12000)
\path(2000,-4250)(2000,-3250)
\whiten\path(1750,-3500)(2250,-3500)(2000,-3250)(1750,-3500)
\blacken\path(1750,-750)(2250,-750)(2000,-1250)(1750,-750)
\blacken\path(1750,10750)(2250,10750)(2000,11250)(1750,10750)

\path(4000,-2000)(4000,12000)
\path(4000,-4250)(4000,-3250)
\whiten\path(3750,-3500)(4250,-3500)(4000,-3250)(3750,-3500)
\blacken\path(3750,-750)(4250,-750)(4000,-1250)(3750,-750)
\blacken\path(3750,10750)(4250,10750)(4000,11250)(3750,10750)

\path(6000,-2000)(6000,12000)
\path(6000,-4250)(6000,-3250)
\whiten\path(5750,-3500)(6250,-3500)(6000,-3250)(5750,-3500)
\blacken\path(5750,-750)(6250,-750)(6000,-1250)(5750,-750)
\blacken\path(5750,10750)(6250,10750)(6000,11250)(5750,10750)

\path(8000,-2000)(8000,12000)
\path(8000,-4250)(8000,-3250)
\whiten\path(7750,-3500)(8250,-3500)(8000,-3250)(7750,-3500)
\blacken\path(7750,-750)(8250,-750)(8000,-1250)(7750,-750)
\blacken\path(7750,10750)(8250,10750)(8000,11250)(7750,10750)

\path(10000,-2000)(10000,12000)
\put(9750,-2750){\tiny$N$}
\path(10000,-4250)(10000,-3250)
\whiten\path(9750,-3500)(10250,-3500)(10000,-3250)(9750,-3500)
\blacken\path(9750,-750)(10250,-750)(10000,-1250)(9750,-750)
\blacken\path(9750,10750)(10250,10750)(10000,11250)(9750,10750)

\path(12000,-2000)(12000,12000)
\put(11750,-2750){\tiny$N+1$}
\path(12000,-4250)(12000,-3250)
\whiten\path(11750,-3500)(12250,-3500)(12000,-3250)(11750,-3500)
\blacken\path(11750,-1250)(12250,-1250)(12000,-750)(11750,-1250)
\blacken\path(11750,10750)(12250,10750)(12000,11250)(11750,10750)

\path(14000,-2000)(14000,12000)
\path(14000,-4250)(14000,-3250)
\whiten\path(13750,-3500)(14250,-3500)(14000,-3250)(13750,-3500)
\blacken\path(13750,-1250)(14250,-1250)(14000,-750)(13750,-1250)
\blacken\path(13750,10750)(14250,10750)(14000,11250)(13750,10750)

\path(16000,-2000)(16000,12000)
\path(16000,-4250)(16000,-3250)
\whiten\path(15750,-3500)(16250,-3500)(16000,-3250)(15750,-3500)
\blacken\path(15750,-1250)(16250,-1250)(16000,-750)(15750,-1250)
\blacken\path(15750,10750)(16250,10750)(16000,11250)(15750,10750)

\path(18000,-2000)(18000,12000)
\put(17750,-2750){\tiny$M$}
\path(18000,-4250)(18000,-3250)
\whiten\path(17750,-3500)(18250,-3500)(18000,-3250)(17750,-3500)
\blacken\path(17750,-1250)(18250,-1250)(18000,-750)(17750,-1250)
\blacken\path(17750,10750)(18250,10750)(18000,11250)(17750,10750)

\end{picture}

\end{minipage}
\end{center}

\caption[Lattice representation of $S_0$]{Lattice representation of $S_0$. The top row of arrows corresponds with the state $|\Uparrow_M\rangle$, while the bottom row of arrows corresponds with the dual state $\langle \Downarrow_{N/M}|$. Each horizontal lattice line represents a $B$-operator $B(v_j,\{w\}_M)$.}
\end{figure}

\begin{figure}[H]

\begin{center}
\begin{minipage}{4.3in}

\setlength{\unitlength}{0.0003cm}
\begin{picture}(20000,20000)(-9000,-10000)

\put(-7000,-4000){\fbox{$u$}}

%%%%%%%%%

\path(-2000,-6000)(20000,-6000)
\put(-3250,-6000){\tiny$n$}
\path(-4500,-6000)(-3500,-6000)
\whiten\path(-3750,-5750)(-3750,-6250)(-3500,-6000)(-3750,-5750)
\blacken\path(-750,-5750)(-750,-6250)(-1250,-6000)(-750,-5750)
\blacken\path(18750,-5750)(18750,-6250)(19250,-6000)(18750,-5750)

\path(-2000,-4000)(20000,-4000)
\path(-4500,-4000)(-3500,-4000)
\whiten\path(-3750,-3750)(-3750,-4250)(-3500,-4000)(-3750,-3750)
\blacken\path(-750,-3750)(-750,-4250)(-1250,-4000)(-750,-3750)
\blacken\path(18750,-3750)(18750,-4250)(19250,-4000)(18750,-3750)

\path(-2000,-2000)(20000,-2000)
\put(-3250,-2000){\tiny$1$}
\path(-4500,-2000)(-3500,-2000)
\whiten\path(-3750,-1750)(-3750,-2250)(-3500,-2000)(-3750,-1750)
\blacken\path(-750,-1750)(-750,-2250)(-1250,-2000)(-750,-1750)
\blacken\path(18750,-1750)(18750,-2250)(19250,-2000)(18750,-1750)

%%%%%%%%%

\put(-7000,5000){\fbox{$v$}}

%%%%%%%%%

\path(-2000,0)(20000,0)
\put(-3250,0){\tiny$N$}
\path(-4500,0)(-3500,0)
\whiten\path(-3750,250)(-3750,-250)(-3500,0)(-3750,250)
\blacken\path(-1250,250)(-1250,-250)(-750,0)(-1250,250)
\blacken\path(19250,250)(19250,-250)(18750,0)(19250,250)

\path(-2000,2000)(20000,2000)
\path(-4500,2000)(-3500,2000)
\whiten\path(-3750,2250)(-3750,1750)(-3500,2000)(-3750,2250)
\blacken\path(-1250,2250)(-1250,1750)(-750,2000)(-1250,2250)
\blacken\path(19250,2250)(19250,1750)(18750,2000)(19250,2250)

\path(-2000,4000)(20000,4000)
\path(-4500,4000)(-3500,4000)
\whiten\path(-3750,4250)(-3750,3750)(-3500,4000)(-3750,4250)
\blacken\path(-1250,4250)(-1250,3750)(-750,4000)(-1250,4250)
\blacken\path(19250,4250)(19250,3750)(18750,4000)(19250,4250)

\path(-2000,6000)(20000,6000)
\path(-4500,6000)(-3500,6000)
\whiten\path(-3750,6250)(-3750,5750)(-3500,6000)(-3750,6250)
\blacken\path(-1250,6250)(-1250,5750)(-750,6000)(-1250,6250)
\blacken\path(19250,6250)(19250,5750)(18750,6000)(19250,6250)

\path(-2000,8000)(20000,8000)
\path(-4500,8000)(-3500,8000)
\whiten\path(-3750,8250)(-3750,7750)(-3500,8000)(-3750,8250)
\blacken\path(-1250,8250)(-1250,7750)(-750,8000)(-1250,8250)
\blacken\path(19250,8250)(19250,7750)(18750,8000)(19250,8250)

\path(-2000,10000)(20000,10000)
\put(-3250,10000){\tiny$1$}
\path(-4500,10000)(-3500,10000)
\whiten\path(-3750,10250)(-3750,9750)(-3500,10000)(-3750,10250)
\blacken\path(-1250,10250)(-1250,9750)(-750,10000)(-1250,10250)
\blacken\path(19250,10250)(19250,9750)(18750,10000)(19250,10250)

%%%%%%%%%%%

\put(20000,-9000){\fbox{$w$}}

%%%%%%%%%%%

\path(0,-8000)(0,12000)
\put(-250,-9000){\tiny$1$}
\path(0,-10500)(0,-9500)
\whiten\path(-250,-9750)(250,-9750)(0,-9500)(-250,-9750)
\blacken\path(-250,-6750)(250,-6750)(0,-7250)(-250,-6750)
\blacken\path(-250,10750)(250,10750)(0,11250)(-250,10750)

\path(2000,-8000)(2000,12000)
\path(2000,-10500)(2000,-9500)
\whiten\path(1750,-9750)(2250,-9750)(2000,-9500)(1750,-9750)
\blacken\path(1750,-6750)(2250,-6750)(2000,-7250)(1750,-6750)
\blacken\path(1750,10750)(2250,10750)(2000,11250)(1750,10750)

\path(4000,-8000)(4000,12000)
\put(3750,-9000){\tiny$\widetilde{N}$}
\path(4000,-10500)(4000,-9500)
\whiten\path(3750,-9750)(4250,-9750)(4000,-9500)(3750,-9750)
\blacken\path(3750,-6750)(4250,-6750)(4000,-7250)(3750,-6750)
\blacken\path(3750,10750)(4250,10750)(4000,11250)(3750,10750)

\path(6000,-8000)(6000,12000)
\put(5750,-9000){\tiny$\widetilde{N}+1$}
\path(6000,-10500)(6000,-9500)
\whiten\path(5750,-9750)(6250,-9750)(6000,-9500)(5750,-9750)
\blacken\path(5750,-7250)(6250,-7250)(6000,-6750)(5750,-7250)
\blacken\path(5750,10750)(6250,10750)(6000,11250)(5750,10750)

\path(8000,-8000)(8000,12000)
\path(8000,-10500)(8000,-9500)
\whiten\path(7750,-9750)(8250,-9750)(8000,-9500)(7750,-9750)
\blacken\path(7750,-7250)(8250,-7250)(8000,-6750)(7750,-7250)
\blacken\path(7750,10750)(8250,10750)(8000,11250)(7750,10750)

\path(10000,-8000)(10000,12000)
\path(10000,-10500)(10000,-9500)
\whiten\path(9750,-9750)(10250,-9750)(10000,-9500)(9750,-9750)
\blacken\path(9750,-7250)(10250,-7250)(10000,-6750)(9750,-7250)
\blacken\path(9750,10750)(10250,10750)(10000,11250)(9750,10750)

\path(12000,-8000)(12000,12000)
\path(12000,-10500)(12000,-9500)
\whiten\path(11750,-9750)(12250,-9750)(12000,-9500)(11750,-9750)
\blacken\path(11750,-7250)(12250,-7250)(12000,-6750)(11750,-7250)
\blacken\path(11750,10750)(12250,10750)(12000,11250)(11750,10750)

\path(14000,-8000)(14000,12000)
\path(14000,-10500)(14000,-9500)
\whiten\path(13750,-9750)(14250,-9750)(14000,-9500)(13750,-9750)
\blacken\path(13750,-7250)(14250,-7250)(14000,-6750)(13750,-7250)
\blacken\path(13750,10750)(14250,10750)(14000,11250)(13750,10750)

\path(16000,-8000)(16000,12000)
\path(16000,-10500)(16000,-9500)
\whiten\path(15750,-9750)(16250,-9750)(16000,-9500)(15750,-9750)
\blacken\path(15750,-7250)(16250,-7250)(16000,-6750)(15750,-7250)
\blacken\path(15750,10750)(16250,10750)(16000,11250)(15750,10750)

\path(18000,-8000)(18000,12000)
\put(17750,-9000){\tiny$M$}
\path(18000,-10500)(18000,-9500)
\whiten\path(17750,-9750)(18250,-9750)(18000,-9500)(17750,-9750)
\blacken\path(17750,-7250)(18250,-7250)(18000,-6750)(17750,-7250)
\blacken\path(17750,10750)(18250,10750)(18000,11250)(17750,10750)

\end{picture}

\end{minipage}
\end{center}

\caption[Lattice representation of $S_n$]{Lattice representation of $S_n$. The top row of arrows corresponds with the state $|\Uparrow_M\rangle$, while the bottom row of arrows corresponds with the dual state $\langle \Downarrow_{\widetilde{N}/M}|$. The highest $N$ horizontal lines represent $B$-operators $B(v_j,\{w\}_M)$, while the lowest $n$ horizontal lines represent $C$-operators $C(u_i,\{w\}_M)$.}

\label{lat}

\end{figure}

\begin{figure}[H]

\begin{center}
\begin{minipage}{4.3in}

\setlength{\unitlength}{0.0003cm}
\begin{picture}(20000,18000)(-9000,-9000)

\put(-7000,-3000){\fbox{$u$}}

%%%%%%%

\path(-2000,-6000)(20000,-6000)
\put(-3250,-6000){\tiny$N$}
\path(-4500,-6000)(-3500,-6000)
\whiten\path(-3750,-5750)(-3750,-6250)(-3500,-6000)(-3750,-5750)
\blacken\path(-750,-5750)(-750,-6250)(-1250,-6000)(-750,-5750)
\blacken\path(18750,-5750)(18750,-6250)(19250,-6000)(18750,-5750)

\path(-2000,-4000)(20000,-4000)
\path(-4500,-4000)(-3500,-4000)
\whiten\path(-3750,-3750)(-3750,-4250)(-3500,-4000)(-3750,-3750)
\blacken\path(-750,-3750)(-750,-4250)(-1250,-4000)(-750,-3750)
\blacken\path(18750,-3750)(18750,-4250)(19250,-4000)(18750,-3750)

\path(-2000,-2000)(20000,-2000)
\path(-4500,-2000)(-3500,-2000)
\whiten\path(-3750,-1750)(-3750,-2250)(-3500,-2000)(-3750,-1750)
\blacken\path(-750,-1750)(-750,-2250)(-1250,-2000)(-750,-1750)
\blacken\path(18750,-1750)(18750,-2250)(19250,-2000)(18750,-1750)

\path(-2000,0)(20000,0)
\put(-3250,0){\tiny$1$}
\path(-4500,0)(-3500,0)
\whiten\path(-3750,250)(-3750,-250)(-3500,0)(-3750,250)
\blacken\path(-750,250)(-750,-250)(-1250,0)(-750,250)
\blacken\path(18750,250)(18750,-250)(19250,0)(18750,250)

%%%%%%%%

\put(-7000,5000){\fbox{$v$}}

%%%%%%%%

\path(-2000,2000)(20000,2000)
\put(-3250,2000){\tiny$N$}
\path(-4500,2000)(-3500,2000)
\whiten\path(-3750,2250)(-3750,1750)(-3500,2000)(-3750,2250)
\blacken\path(-1250,2250)(-1250,1750)(-750,2000)(-1250,2250)
\blacken\path(19250,2250)(19250,1750)(18750,2000)(19250,2250)

\path(-2000,4000)(20000,4000)
\path(-4500,4000)(-3500,4000)
\whiten\path(-3750,4250)(-3750,3750)(-3500,4000)(-3750,4250)
\blacken\path(-1250,4250)(-1250,3750)(-750,4000)(-1250,4250)
\blacken\path(19250,4250)(19250,3750)(18750,4000)(19250,4250)

\path(-2000,6000)(20000,6000)
\path(-4500,6000)(-3500,6000)
\whiten\path(-3750,6250)(-3750,5750)(-3500,6000)(-3750,6250)
\blacken\path(-1250,6250)(-1250,5750)(-750,6000)(-1250,6250)
\blacken\path(19250,6250)(19250,5750)(18750,6000)(19250,6250)

\path(-2000,8000)(20000,8000)
\put(-3250,8000){\tiny$1$}
\path(-4500,8000)(-3500,8000)
\whiten\path(-3750,8250)(-3750,7750)(-3500,8000)(-3750,8250)
\blacken\path(-1250,8250)(-1250,7750)(-750,8000)(-1250,8250)
\blacken\path(19250,8250)(19250,7750)(18750,8000)(19250,8250)

%%%%%%%%%%%

\put(20000,-9000){\fbox{$w$}}

%%%%%%%%%%%

\path(0,-8000)(0,10000)
\put(-250,-8750){\tiny$1$}
\path(0,-10250)(0,-9250)
\whiten\path(-250,-9500)(250,-9500)(0,-9250)(-250,-9500)
\blacken\path(-250,-7250)(250,-7250)(0,-6750)(-250,-7250)
\blacken\path(-250,8750)(250,8750)(0,9250)(-250,8750)

\path(2000,-8000)(2000,10000)
\path(2000,-10250)(2000,-9250)
\whiten\path(1750,-9500)(2250,-9500)(2000,-9250)(1750,-9500)
\blacken\path(1750,-7250)(2250,-7250)(2000,-6750)(1750,-7250)
\blacken\path(1750,8750)(2250,8750)(2000,9250)(1750,8750)

\path(4000,-8000)(4000,10000)
\path(4000,-10250)(4000,-9250)
\whiten\path(3750,-9500)(4250,-9500)(4000,-9250)(3750,-9500)
\blacken\path(3750,-7250)(4250,-7250)(4000,-6750)(3750,-7250)
\blacken\path(3750,8750)(4250,8750)(4000,9250)(3750,8750)

\path(6000,-8000)(6000,10000)
\path(6000,-10250)(6000,-9250)
\whiten\path(5750,-9500)(6250,-9500)(6000,-9250)(5750,-9500)
\blacken\path(5750,-7250)(6250,-7250)(6000,-6750)(5750,-7250)
\blacken\path(5750,8750)(6250,8750)(6000,9250)(5750,8750)

\path(8000,-8000)(8000,10000)
\path(8000,-10250)(8000,-9250)
\whiten\path(7750,-9500)(8250,-9500)(8000,-9250)(7750,-9500)
\blacken\path(7750,-7250)(8250,-7250)(8000,-6750)(7750,-7250)
\blacken\path(7750,8750)(8250,8750)(8000,9250)(7750,8750)

\path(10000,-8000)(10000,10000)
\path(10000,-10250)(10000,-9250)
\whiten\path(9750,-9500)(10250,-9500)(10000,-9250)(9750,-9500)
\blacken\path(9750,-7250)(10250,-7250)(10000,-6750)(9750,-7250)
\blacken\path(9750,8750)(10250,8750)(10000,9250)(9750,8750)

\path(12000,-8000)(12000,10000)
\path(12000,-10250)(12000,-9250)
\whiten\path(11750,-9500)(12250,-9500)(12000,-9250)(11750,-9500)
\blacken\path(11750,-7250)(12250,-7250)(12000,-6750)(11750,-7250)
\blacken\path(11750,8750)(12250,8750)(12000,9250)(11750,8750)

\path(14000,-8000)(14000,10000)
\path(14000,-10250)(14000,-9250)
\whiten\path(13750,-9500)(14250,-9500)(14000,-9250)(13750,-9500)
\blacken\path(13750,-7250)(14250,-7250)(14000,-6750)(13750,-7250)
\blacken\path(13750,8750)(14250,8750)(14000,9250)(13750,8750)

\path(16000,-8000)(16000,10000)
\path(16000,-10250)(16000,-9250)
\whiten\path(15750,-9500)(16250,-9500)(16000,-9250)(15750,-9500)
\blacken\path(15750,-7250)(16250,-7250)(16000,-6750)(15750,-7250)
\blacken\path(15750,8750)(16250,8750)(16000,9250)(15750,8750)

\path(18000,-8000)(18000,10000)
\put(17750,-8750){\tiny$M$}
\path(18000,-10250)(18000,-9250)
\whiten\path(17750,-9500)(18250,-9500)(18000,-9250)(17750,-9500)
\blacken\path(17750,-7250)(18250,-7250)(18000,-6750)(17750,-7250)
\blacken\path(17750,8750)(18250,8750)(18000,9250)(17750,8750)

\end{picture}

\end{minipage}
\end{center}

\caption[Lattice representation of $S_N$]{Lattice representation of $S_N$. The top row of arrows corresponds with the state $|\Uparrow_M\rangle$, while the bottom row of arrows corresponds with the dual state $\langle \Uparrow_M|$. The highest $N$ horizontal lines represent $B$-operators $B(v_j,\{w\}_M)$, while the lowest $N$ horizontal lines represent $C$-operators $C(u_i,\{w\}_M)$.}
\end{figure}
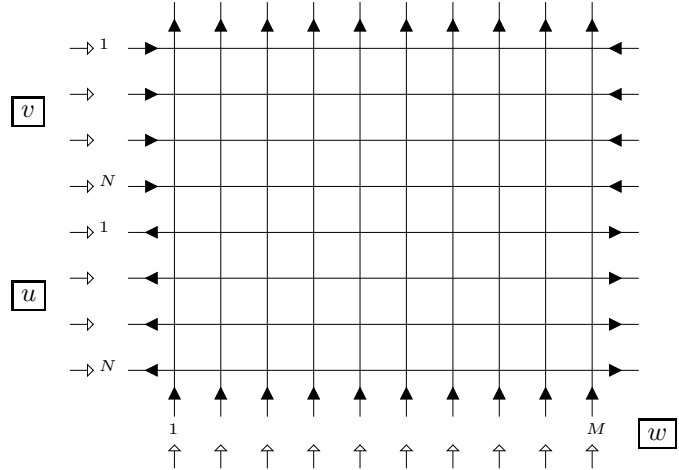

\subsection{Conditions on $S_n(\{u\}_n,\{v\}_N,\{w\}_M)$}

Progressing in the same manner as in \S \ref{s-pfcond}, we will show that the Bethe scalar products $S_{n}( \{u\}_n, \{v\}_N, \{w\}_M )$ satisfy a set of properties which determine them uniquely.

\begin{lemma}
\label{sncond}
{\rm 
We adopt the usual shorthand $S_n = S_n(\{u\}_n,\{v\}_N,\{w\}_M)$. For all $1 \leq n \leq N$ we claim that 

\setcounter{conditions}{0}
\begin{conditions}
{\rm $S_n$ is symmetric in the variables $\{w_{\widetilde{N}+1},\ldots,w_M\}$.}
\end{conditions}

\begin{conditions}
{\rm $S_n$ is a trigonometric polynomial of degree $M-1$ in $u_n$, with zeros occurring at the points $u_n = w_i-\gamma$, for all $1 \leq i \leq \widetilde{N}$.}
\end{conditions}

\begin{conditions}
{\rm Setting $u_n = w_{\widetilde{N}+1}$, $S_n$ satisfies the recursion relation
\begin{align}
&
S_{n}
\Big|_{u_n = w_{\widetilde{N}+1}}
=
\prod_{i=1}^{M}
[w_{\widetilde{N}+1}-w_i+\gamma]
S_{n-1},
\label{sNrec0-xxz}
\end{align}
where $S_{n-1}$ denotes the Bethe scalar product $S_{n-1}(\{u\}_{n-1},\{v\}_N,\{w\}_M)$.
}
\end{conditions}

In addition, we have the supplementary condition 
\begin{conditions}
{\rm 
$S_0$ and $Z_N$ are related via the equation
\begin{align}
S_{0} \Big(\{v\}_N,\{w\}_M \Big)
=
\prod_{i=1}^{N} 
\prod_{j=N+1}^{M}
[v_i-w_j]
Z_N \Big( \{v\}_N,\{w\}_{N} \Big),
\label{cond4}
\end{align}
where we have defined $\{w\}_{N} = \{w_1,\ldots,w_N\}$.}
\end{conditions}
}
\end{lemma}

\begin{proof}

\setcounter{conditions}{0}
\begin{conditions}
{\rm We introduce the auxiliary state vectors
\begin{align}
\langle \Uparrow^{a}_N|
=
\bigotimes_{i=1}^{N} \uparrow^{*}_{a_i},
\quad
\langle \Downarrow^{b}_n|
=
\bigotimes_{i=1}^{n} \downarrow^{*}_{b_i},
\quad
|\Downarrow^{a}_N\rangle
=
\bigotimes_{i=1}^{N} \downarrow_{a_i},
\quad
|\Uparrow^{b}_{n}\rangle
=
\bigotimes_{i=1}^{n} \uparrow_{b_i}
\label{auxvec}
\end{align}
which allow us to write
\begin{align}
S_n
=
\langle \Downarrow_{\widetilde{N}/M}|
\otimes
\langle\Uparrow_{N}^{a}|
\otimes
\langle\Downarrow_n^{b}| 
T\Big(
\{v\}_N \cup \{u\}_n,\{w\}_M
\Big)
|\Uparrow_n^{b}\rangle
\otimes
|\Downarrow_N^{a}\rangle
\otimes
|\Uparrow_M\rangle,
\end{align}
where we have defined
\begin{align}
& T\Big(\{v\}_N \cup \{u\}_{n},\{w\}_M\Big)
=
\\
&
T_{b_{n}}(u_n,\{w\}_M)\ldots T_{b_{1}}(u_1,\{w\}_M)
T_{a_N}(v_N,\{w\}_M) \ldots T_{a_1}(v_1,\{w\}_M).
\nonumber
\end{align}
By application of lemma 4 we thus obtain
\begin{align}
T\Big(\{v\}_N \cup \{u\}_{n},\{w\}_M\Big)
=
(-)^{M\widetilde{N}}
\overline{T}_1(w_1,\{\bar{v}\}_N \cup \{\bar{u}\}_n)
\ldots
\overline{T}_M(w_M,\{\bar{v}\}_N \cup \{\bar{u}\}_n),
\end{align}
where for all $1 \leq i \leq M$ we have set
\begin{align}
\overline{T}_i(w_i,\{\bar{v}\}_N \cup \{\bar{u}\}_n)
=
\left(
\begin{array}{rr}
D(w_i,\{\bar{v}\}_N \cup \{\bar{u}\}_n)
&
-B(w_i,\{\bar{v}\}_N \cup \{\bar{u}\}_n)
\\
-C(w_i,\{\bar{v}\}_N \cup \{\bar{u}\}_n)
&
A(w_i,\{\bar{v}\}_N \cup \{\bar{u}\}_n)
\end{array}
\right)_i
\end{align}
with $\{\bar{v}\}_N \cup \{\bar{u}\}_n = \{v_1+\gamma,\ldots,v_N+\gamma,u_1+\gamma,\ldots,u_n+\gamma\}$. Finally, contracting on the quantum spaces $V_1,\ldots,V_M$ gives
\begin{align}
&S_n\Big(\{u\}_n,\{v\}_N,\{w\}_M\Big)
=
\label{preceq}
\\
&
(-)^{(M+1)\widetilde{N}}
\langle \Uparrow_N^a|
\otimes
\langle \Downarrow_n^b|
\prod_{i=1}^{\widetilde{N}}
C(w_i,\{\bar{v}\}_N\cup\{\bar{u}\}_n)
\prod_{j=\widetilde{N}+1}^{M}
D(w_j,\{\bar{v}\}_N\cup\{\bar{u}\}_n)
|\Uparrow_n^b\rangle
\otimes
|\Downarrow_N^a\rangle.
\nonumber
\end{align}
The diagrammatic interpretation of (\ref{preceq}) is shown in figure \ref{altgraph}.

\begin{figure}[H]

\begin{center}
\begin{minipage}{4.3in}

\setlength{\unitlength}{0.0003cm}
\begin{picture}(20000,23000)(-11000,-9000)

\put(-8000,3000){\fbox{$w$}}

%%%%%%%%

\path(-2000,-6000)(18000,-6000)
\put(-4350,-6000){\tiny$1$}
\path(-5500,-6000)(-4500,-6000)
\whiten\path(-4750,-5750)(-4750,-6250)(-4500,-6000)(-4750,-5750)
\blacken\path(-750,-5750)(-750,-6250)(-1250,-6000)(-750,-5750)
\blacken\path(16750,-5750)(16750,-6250)(17250,-6000)(16750,-5750)

\path(-2000,-4000)(18000,-4000)
\path(-5500,-4000)(-4500,-4000)
\whiten\path(-4750,-3750)(-4750,-4250)(-4500,-4000)(-4750,-3750)
\blacken\path(-750,-3750)(-750,-4250)(-1250,-4000)(-750,-3750)
\blacken\path(16750,-3750)(16750,-4250)(17250,-4000)(16750,-3750)

\path(-2000,-2000)(18000,-2000)
\put(-4350,-2000){\tiny$\widetilde{N}$}
\path(-5500,-2000)(-4500,-2000)
\whiten\path(-4750,-1750)(-4750,-2250)(-4500,-2000)(-4750,-1750)
\blacken\path(-750,-1750)(-750,-2250)(-1250,-2000)(-750,-1750)
\blacken\path(16750,-1750)(16750,-2250)(17250,-2000)(16750,-1750)

\path(-2000,0)(18000,0)
\put(-4350,0){\tiny$\widetilde{N}+1$}
\path(-5500,0)(-4500,0)
\whiten\path(-4750,250)(-4750,-250)(-4500,0)(-4750,250)
\blacken\path(-750,250)(-750,-250)(-1250,0)(-750,250)
\blacken\path(17250,250)(17250,-250)(16750,0)(17250,250)

\path(-2000,2000)(18000,2000)
\path(-5500,2000)(-4500,2000)
\whiten\path(-4750,2250)(-4750,1750)(-4500,2000)(-4750,2250)
\blacken\path(-750,2250)(-750,1750)(-1250,2000)(-750,2250)
\blacken\path(17250,2250)(17250,1750)(16750,2000)(17250,2250)

\path(-2000,4000)(18000,4000)
\path(-5500,4000)(-4500,4000)
\whiten\path(-4750,4250)(-4750,3750)(-4500,4000)(-4750,4250)
\blacken\path(-750,4250)(-750,3750)(-1250,4000)(-750,4250)
\blacken\path(17250,4250)(17250,3750)(16750,4000)(17250,4250)

\path(-2000,6000)(18000,6000)
\path(-5500,6000)(-4500,6000)
\whiten\path(-4750,6250)(-4750,5750)(-4500,6000)(-4750,6250)
\blacken\path(-750,6250)(-750,5750)(-1250,6000)(-750,6250)
\blacken\path(17250,6250)(17250,5750)(16750,6000)(17250,6250)

\path(-2000,8000)(18000,8000)
\path(-5500,8000)(-4500,8000)
\whiten\path(-4750,8250)(-4750,7750)(-4500,8000)(-4750,8250)
\blacken\path(-750,8250)(-750,7750)(-1250,8000)(-750,8250)
\blacken\path(17250,8250)(17250,7750)(16750,8000)(17250,8250)

\path(-2000,10000)(18000,10000)
\path(-5500,10000)(-4500,10000)
\whiten\path(-4750,10250)(-4750,9750)(-4500,10000)(-4750,10250)
\blacken\path(-750,10250)(-750,9750)(-1250,10000)(-750,10250)
\blacken\path(17250,10250)(17250,9750)(16750,10000)(17250,10250)

\path(-2000,12000)(18000,12000)
\put(-4350,12000){\tiny$M$}
\path(-5500,12000)(-4500,12000)
\whiten\path(-4750,12250)(-4750,11750)(-4500,12000)(-4750,12250)
\blacken\path(-750,12250)(-750,11750)(-1250,12000)(-750,12250)
\blacken\path(17250,12250)(17250,11750)(16750,12000)(17250,12250)

%%%%%%%%%%%

\path(0,-8000)(0,14000)
\put(-250,-8750){\tiny$\bar{v}_1$}
\path(0,-10500)(0,-9500)
\whiten\path(-250,-9750)(250,-9750)(0,-9500)(-250,-9750)
\blacken\path(-250,-7250)(250,-7250)(0,-6750)(-250,-7250)
\blacken\path(-250,13250)(250,13250)(0,12750)(-250,13250)

\path(2000,-8000)(2000,14000)
\path(2000,-10500)(2000,-9500)
\whiten\path(1750,-9750)(2250,-9750)(2000,-9500)(1750,-9750)
\blacken\path(1750,-7250)(2250,-7250)(2000,-6750)(1750,-7250)
\blacken\path(1750,13250)(2250,13250)(2000,12750)(1750,13250)

\path(4000,-8000)(4000,14000)
\path(4000,-10500)(4000,-9500)
\whiten\path(3750,-9750)(4250,-9750)(4000,-9500)(3750,-9750)
\blacken\path(3750,-7250)(4250,-7250)(4000,-6750)(3750,-7250)
\blacken\path(3750,13250)(4250,13250)(4000,12750)(3750,13250)

\path(6000,-8000)(6000,14000)
\path(6000,-10500)(6000,-9500)
\whiten\path(5750,-9750)(6250,-9750)(6000,-9500)(5750,-9750)
\blacken\path(5750,-7250)(6250,-7250)(6000,-6750)(5750,-7250)
\blacken\path(5750,13250)(6250,13250)(6000,12750)(5750,13250)

\path(8000,-8000)(8000,14000)
\path(8000,-10500)(8000,-9500)
\whiten\path(7750,-9750)(8250,-9750)(8000,-9500)(7750,-9750)
\blacken\path(7750,-7250)(8250,-7250)(8000,-6750)(7750,-7250)
\blacken\path(7750,13250)(8250,13250)(8000,12750)(7750,13250)

\path(10000,-8000)(10000,14000)
\put(9750,-8750){\tiny$\bar{v}_N$}
\path(10000,-10500)(10000,-9500)
\whiten\path(9750,-9750)(10250,-9750)(10000,-9500)(9750,-9750)
\blacken\path(9750,-7250)(10250,-7250)(10000,-6750)(9750,-7250)
\blacken\path(9750,13250)(10250,13250)(10000,12750)(9750,13250)

\path(12000,-8000)(12000,14000)
\put(11750,-8750){\tiny$\bar{u}_1$}
\path(12000,-10500)(12000,-9500)
\whiten\path(11750,-9750)(12250,-9750)(12000,-9500)(11750,-9750)
\blacken\path(11750,-6750)(12250,-6750)(12000,-7250)(11750,-6750)
\blacken\path(11750,12750)(12250,12750)(12000,13250)(11750,12750)

\path(14000,-8000)(14000,14000)
\path(14000,-10500)(14000,-9500)
\whiten\path(13750,-9750)(14250,-9750)(14000,-9500)(13750,-9750)
\blacken\path(13750,-6750)(14250,-6750)(14000,-7250)(13750,-6750)
\blacken\path(13750,12750)(14250,12750)(14000,13250)(13750,12750)

\path(16000,-8000)(16000,14000)
\put(15750,-8750){\tiny$\bar{u}_n$}
\path(16000,-10500)(16000,-9500)
\whiten\path(15750,-9750)(16250,-9750)(16000,-9500)(15750,-9750)
\blacken\path(15750,-6750)(16250,-6750)(16000,-7250)(15750,-6750)
\blacken\path(15750,12750)(16250,12750)(16000,13250)(15750,12750)

\end{picture}

\end{minipage}
\end{center}

\caption[Alternative graphical representation of $S_n$]{Alternative graphical representation of $S_n$. Neglecting an overall minus sign, $S_n$ is equal to this lattice, which is essentially a rotation of figure \ref{lat}. The top row of arrows represents the state $|\Downarrow_N^a\rangle \otimes |\Uparrow_n^b\rangle$, while the bottom row of arrows represents the dual state $\langle \Uparrow_N^a| \otimes \langle \Downarrow_n^b|$. The lowest $\widetilde{N}$ horizontal lines represent $C$-operators, with the remaining lines representing $D$-operators.}

\label{altgraph}

\end{figure}
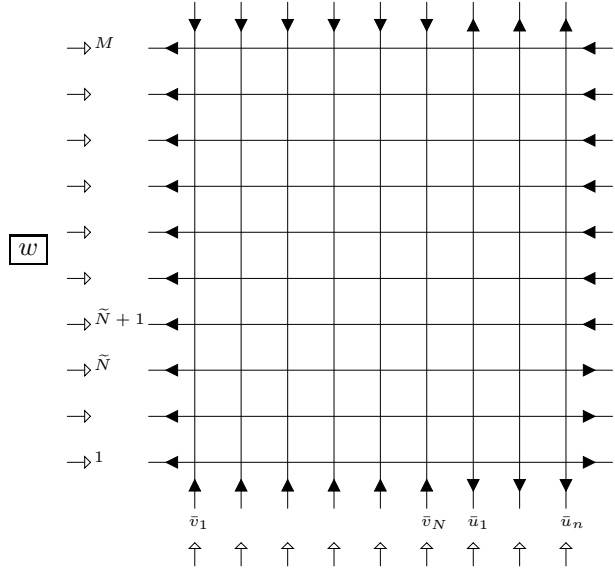

By equation (\ref{DD}) the $D$-operators in (\ref{preceq}) commute, proving that $S_n$ is symmetric in the variables $\{w_{\widetilde{N}+1},\ldots,w_M\}$.
} 
\end{conditions}

\begin{conditions}
{\rm
Inserting the set of states $\sum_{m > \widetilde{N}}
\sigma_m^{-} |\Downarrow_{\widetilde{N}/M}\rangle \langle \Downarrow_{\widetilde{N}/M}| \sigma_m^{+}$ after the first $C$-operator appearing in (\ref{sn-xxz}), we obtain the expansion
\begin{align}
&S_n\Big(\{u\}_n,\{v\}_N,\{w\}_M\Big) 
=
\label{exp100}
\\
&
\sum_{m>\widetilde{N}}
\langle \Downarrow_{\widetilde{N}/M}|
C(u_n,\{w\}_M)
\sigma_m^{-} 
|\Downarrow_{\widetilde{N}/M}\rangle
\langle \Downarrow_{\widetilde{N}/M} |
\sigma_m^{+} 
\prod_{i=1}^{n-1}
C(u_i,\{w\}_M)
\prod_{j=1}^{N}
B(v_j,\{w\}_M)
|\Uparrow_M\rangle
\nonumber
\end{align}
in which all dependence on $u_n$ appears in the first factor of the sum. We therefore wish to calculate 
$
\langle \Downarrow_{\widetilde{N}/M}|
C(u_n,\{w\}_M)
\sigma_m^{-} 
|\Downarrow_{\widetilde{N}/M}\rangle
$
for all $\widetilde{N} < m \leq M$, and do so by identifying it with the string of vertices shown below:

\begin{figure}[H]

\begin{center}
\begin{minipage}{4.3in}

\setlength{\unitlength}{0.0003cm}
\begin{picture}(20000,13000)(-9000,-9000)

\path(-2000,2000)(20000,2000)
\put(-3250,2000){\tiny$u_n$}
\path(-4500,2000)(-3500,2000)
\whiten\path(-3750,2250)(-3750,1750)(-3500,2000)(-3750,2250)
\blacken\path(-750,2250)(-750,1750)(-1250,2000)(-750,2250)
\blacken\path(18750,2250)(18750,1750)(19250,2000)(18750,2250)

%%%%%%%%%%%

\path(0,000)(0,4000)
\put(-250,-1000){\tiny$1$}
\path(0,-2500)(0,-1500)
\whiten\path(-250,-1750)(250,-1750)(0,-1500)(-250,-1750)
\blacken\path(-250,1250)(250,1250)(0,750)(-250,1250)
\blacken\path(-250,3250)(250,3250)(0,2750)(-250,3250)

\path(2000,000)(2000,4000)
\path(2000,-2500)(2000,-1500)
\whiten\path(1750,-1750)(2250,-1750)(2000,-1500)(1750,-1750)
\blacken\path(1750,1250)(2250,1250)(2000,750)(1750,1250)
\blacken\path(1750,3250)(2250,3250)(2000,2750)(1750,3250)

\path(4000,000)(4000,4000)
\put(3750,-1000){\tiny$\widetilde{N}$}
\path(4000,-2500)(4000,-1500)
\whiten\path(3750,-1750)(4250,-1750)(4000,-1500)(3750,-1750)
\blacken\path(3750,1250)(4250,1250)(4000,750)(3750,1250)
\blacken\path(3750,3250)(4250,3250)(4000,2750)(3750,3250)

\path(6000,000)(6000,4000)
\put(5750,-1000){\tiny$\widetilde{N}+1$}
\path(6000,-2500)(6000,-1500)
\whiten\path(5750,-1750)(6250,-1750)(6000,-1500)(5750,-1750)
\blacken\path(5750,750)(6250,750)(6000,1250)(5750,750)
\blacken\path(5750,2750)(6250,2750)(6000,3250)(5750,2750)

\path(8000,000)(8000,4000)
\path(8000,-2500)(8000,-1500)
\whiten\path(7750,-1750)(8250,-1750)(8000,-1500)(7750,-1750)
\blacken\path(7750,750)(8250,750)(8000,1250)(7750,750)
\blacken\path(7750,2750)(8250,2750)(8000,3250)(7750,2750)

\path(10000,000)(10000,4000)
\path(10000,-2500)(10000,-1500)
\whiten\path(9750,-1750)(10250,-1750)(10000,-1500)(9750,-1750)
\blacken\path(9750,750)(10250,750)(10000,1250)(9750,750)
\blacken\path(9750,2750)(10250,2750)(10000,3250)(9750,2750)

\path(12000,000)(12000,4000)
\put(11750,-1000){\tiny$m$}
\path(12000,-2500)(12000,-1500)
\whiten\path(11750,-1750)(12250,-1750)(12000,-1500)(11750,-1750)
\blacken\path(11750,750)(12250,750)(12000,1250)(11750,750)
\blacken\path(11750,3250)(12250,3250)(12000,2750)(11750,3250)

\path(14000,000)(14000,4000)
\path(14000,-2500)(14000,-1500)
\whiten\path(13750,-1750)(14250,-1750)(14000,-1500)(13750,-1750)
\blacken\path(13750,750)(14250,750)(14000,1250)(13750,750)
\blacken\path(13750,2750)(14250,2750)(14000,3250)(13750,2750)

\path(16000,000)(16000,4000)
\path(16000,-2500)(16000,-1500)
\whiten\path(15750,-1750)(16250,-1750)(16000,-1500)(15750,-1750)
\blacken\path(15750,750)(16250,750)(16000,1250)(15750,750)
\blacken\path(15750,2750)(16250,2750)(16000,3250)(15750,2750)

\path(18000,000)(18000,4000)
\put(17750,-1000){\tiny$M$}
\path(18000,-2500)(18000,-1500)
\whiten\path(17750,-1750)(18250,-1750)(18000,-1500)(17750,-1750)
\blacken\path(17750,750)(18250,750)(18000,1250)(17750,750)
\blacken\path(17750,2750)(18250,2750)(18000,3250)(17750,2750)

%%%%%%%%%%

\put(20000,-1000){\fbox{$w$}}

%%%%%2nd lattice%%%%%%%%%%
\put(-7000,-6250){$=$}

\path(-2000,-6000)(20000,-6000)
\put(-3250,-6000){\tiny$u_n$}
\path(-4500,-6000)(-3500,-6000)
\whiten\path(-3750,-5750)(-3750,-6250)(-3500,-6000)(-3750,-5750)
\blacken\path(-750,-5750)(-750,-6250)(-1250,-6000)(-750,-5750)
\blacken\path(1250,-5750)(1250,-6250)(750,-6000)(1250,-5750)
\blacken\path(3250,-5750)(3250,-6250)(2750,-6000)(3250,-5750)
\blacken\path(5250,-5750)(5250,-6250)(4750,-6000)(5250,-5750)
\blacken\path(7250,-5750)(7250,-6250)(6750,-6000)(7250,-5750)
\blacken\path(9250,-5750)(9250,-6250)(8750,-6000)(9250,-5750)
\blacken\path(11250,-5750)(11250,-6250)(10750,-6000)(11250,-5750)
\blacken\path(12750,-5750)(12750,-6250)(13250,-6000)(12750,-5750)
\blacken\path(14750,-5750)(14750,-6250)(15250,-6000)(14750,-5750)
\blacken\path(16750,-5750)(16750,-6250)(17250,-6000)(16750,-5750)
\blacken\path(18750,-5750)(18750,-6250)(19250,-6000)(18750,-5750)

%%%%%%%%%%%

\path(0,-8000)(0,-4000)
\put(-250,-9000){\tiny$1$}
\path(0,-10500)(0,-9500)
\whiten\path(-250,-9750)(250,-9750)(0,-9500)(-250,-9750)
\blacken\path(-250,-6750)(250,-6750)(0,-7250)(-250,-6750)
\blacken\path(-250,-4750)(250,-4750)(0,-5250)(-250,-4750)

\path(2000,-8000)(2000,-4000)
\path(2000,-10500)(2000,-9500)
\whiten\path(1750,-9750)(2250,-9750)(2000,-9500)(1750,-9750)
\blacken\path(1750,-6750)(2250,-6750)(2000,-7250)(1750,-6750)
\blacken\path(1750,-4750)(2250,-4750)(2000,-5250)(1750,-4750)

\path(4000,-8000)(4000,-4000)
\put(3750,-9000){\tiny$\widetilde{N}$}
\path(4000,-10500)(4000,-9500)
\whiten\path(3750,-9750)(4250,-9750)(4000,-9500)(3750,-9750)
\blacken\path(3750,-6750)(4250,-6750)(4000,-7250)(3750,-6750)
\blacken\path(3750,-4750)(4250,-4750)(4000,-5250)(3750,-4750)

\path(6000,-8000)(6000,-4000)
\put(5750,-9000){\tiny$\widetilde{N}+1$}
\path(6000,-10500)(6000,-9500)
\whiten\path(5750,-9750)(6250,-9750)(6000,-9500)(5750,-9750)
\blacken\path(5750,-7250)(6250,-7250)(6000,-6750)(5750,-7250)
\blacken\path(5750,-5250)(6250,-5250)(6000,-4750)(5750,-5250)

\path(8000,-8000)(8000,-4000)
\path(8000,-10500)(8000,-9500)
\whiten\path(7750,-9750)(8250,-9750)(8000,-9500)(7750,-9750)
\blacken\path(7750,-7250)(8250,-7250)(8000,-6750)(7750,-7250)
\blacken\path(7750,-5250)(8250,-5250)(8000,-4750)(7750,-5250)

\path(10000,-8000)(10000,-4000)
\path(10000,-10500)(10000,-9500)
\whiten\path(9750,-9750)(10250,-9750)(10000,-9500)(9750,-9750)
\blacken\path(9750,-7250)(10250,-7250)(10000,-6750)(9750,-7250)
\blacken\path(9750,-5250)(10250,-5250)(10000,-4750)(9750,-5250)

\path(12000,-8000)(12000,-4000)
\put(11750,-9000){\tiny$m$}
\path(12000,-10500)(12000,-9500)
\whiten\path(11750,-9750)(12250,-9750)(12000,-9500)(11750,-9750)
\blacken\path(11750,-7250)(12250,-7250)(12000,-6750)(11750,-7250)
\blacken\path(11750,-4750)(12250,-4750)(12000,-5250)(11750,-4750)

\path(14000,-8000)(14000,-4000)
\path(14000,-10500)(14000,-9500)
\whiten\path(13750,-9750)(14250,-9750)(14000,-9500)(13750,-9750)
\blacken\path(13750,-7250)(14250,-7250)(14000,-6750)(13750,-7250)
\blacken\path(13750,-5250)(14250,-5250)(14000,-4750)(13750,-5250)

\path(16000,-8000)(16000,-4000)
\path(16000,-10500)(16000,-9500)
\whiten\path(15750,-9750)(16250,-9750)(16000,-9500)(15750,-9750)
\blacken\path(15750,-7250)(16250,-7250)(16000,-6750)(15750,-7250)
\blacken\path(15750,-5250)(16250,-5250)(16000,-4750)(15750,-5250)

\path(18000,-8000)(18000,-4000)
\put(17750,-9000){\tiny$M$}
\path(18000,-10500)(18000,-9500)
\whiten\path(17750,-9750)(18250,-9750)(18000,-9500)(17750,-9750)
\blacken\path(17750,-7250)(18250,-7250)(18000,-6750)(17750,-7250)
\blacken\path(17750,-5250)(18250,-5250)(18000,-4750)(17750,-5250)

%%%%%%%

\put(20000,-9000){\fbox{$w$}}

\end{picture}

\end{minipage}
\end{center}

\caption[Peeling away the bottom row of $S_n$]{Peeling away the bottom row of $S_n$. The upper diagram represents $\langle \Downarrow_{\widetilde{N}/M}| C(u_n,\{w\}_M) \sigma_m^{-} |\Downarrow_{\widetilde{N}/M}\rangle$, with the internal black arrows being summed over all configurations. The lower diagram represents the only surviving configuration.}

\label{svert100}

\end{figure}
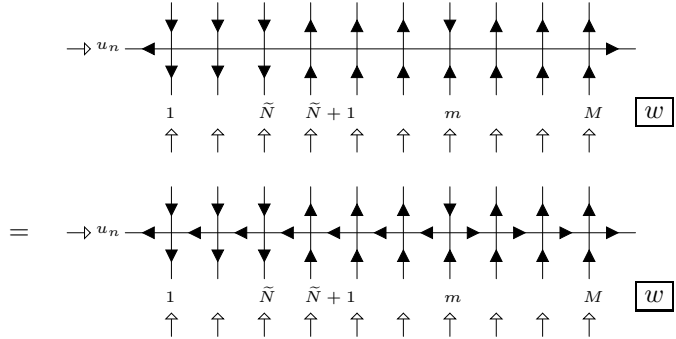

Replacing each vertex in figure \ref{svert100} with its corresponding trigonometric weight, we find that
\begin{align}
&
\langle \Downarrow_{\widetilde{N}/M}|
C(u_n,\{w\}_M)
\sigma_m^{-} 
|\Downarrow_{\widetilde{N}/M}\rangle
=
\label{factor100}
\prod_{i=1}^{\widetilde{N}}
[u_n - w_i +\gamma]
\prod_{\widetilde{N} < i < m}
[u_n -w_i]
[\gamma]
\prod_{m < i \leq M}
[u_n -w_i +\gamma].
\end{align}
Substituting (\ref{factor100}) into the expansion (\ref{exp100}) gives
\begin{align}
S_n
&=
[\gamma] \prod_{i=1}^{\widetilde{N}} [u_n-w_i+\gamma]
\sum_{m > \widetilde{N}}
\prod_{\widetilde{N} < i < m}
[u_n-w_i]
\prod_{m < i \leq M}
[u_n-w_i+\gamma]
\label{snexp-xxz}
\\
&
\times
\langle \Downarrow_{\widetilde{N}/M}|
\sigma_m^{+}
\prod_{i=1}^{n-1}
C(u_i,\{w\}_M)
\prod_{j=1}^{N}
B(v_j,\{w\}_M)
|\Uparrow_M\rangle.
\nonumber
\end{align}
From the last expression it is apparent that $S_n$ is a trigonometric polynomial of degree $M-1$ in $u_n$. Furthermore, $\widetilde{N}$ of the zeros of this polynomial are contained in the factor $\prod_{i=1}^{\widetilde{N}} [u_n-w_i+\gamma]$. 
}
\end{conditions}

\begin{conditions}
{\rm
Setting $u_n = w_{\widetilde{N}+1}$ in equation (\ref{snexp-xxz}), all terms in the sum collapse to zero except the term corresponding to $m = \widetilde{N}+1$ and we obtain
\begin{align}
S_n
\Big|_{u_n = w_{\widetilde{N}+1}}
&=
\prod_{i = 1}^{M}
[w_{\widetilde{N}+1}-w_i+\gamma]
\langle \Downarrow_{\widetilde{N}+1/M}|
\prod_{i=1}^{n-1} 
C(u_{i},\{w\}_M)
\prod_{j=1}^{N}
B(v_j,\{w\}_M)
|\Uparrow_M\rangle,
\label{recproof}
\\
&=
\prod_{i=1}^{M}
[w_{\widetilde{N}+1}-w_i+\gamma]
S_{n-1} \Big( \{u\}_{n-1},\{v\}_N,\{w\}_M \Big),
\nonumber
\end{align}
where we have performed the trivial rearrangements $[\gamma] = [w_{\widetilde{N}+1}-w_{\widetilde{N}+1}+\gamma]$ as well as $\langle \Downarrow_{\widetilde{N}/M}| \sigma_{\widetilde{N}+1}^{+} = \langle \Downarrow_{\widetilde{N}+1/M}|$ to produce the first line of (\ref{recproof}), while the second line follows directly from the definition of $S_{n-1}$. Hence we have proved the recursion relation (\ref{sNrec0-xxz}). The diagrammatic interpretation of this identity is given below.

\begin{figure}[H]

\begin{center}
\begin{minipage}{4.3in}

\setlength{\unitlength}{0.0003cm}
\begin{picture}(20000,21000)(-9000,-9000)

\put(-7000,-4000){\fbox{$u$}}

%%%%%%

\path(-2000,-6000)(20000,-6000)
\put(-3250,-6000){\tiny$n$}
\path(-4500,-6000)(-3500,-6000)
\whiten\path(-3750,-5750)(-3750,-6250)(-3500,-6000)(-3750,-5750)
\blacken\path(-750,-5750)(-750,-6250)(-1250,-6000)(-750,-5750)
\blacken\path(1250,-5750)(1250,-6250)(750,-6000)(1250,-5750)
\blacken\path(3250,-5750)(3250,-6250)(2750,-6000)(3250,-5750)
\blacken\path(5250,-5750)(5250,-6250)(4750,-6000)(5250,-5750)
\blacken\path(6750,-5750)(6750,-6250)(7250,-6000)(6750,-5750)
\blacken\path(8750,-5750)(8750,-6250)(9250,-6000)(8750,-5750)
\blacken\path(10750,-5750)(10750,-6250)(11250,-6000)(10750,-5750)
\blacken\path(12750,-5750)(12750,-6250)(13250,-6000)(12750,-5750)
\blacken\path(14750,-5750)(14750,-6250)(15250,-6000)(14750,-5750)
\blacken\path(16750,-5750)(16750,-6250)(17250,-6000)(16750,-5750)
\blacken\path(18750,-5750)(18750,-6250)(19250,-6000)(18750,-5750)

\path(-2000,-4000)(20000,-4000)
\path(-4500,-4000)(-3500,-4000)
\whiten\path(-3750,-3750)(-3750,-4250)(-3500,-4000)(-3750,-3750)
\blacken\path(-750,-3750)(-750,-4250)(-1250,-4000)(-750,-3750)
\blacken\path(18750,-3750)(18750,-4250)(19250,-4000)(18750,-3750)

\path(-2000,-2000)(20000,-2000)
\put(-3250,-2000){\tiny$1$}
\path(-4500,-2000)(-3500,-2000)
\whiten\path(-3750,-1750)(-3750,-2250)(-3500,-2000)(-3750,-1750)
\blacken\path(-750,-1750)(-750,-2250)(-1250,-2000)(-750,-1750)
\blacken\path(18750,-1750)(18750,-2250)(19250,-2000)(18750,-1750)

%%%%%%%%%%%

\put(-7000,5000){\fbox{$v$}}

%%%%%%%%%%

\path(-2000,0)(20000,0)
\put(-3250,0){\tiny$N$}
\path(-4500,0)(-3500,0)
\whiten\path(-3750,250)(-3750,-250)(-3500,0)(-3750,250)
\blacken\path(-1250,250)(-1250,-250)(-750,0)(-1250,250)
\blacken\path(19250,250)(19250,-250)(18750,0)(19250,250)

\path(-2000,2000)(20000,2000)
\path(-4500,2000)(-3500,2000)
\whiten\path(-3750,2250)(-3750,1750)(-3500,2000)(-3750,2250)
\blacken\path(-1250,2250)(-1250,1750)(-750,2000)(-1250,2250)
\blacken\path(19250,2250)(19250,1750)(18750,2000)(19250,2250)

\path(-2000,4000)(20000,4000)
\path(-4500,4000)(-3500,4000)
\whiten\path(-3750,4250)(-3750,3750)(-3500,4000)(-3750,4250)
\blacken\path(-1250,4250)(-1250,3750)(-750,4000)(-1250,4250)
\blacken\path(19250,4250)(19250,3750)(18750,4000)(19250,4250)

\path(-2000,6000)(20000,6000)
\path(-4500,6000)(-3500,6000)
\whiten\path(-3750,6250)(-3750,5750)(-3500,6000)(-3750,6250)
\blacken\path(-1250,6250)(-1250,5750)(-750,6000)(-1250,6250)
\blacken\path(19250,6250)(19250,5750)(18750,6000)(19250,6250)

\path(-2000,8000)(20000,8000)
\path(-4500,8000)(-3500,8000)
\whiten\path(-3750,8250)(-3750,7750)(-3500,8000)(-3750,8250)
\blacken\path(-1250,8250)(-1250,7750)(-750,8000)(-1250,8250)
\blacken\path(19250,8250)(19250,7750)(18750,8000)(19250,8250)

\path(-2000,10000)(20000,10000)
\put(-3250,10000){\tiny$1$}
\path(-4500,10000)(-3500,10000)
\whiten\path(-3750,10250)(-3750,9750)(-3500,10000)(-3750,10250)
\blacken\path(-1250,10250)(-1250,9750)(-750,10000)(-1250,10250)
\blacken\path(19250,10250)(19250,9750)(18750,10000)(19250,10250)

%%%%%%%%%%%

\put(20000,-9000){\fbox{$w$}}

%%%%%%%%%%%%

\path(0,-8000)(0,12000)
\put(-250,-9000){\tiny$1$}
\path(0,-10250)(0,-9250)
\whiten\path(-250,-9500)(250,-9500)(0,-9250)(-250,-9500)
\blacken\path(-250,-6750)(250,-6750)(0,-7250)(-250,-6750)
\blacken\path(-250,-4750)(250,-4750)(0,-5250)(-250,-4750)
\blacken\path(-250,10750)(250,10750)(0,11250)(-250,10750)

\path(2000,-8000)(2000,12000)
\path(2000,-10250)(2000,-9250)
\whiten\path(1750,-9500)(2250,-9500)(2000,-9250)(1750,-9500)
\blacken\path(1750,-6750)(2250,-6750)(2000,-7250)(1750,-6750)
\blacken\path(1750,-4750)(2250,-4750)(2000,-5250)(1750,-4750)
\blacken\path(1750,10750)(2250,10750)(2000,11250)(1750,10750)

\path(4000,-8000)(4000,12000)
\put(3750,-9000){\tiny$\widetilde{N}$}
\path(4000,-10250)(4000,-9250)
\whiten\path(3750,-9500)(4250,-9500)(4000,-9250)(3750,-9500)
\blacken\path(3750,-6750)(4250,-6750)(4000,-7250)(3750,-6750)
\blacken\path(3750,-4750)(4250,-4750)(4000,-5250)(3750,-4750)
\blacken\path(3750,10750)(4250,10750)(4000,11250)(3750,10750)

\path(6000,-8000)(6000,12000)
\put(5750,-9000){\tiny$\widetilde{N}+1$}
\path(6000,-10250)(6000,-9250)
\whiten\path(5750,-9500)(6250,-9500)(6000,-9250)(5750,-9500)
\blacken\path(5750,-7250)(6250,-7250)(6000,-6750)(5750,-7250)
\blacken\path(5750,-4750)(6250,-4750)(6000,-5250)(5750,-4750)
\blacken\path(5750,10750)(6250,10750)(6000,11250)(5750,10750)

\path(8000,-8000)(8000,12000)
\path(8000,-10250)(8000,-9250)
\whiten\path(7750,-9500)(8250,-9500)(8000,-9250)(7750,-9500)
\blacken\path(7750,-7250)(8250,-7250)(8000,-6750)(7750,-7250)
\blacken\path(7750,-5250)(8250,-5250)(8000,-4750)(7750,-5250)
\blacken\path(7750,10750)(8250,10750)(8000,11250)(7750,10750)

\path(10000,-8000)(10000,12000)
\path(10000,-10250)(10000,-9250)
\whiten\path(9750,-9500)(10250,-9500)(10000,-9250)(9750,-9500)
\blacken\path(9750,-7250)(10250,-7250)(10000,-6750)(9750,-7250)
\blacken\path(9750,-5250)(10250,-5250)(10000,-4750)(9750,-5250)
\blacken\path(9750,10750)(10250,10750)(10000,11250)(9750,10750)

\path(12000,-8000)(12000,12000)
\path(12000,-10250)(12000,-9250)
\whiten\path(11750,-9500)(12250,-9500)(12000,-9250)(11750,-9500)
\blacken\path(11750,-7250)(12250,-7250)(12000,-6750)(11750,-7250)
\blacken\path(11750,-5250)(12250,-5250)(12000,-4750)(11750,-5250)
\blacken\path(11750,10750)(12250,10750)(12000,11250)(11750,10750)

\path(14000,-8000)(14000,12000)
\path(14000,-10250)(14000,-9250)
\whiten\path(13750,-9500)(14250,-9500)(14000,-9250)(13750,-9500)
\blacken\path(13750,-7250)(14250,-7250)(14000,-6750)(13750,-7250)
\blacken\path(13750,-5250)(14250,-5250)(14000,-4750)(13750,-5250)
\blacken\path(13750,10750)(14250,10750)(14000,11250)(13750,10750)

\path(16000,-8000)(16000,12000)
\path(16000,-10250)(16000,-9250)
\whiten\path(15750,-9500)(16250,-9500)(16000,-9250)(15750,-9500)
\blacken\path(15750,-7250)(16250,-7250)(16000,-6750)(15750,-7250)
\blacken\path(15750,-5250)(16250,-5250)(16000,-4750)(15750,-5250)
\blacken\path(15750,10750)(16250,10750)(16000,11250)(15750,10750)

\path(18000,-8000)(18000,12000)
\put(17750,-9000){\tiny$M$}
\path(18000,-10250)(18000,-9250)
\whiten\path(17750,-9500)(18250,-9500)(18000,-9250)(17750,-9500)
\blacken\path(17750,-7250)(18250,-7250)(18000,-6750)(17750,-7250)
\blacken\path(17750,-5250)(18250,-5250)(18000,-4750)(17750,-5250)
\blacken\path(17750,10750)(18250,10750)(18000,11250)(17750,10750)

\end{picture}

\end{minipage}
\end{center}

\caption[Freezing the last row of the $S_n$ lattice]{Freezing the last row of the $S_n$ lattice. In general, the vertex at the intersection of the $u_n$ and $w_{\widetilde{N}+1}$ lines can be of type $b_{-}(u_n,w_{\widetilde{N}+1})$ or $c_{-}(u_n,w_{\widetilde{N}+1})$. Setting $u_n = w_{\widetilde{N}+1}$ causes all configurations with a $b_{-}(u_n,w_{\widetilde{N}+1})$ vertex to vanish, and we are left with a frozen row of vertices as shown. This row of vertices produces the prefactor in (\ref{sNrec0-xxz}), whilst the remainder of the lattice represents $S_{n-1}$. }
\end{figure}
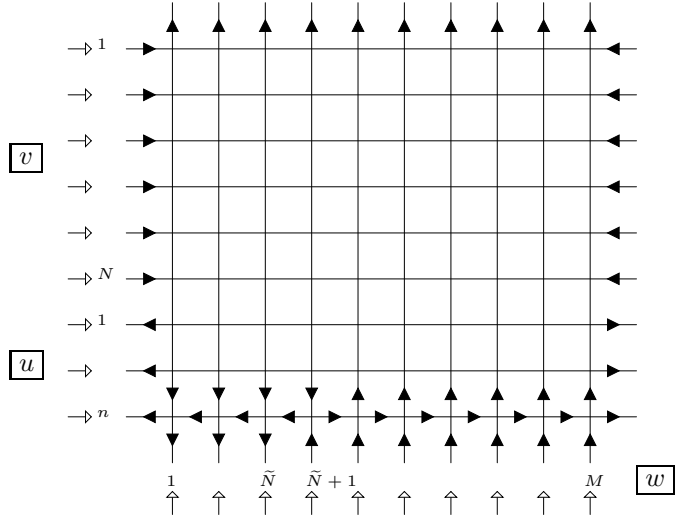

}
\end{conditions}

\begin{conditions}
{\rm From the definition of the state vectors (\ref{auxvec}) we obtain 
\begin{align}
S_{0} \Big(\{v\}_N,\{w\}_M\Big) 
=
\langle \Downarrow_{N/M}| \otimes \langle \Uparrow_N^{a}|
T\Big( \{v\}_{N},\{w\}_M \Big) 
|\Downarrow_N^{a}\rangle \otimes |\Uparrow_M\rangle,  
\end{align}
with $T(\{v\}_N,\{w\}_M)$ given by (\ref{doublemon}). Using lemma 4 and contracting on the quantum spaces $V_1,\ldots,V_M$ gives
\begin{align}
S_{0} \Big(\{v\}_N,\{w\}_M\Big) 
=
(-)^{(M+1)N}
\langle \Uparrow_N^{a} |
\prod_{i=1}^{N}
C(w_i,\{\bar{v}\}_N)
\prod_{j=N+1}^{M}
D(w_{j},\{\bar{v}\}_N)
|\Downarrow_N^{a}\rangle.
\end{align}
Now since $|\Downarrow_N^{a}\rangle$ is an eigenvector of the $D$-operators, as can be seen from equation (\ref{diagonal2}), we have
\begin{align}
S_{0} \Big(\{v\}_N,\{w\}_M\Big) 
=
(-)^{(M+1)N}
\prod_{i=1}^{N}
\prod_{j=N+1}^{M}
[w_j-\bar{v}_i+\gamma]
\langle \Uparrow_N^{a} |
\prod_{i=1}^{N}
C(w_i,\{\bar{v}\}_N) 
|\Downarrow_N^{a}\rangle
\end{align}
or equivalently, substituting $\bar{v}_i = v_i+\gamma$ for all $1 \leq i \leq N$ into the previous equation,
\begin{align}
S_{0} \Big(\{v\}_N,\{w\}_M\Big) 
=
\prod_{i=1}^{N}
\prod_{j=N+1}^{M}
[v_i-w_j]
\langle \Uparrow_N^{a} |
\prod_{i=1}^{N}
C(w_i,\{\bar{v}\}_N) 
|\Downarrow_N^{a}\rangle.
\label{cond4pf}
\end{align}
Comparing with the alternative expression (\ref{zequiv}) for the domain wall partition function, equation (\ref{cond4pf}) completes the proof of (\ref{cond4}). A graphical version of this identity is given below.

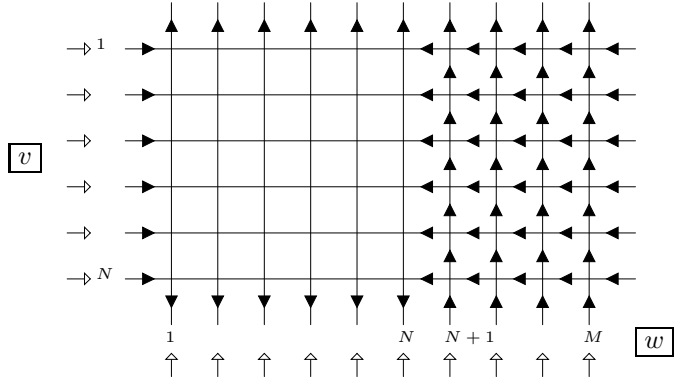
\begin{figure}[H]

\begin{center}
\begin{minipage}{4.3in}

\setlength{\unitlength}{0.0003cm}
\begin{picture}(20000,15000)(-9000,-3000)

\put(-7000,5000){\fbox{$v$}}

%%%%%%%%%%

\path(-2000,0)(20000,0)
\put(-3250,0){\tiny$N$}
\path(-4500,0)(-3500,0)
\whiten\path(-3750,250)(-3750,-250)(-3500,0)(-3750,250)
\blacken\path(-1250,250)(-1250,-250)(-750,0)(-1250,250)
\blacken\path(11250,250)(11250,-250)(10750,0)(11250,250)
\blacken\path(13250,250)(13250,-250)(12750,0)(13250,250)
\blacken\path(15250,250)(15250,-250)(14750,0)(15250,250)
\blacken\path(17250,250)(17250,-250)(16750,0)(17250,250)
\blacken\path(19250,250)(19250,-250)(18750,0)(19250,250)

\path(-2000,2000)(20000,2000)
\path(-4500,2000)(-3500,2000)
\whiten\path(-3750,2250)(-3750,1750)(-3500,2000)(-3750,2250)
\blacken\path(-1250,2250)(-1250,1750)(-750,2000)(-1250,2250)
\blacken\path(11250,2250)(11250,1750)(10750,2000)(11250,2250)
\blacken\path(13250,2250)(13250,1750)(12750,2000)(13250,2250)
\blacken\path(15250,2250)(15250,1750)(14750,2000)(15250,2250)
\blacken\path(17250,2250)(17250,1750)(16750,2000)(17250,2250)
\blacken\path(19250,2250)(19250,1750)(18750,2000)(19250,2250)

\path(-2000,4000)(20000,4000)
\path(-4500,4000)(-3500,4000)
\whiten\path(-3750,4250)(-3750,3750)(-3500,4000)(-3750,4250)
\blacken\path(-1250,4250)(-1250,3750)(-750,4000)(-1250,4250)
\blacken\path(11250,4250)(11250,3750)(10750,4000)(11250,4250)
\blacken\path(13250,4250)(13250,3750)(12750,4000)(13250,4250)
\blacken\path(15250,4250)(15250,3750)(14750,4000)(15250,4250)
\blacken\path(17250,4250)(17250,3750)(16750,4000)(17250,4250)
\blacken\path(19250,4250)(19250,3750)(18750,4000)(19250,4250)

\path(-2000,6000)(20000,6000)
\path(-4500,6000)(-3500,6000)
\whiten\path(-3750,6250)(-3750,5750)(-3500,6000)(-3750,6250)
\blacken\path(-1250,6250)(-1250,5750)(-750,6000)(-1250,6250)
\blacken\path(11250,6250)(11250,5750)(10750,6000)(11250,6250)
\blacken\path(13250,6250)(13250,5750)(12750,6000)(13250,6250)
\blacken\path(15250,6250)(15250,5750)(14750,6000)(15250,6250)
\blacken\path(17250,6250)(17250,5750)(16750,6000)(17250,6250)
\blacken\path(19250,6250)(19250,5750)(18750,6000)(19250,6250)

\path(-2000,8000)(20000,8000)
\path(-4500,8000)(-3500,8000)
\whiten\path(-3750,8250)(-3750,7750)(-3500,8000)(-3750,8250)
\blacken\path(-1250,8250)(-1250,7750)(-750,8000)(-1250,8250)
\blacken\path(11250,8250)(11250,7750)(10750,8000)(11250,8250)
\blacken\path(13250,8250)(13250,7750)(12750,8000)(13250,8250)
\blacken\path(15250,8250)(15250,7750)(14750,8000)(15250,8250)
\blacken\path(17250,8250)(17250,7750)(16750,8000)(17250,8250)
\blacken\path(19250,8250)(19250,7750)(18750,8000)(19250,8250)

\path(-2000,10000)(20000,10000)
\put(-3250,10000){\tiny$1$}
\path(-4500,10000)(-3500,10000)
\whiten\path(-3750,10250)(-3750,9750)(-3500,10000)(-3750,10250)
\blacken\path(-1250,10250)(-1250,9750)(-750,10000)(-1250,10250)
\blacken\path(11250,10250)(11250,9750)(10750,10000)(11250,10250)
\blacken\path(13250,10250)(13250,9750)(12750,10000)(13250,10250)
\blacken\path(15250,10250)(15250,9750)(14750,10000)(15250,10250)
\blacken\path(17250,10250)(17250,9750)(16750,10000)(17250,10250)
\blacken\path(19250,10250)(19250,9750)(18750,10000)(19250,10250)

%%%%%%%%%%%

\path(0,-2000)(0,12000)
\put(-250,-2750){\tiny$1$}
\path(0,-4250)(0,-3250)
\whiten\path(-250,-3500)(250,-3500)(0,-3250)(-250,-3500)
\blacken\path(-250,-750)(250,-750)(0,-1250)(-250,-750)
\blacken\path(-250,10750)(250,10750)(0,11250)(-250,10750)

\path(2000,-2000)(2000,12000)
\path(2000,-4250)(2000,-3250)
\whiten\path(1750,-3500)(2250,-3500)(2000,-3250)(1750,-3500)
\blacken\path(1750,-750)(2250,-750)(2000,-1250)(1750,-750)
\blacken\path(1750,10750)(2250,10750)(2000,11250)(1750,10750)

\path(4000,-2000)(4000,12000)
\path(4000,-4250)(4000,-3250)
\whiten\path(3750,-3500)(4250,-3500)(4000,-3250)(3750,-3500)
\blacken\path(3750,-750)(4250,-750)(4000,-1250)(3750,-750)
\blacken\path(3750,10750)(4250,10750)(4000,11250)(3750,10750)

\path(6000,-2000)(6000,12000)
\path(6000,-4250)(6000,-3250)
\whiten\path(5750,-3500)(6250,-3500)(6000,-3250)(5750,-3500)
\blacken\path(5750,-750)(6250,-750)(6000,-1250)(5750,-750)
\blacken\path(5750,10750)(6250,10750)(6000,11250)(5750,10750)

\path(8000,-2000)(8000,12000)
\path(8000,-4250)(8000,-3250)
\whiten\path(7750,-3500)(8250,-3500)(8000,-3250)(7750,-3500)
\blacken\path(7750,-750)(8250,-750)(8000,-1250)(7750,-750)
\blacken\path(7750,10750)(8250,10750)(8000,11250)(7750,10750)

\path(10000,-2000)(10000,12000)
\put(9750,-2750){\tiny$N$}
\path(10000,-4250)(10000,-3250)
\whiten\path(9750,-3500)(10250,-3500)(10000,-3250)(9750,-3500)
\blacken\path(9750,-750)(10250,-750)(10000,-1250)(9750,-750)
\blacken\path(9750,10750)(10250,10750)(10000,11250)(9750,10750)

\path(12000,-2000)(12000,12000)
\put(11750,-2750){\tiny$N+1$}
\path(12000,-4250)(12000,-3250)
\whiten\path(11750,-3500)(12250,-3500)(12000,-3250)(11750,-3500)
\blacken\path(11750,-1250)(12250,-1250)(12000,-750)(11750,-1250)
\blacken\path(11750,750)(12250,750)(12000,1250)(11750,750)
\blacken\path(11750,2750)(12250,2750)(12000,3250)(11750,2750)
\blacken\path(11750,4750)(12250,4750)(12000,5250)(11750,4750)
\blacken\path(11750,6750)(12250,6750)(12000,7250)(11750,6750)
\blacken\path(11750,8750)(12250,8750)(12000,9250)(11750,8750)
\blacken\path(11750,10750)(12250,10750)(12000,11250)(11750,10750)

\path(14000,-2000)(14000,12000)
\path(14000,-4250)(14000,-3250)
\whiten\path(13750,-3500)(14250,-3500)(14000,-3250)(13750,-3500)
\blacken\path(13750,-1250)(14250,-1250)(14000,-750)(13750,-1250)
\blacken\path(13750,750)(14250,750)(14000,1250)(13750,750)
\blacken\path(13750,2750)(14250,2750)(14000,3250)(13750,2750)
\blacken\path(13750,4750)(14250,4750)(14000,5250)(13750,4750)
\blacken\path(13750,6750)(14250,6750)(14000,7250)(13750,6750)
\blacken\path(13750,8750)(14250,8750)(14000,9250)(13750,8750)
\blacken\path(13750,10750)(14250,10750)(14000,11250)(13750,10750)

\path(16000,-2000)(16000,12000)
\path(16000,-4250)(16000,-3250)
\whiten\path(15750,-3500)(16250,-3500)(16000,-3250)(15750,-3500)
\blacken\path(15750,-1250)(16250,-1250)(16000,-750)(15750,-1250)
\blacken\path(15750,750)(16250,750)(16000,1250)(15750,750)
\blacken\path(15750,2750)(16250,2750)(16000,3250)(15750,2750)
\blacken\path(15750,4750)(16250,4750)(16000,5250)(15750,4750)
\blacken\path(15750,6750)(16250,6750)(16000,7250)(15750,6750)
\blacken\path(15750,8750)(16250,8750)(16000,9250)(15750,8750)
\blacken\path(15750,10750)(16250,10750)(16000,11250)(15750,10750)

\path(18000,-2000)(18000,12000)
\put(17750,-2750){\tiny$M$}
\path(18000,-4250)(18000,-3250)
\whiten\path(17750,-3500)(18250,-3500)(18000,-3250)(17750,-3500)
\blacken\path(17750,-1250)(18250,-1250)(18000,-750)(17750,-1250)
\blacken\path(17750,750)(18250,750)(18000,1250)(17750,750)
\blacken\path(17750,2750)(18250,2750)(18000,3250)(17750,2750)
\blacken\path(17750,4750)(18250,4750)(18000,5250)(17750,4750)
\blacken\path(17750,6750)(18250,6750)(18000,7250)(17750,6750)
\blacken\path(17750,8750)(18250,8750)(18000,9250)(17750,8750)
\blacken\path(17750,10750)(18250,10750)(18000,11250)(17750,10750)

%%%%%%%%%

\put(20000,-3000){\fbox{$w$}}

\end{picture}

\end{minipage}
\end{center}

\caption[Equivalence between $S_0$ and $Z_N$]{Equivalence between $S_0$ and $Z_N$. The final $M-N$ columns of the $S_0$ lattice must assume the configuration shown. All other configurations vanish. The block of vertices thus obtained corresponds with the prefactor in (\ref{cond4}), whilst the remainder of the lattice represents $Z_N$.}
\end{figure}

}
\end{conditions}

\end{proof}

\subsection{Determinant expression for $S_n(\{u\}_n,\{v\}_N,\{w\}_M)$}

\begin{lemma}
{\rm 
Let us define the functions 
\begin{align}
&
f_i(w)
=
\frac{[\gamma]}{[v_i-w]}
\prod_{k\not=i}^{N} [v_k-w+\gamma],
\label{ffunct}
\\
&
g_i(u)
=
\label{gfunct}
\frac{[\gamma]}{[v_i-u]}
\left(
\prod_{k\not=i}^{N} [v_k-u+\gamma] \prod_{k=1}^{M} [u-w_k+\gamma]
-
\prod_{k\not=i}^{N} [v_k-u-\gamma] \prod_{k=1}^{M}[u-w_k]
\right).
\end{align}
Using these definitions, we construct the $N \times N$ matrix
\begin{align}
\mathcal{M}_n \Big(\{u\}_n,\{v\}_N,\{w\}_M\Big)
=
\label{M}
\left(
\begin{array}{cccccc}
f_1(w_1) & \cdots & f_1(w_{\widetilde{N}}) & 
g_1(u_n) & \cdots & g_1(u_1)
\\
\vdots & & \vdots & \vdots & & \vdots
\\
f_N(w_1) & \cdots & f_N(w_{\widetilde{N}}) &
g_N(u_n) & \cdots & g_N(u_1)
\end{array}
\right).
\end{align}
Assuming that the parameters $\{v\}_N$ satisfy the Bethe equations (\ref{bethe2}), we have
\begin{align}
S_n 
\label{sNcalc-xxz}
=
\frac{
\displaystyle{
\prod_{i=1}^{N} \prod_{j=1}^{M} [v_i-w_j]
\det\mathcal{M}_n \Big(\{u\}_n,\{v\}_N,\{w\}_M\Big)
}
}
{
\displaystyle{
\prod_{i=1}^{n} \prod_{j=1}^{\widetilde{N}}[u_i-w_j]
\prod_{1 \leq i < j \leq n} [u_i-u_j]
\prod_{1 \leq i < j \leq N} [v_i-v_j]
\prod_{1 \leq i < j \leq \widetilde{N}} [w_j-w_i]
}
}.
\end{align}
The expression (\ref{sNcalc-xxz}) for the intermediate Bethe scalar product $S_n$ originally appeared in appendix C of \cite{kmt}.

}
\end{lemma}

\begin{proof}
Firstly, one must show that $S_n$ is uniquely determined by the set of conditions in lemma \ref{sncond}. This is accomplished using similar arguments to those presented in lemma \ref{uniqueness}, and we shall assume this fact {\it a priori.} Hence it will be sufficient to show that the expression (\ref{sNcalc-xxz}) satisfies the list of properties given in lemma \ref{sncond}.

\setcounter{conditions}{0}
\begin{conditions}
{\rm
All dependence of the expression (\ref{sNcalc-xxz}) on the variables $\{w_{\widetilde{N}+1},\ldots,w_M\}$ occurs in the factor $\prod_{i=1}^{N} \prod_{j=1}^{M} [v_i-w_j]$ and in the functions $g_i(u_j)$ in the determinant. Clearly, these terms are invariant under the permutation $w_i \leftrightarrow w_j$ for all $i \not= j$. Hence the expression (\ref{sNcalc-xxz}) is symmetric in $\{w_{\widetilde{N}+1},\ldots,w_M\}$.
}
\end{conditions}

\begin{conditions} 
{\rm Consider the expression (\ref{gfunct}) for $g_i(u_n)$. Since the variables $\{v\}_N$ satisfy the Bethe equations (\ref{bethe2}), the numerator of $g_i(u_n)$ vanishes in the limit $u_n \rightarrow v_i$. It follows that the pole in (\ref{gfunct}) is removable, and therefore $g_i(u_n)$ is a trigonometric polynomial of degree $M+N-2$ in $u_n$. Using this fact, we see that (\ref{sNcalc-xxz}) is a quotient of trigonometric polynomials in $u_n$. The polynomial in the numerator has degree $M+N-2$, while the polynomial in the denominator has degree $N-1$. We must show that every zero in the denominator is cancelled by a zero in the numerator.

Setting $u_n = u_j$ for $1 \leq j \leq n-1$ causes two columns of the determinant to become equal, producing $n-1$ zeros in the numerator which cancel $n-1$ of the zeros in the denominator. Furthermore since 
\begin{align}
g_i(w_j)
&=
\frac{[\gamma]}{[v_i-w_j]}
\prod_{k \not= i}^{N} [v_k-w_j+\gamma]
\prod_{k = 1}^{M} [w_j-w_k+\gamma]
=
\prod_{k=1}^{M}[w_j-w_k+\gamma]
f_i(w_j),
\label{stuff}
\end{align}
it follows that by setting $u_n = w_j$ for $1 \leq j \leq \widetilde{N}$, two columns of the determinant are equal up to a multiplicative factor, producing $\widetilde{N}$ zeros in the numerator which cancel $\widetilde{N}$ of the zeros in the denominator. This proves that the expression (\ref{sNcalc-xxz}) is a trigonometric polynomial of degree $M-1$ in $u_n$. 

Finally, since
\begin{align}
g_i(w_j-\gamma)
&=
\frac{-[\gamma]}{[v_i-w_j+\gamma]}
\prod_{k \not= i}^{N} [v_k-w_j]
\prod_{k=1}^{M} [w_j-w_k-\gamma],
\\
&=
-
\prod_{k=1}^{M} [w_j-w_k-\gamma]
\prod_{k=1}^{N} \frac{[v_k-w_j]}{[v_k-w_j+\gamma]}
f_i(w_j),
\nonumber
\end{align}
we see that by setting $u_n = w_j -\gamma$ for all $1 \leq j \leq \widetilde{N}$, two columns of the determinant are equal up to a multiplicative factor, producing the $\widetilde{N}$ zeros which (\ref{sNcalc-xxz}) requires in order to satisfy property {\bf 2}. 
}
\end{conditions}

\begin{conditions}
{\rm Using equation (\ref{stuff}) and the definition of the matrix (\ref{M}), it is clear that 
\begin{align}
&
\det \mathcal{M}_n \Big( \{u\}_n,\{v\}_N,\{w\}_M  \Big)
\Big|_{u_n = w_{\widetilde{N}+1}}
=
\label{sNrec-xxz}
\\
&
\prod_{k=1}^{M}
[w_{\widetilde{N}+1}-w_k+\gamma]
\det \mathcal{M}_{n-1} \Big( \{u\}_{n-1}, \{v\}_N, \{w\}_M \Big).
\nonumber
\end{align}
Furthermore, we notice the trivial product identity
\begin{align}
&
\left.
\left(
\prod_{i=1}^{n} \prod_{j=1}^{\widetilde{N}} [u_i-w_j]
\prod_{1 \leq i<j \leq n} [u_i-u_j]
\prod_{1 \leq i< j \leq \widetilde{N}} [w_j -w_i]
\right)
\right|_{u_n = w_{\widetilde{N}+1}}
=
\label{sNrec2-xxz}
\\
&
\phantom{\Big((}
\prod_{i=1}^{n-1} \prod_{j=1}^{\widetilde{N}+1} [u_i-w_j]
\prod_{1 \leq i<j \leq n-1} [u_i-u_j]
\prod_{1 \leq i < j \leq \widetilde{N}+1}[w_j - w_i].
\nonumber
\end{align}
Combining the results (\ref{sNrec-xxz}) and (\ref{sNrec2-xxz}), we find that the expression (\ref{sNcalc-xxz}) satisfies the recursion relation (\ref{sNrec0-xxz}).
}
\end{conditions}

\begin{conditions}
{\rm
Taking the $n=0$ case of (\ref{sNcalc-xxz}) yields
\begin{align}
S_0\Big(
\{v\}_N,\{w\}_M
\Big)
=
\frac{\displaystyle{
\prod_{i=1}^{N}
\prod_{j=1}^{M}
[v_i-w_j]
\det\mathcal{M}_0
}}
{\displaystyle{
\prod_{1\leq i<j \leq N}
[v_i-v_j][w_j-w_i]
}},
\label{s0last}
\end{align}
with the matrix $\mathcal{M}_0$ given by
\begin{align}
\mathcal{M}_0
&=
\left(
\begin{array}{ccc}
\frac{[\gamma]}{[v_1-w_1]} \prod_{k\not=1} [v_k-w_1+\gamma]
&
\cdots
&
\frac{[\gamma]}{[v_1-w_N]} \prod_{k\not=1} [v_k-w_N+\gamma]
\\
\vdots
&
&
\vdots
\\
\frac{[\gamma]}{[v_N-w_1]} \prod_{k\not=N} [v_k-w_1+\gamma]
&
\cdots
&
\frac{[\gamma]}{[v_N-w_N]} \prod_{k\not=N} [v_k-w_N+\gamma]
\end{array}
\right).
\end{align}
Comparing (\ref{s0last}) with the determinant expression (\ref{dwpf}) for the domain wall partition function, we see that $S_0 = \prod_{i=1}^{N} \prod_{j=N+1}^{M} [v_i-w_j] Z_N$, as required.
}
\end{conditions}

\end{proof}

\subsection{Evaluation of $S_N(\{u\}_N,\{v\}_N,\{w\}_M)$}

Let us now consider the $n=N$ case of equation (\ref{sNcalc-xxz}) in more detail. For purely aesthetic purposes, we simultaneously reverse the order of the columns in the matrix $\mathcal{M}_N$ and the order of the variables in the Vandermonde $\prod_{1\leq i<j \leq N} [u_i-u_j]$. We also take the transpose of the matrix $\mathcal{M}_N$. The formula (\ref{sNcalc-xxz}) is invariant under these transformations, and we obtain 
\begin{align}
&
S_N\Big(\{u\}_N,\{v\}_N,\{w\}_M\Big)
=
\frac{\displaystyle{
[\gamma]^N
\prod_{i=1}^{N} \prod_{j=1}^{M}
[v_i-w_j]
}}
{\displaystyle{
\prod_{1\leq i < j \leq N}
[u_j-u_i][v_i-v_j]
}}
\times
\label{betscal}
\\
&
\det \left(
\frac{\displaystyle{
\prod_{k\not= j}^{N}[v_k-u_i-\gamma]
\prod_{k=1}^{M} [u_i-w_k]
-
\prod_{k\not=j}^{N} [v_k-u_i+\gamma]
\prod_{k=1}^{M} [u_i-w_k+\gamma]
}}
{[u_i-v_j]}
\right)_{1\leq i,j \leq N}.
\nonumber
\end{align}
The expression (\ref{betscal}) for the Bethe scalar product was discovered by Slavnov in \cite{sla}. The original proof required a recursion relation between scalar products of dimension $N$ and $N-1$, which can be found in section 3, chapter IX of \cite{kbi}. 

\section{Trigonometric Felderhof model}
\label{ff-fel}

In this section we devote our attention to the trigonometric Felderhof model. Most of the material that we present originally appeared in \cite{fwz1}.

\subsection{Hamiltonian $H$}

The Hamiltonian of the XX spin-$\frac{1}{2}$ chain in an external magnetic field $h$ is given by
\begin{align}
H
=
\frac{1}{2}
\sum_{m=1}^{M}
\Big(
\sigma_m^{x} \sigma_{m+1}^{x}
+
\sigma_m^{y} \sigma_{m+1}^{y}
+
2h \sigma_m^{z}
\Big),
\label{ham-tf}
\end{align}
assuming the usual periodicity of the Pauli matrices. Similarly to the anisotropy constant in the XXZ model, we will find it convenient to parametrize $h = \cosh(2p)$, where $p$ is a variable to be introduced below.

\subsection{$R$-matrix and Yang-Baxter equation}

At the free fermion point $\gamma = \frac{\pi i}{2}$, the $R$-matrix (\ref{Rmat1}) may be generalized to include extra variables, in such a way that the Yang-Baxter equation remains satisfied. This leads to the trigonometric limit of the model introduced by Felderhof in \cite{fel2}, and accordingly we call it the {\it trigonometric Felderhof model.} This model was also studied in \cite{da1}, as the first in a hierarchy of vertex models with increasing spin. The $R$-matrix for the trigonometric Felderhof model is given by
\begin{align}
R_{ab}(u,p,v,q)
=
\label{Rmat-tf}
\left(
\begin{array}{cccc}
a_{+}(u,p,v,q) & 0 & 0 & 0
\\
0 & b_{+}(u,p,v,q) & c_{+}(u,p,v,q) & 0
\\
0 & c_{-}(u,p,v,q) & b_{-}(u,p,v,q) & 0
\\
0 & 0 & 0 & a_{-}(u,p,v,q)
\end{array}
\right)_{ab},
\end{align}
where we have defined the functions
\begin{align}
&
a_{\pm}(u,p,v,q) = [\pm(u-v)+p+q],
\label{a-tf}
\\
&
b_{\pm}(u,p,v,q) = [u-v\pm(q-p)],
\label{b-tf}
\\
&
c_{\pm}(u,p,v,q) = [2p]^{\frac{1}{2}} [2q]^{\frac{1}{2}},
\label{c-tf}
\end{align}
with $[u] = \sinh u$ as usual.\footnote{The parametrization of \cite{da1} is recovered by multiplying all weights by $e^{u-v+p+q}$ and setting $e^{2p} = \alpha, e^{2q} = \beta$.} The $R$-matrix is an element of ${\rm End}(V_a \otimes V_b)$, and the variables $u,v$ are rapidities associated to the respective vector spaces $V_a,V_b$. The new features of this $R$-matrix are the variables $p,q$. These are called {\it external fields}, and are associated to the respective vector spaces $V_a,V_b$. We recover the free fermion point of the six-vertex model by setting $p=q=\frac{\pi i}{4}$.

The entries of the $R$-matrix (\ref{Rmat-tf}) admit the same graphical representation as those of the $R$-matrix (\ref{Rmat1}). The only difference is that each vertex line now accommodates a rapidity variable and an external field. Hence we identify the functions (\ref{a-tf})--(\ref{c-tf}) with the vertices shown below.

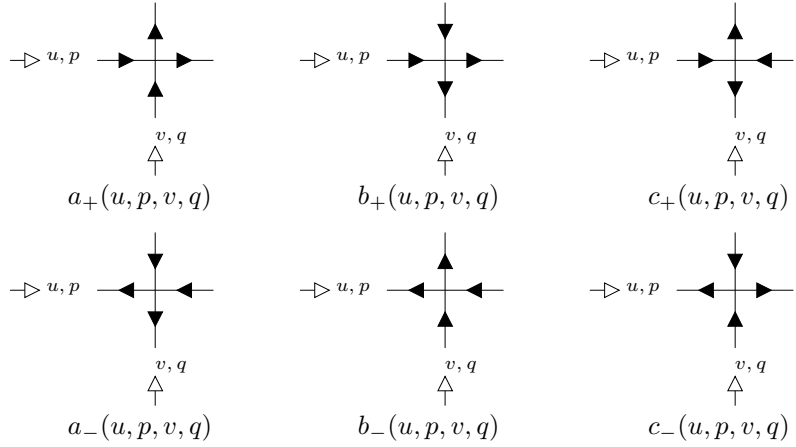
\begin{figure}[H]
\begin{center}
\begin{minipage}{4.3in}

\setlength{\unitlength}{0.00038cm}
\begin{picture}(20000,14000)(-4500,-12500)

%%%Vertex a_{+}%%%
\path(-2000,0000)(2000,0000)
\blacken\path(-1250,250)(-1250,-250)(-750,0)(-1250,250)
\blacken\path(750,250)(750,-250)(1250,0)(750,250)
%%%
\put(-3750,0){\scriptsize{$u,p$}}
\path(-5000,0)(-4000,0)
\whiten\path(-4500,250)(-4500,-250)(-4000,0)(-4500,250)
%%%
\path(0000,-2000)(0000,2000)
\blacken\path(-250,-1250)(250,-1250)(0,-750)(-250,-1250)
\blacken\path(-250,750)(250,750)(0,1250)(-250,750)
%%%
\put(-3000,-5000){$a_{+}(u,p,v,q)$}
\put(0,-2750){\scriptsize{$v,q$}}
\path(0,-4000)(0,-3000)
\whiten\path(-250,-3500)(250,-3500)(0,-3000)(-250,-3500)

%%%Vertex b_{+}%%%
\path(8000,0000)(12000,0000)
\blacken\path(8750,250)(8750,-250)(9250,0)(8750,250)
\blacken\path(10750,250)(10750,-250)(11250,0)(10750,250)
%%%
\put(6250,0){\scriptsize{$u,p$}}
\path(5000,0)(6000,0)
\whiten\path(5500,250)(5500,-250)(6000,0)(5500,250)
%%%
\path(10000,-2000)(10000,2000)
\blacken\path(9750,-750)(10250,-750)(10000,-1250)(9750,-750)
\blacken\path(9750,1250)(10250,1250)(10000,750)(9750,1250)
%%%
\put(7000,-5000){$b_{+}(u,p,v,q)$}
\put(10000,-2750){\scriptsize{$v,q$}}
\path(10000,-4000)(10000,-3000)
\whiten\path(9750,-3500)(10250,-3500)(10000,-3000)(9750,-3500)

%%%Vertex c_{+}%%%
\path(18000,0000)(22000,0000)
\blacken\path(18750,250)(18750,-250)(19250,0)(18750,250)
\blacken\path(21250,250)(21250,-250)(20750,0)(21250,250)
%%%
\put(16250,0){\scriptsize{$u,p$}}
\path(15000,0)(16000,0)
\whiten\path(15500,250)(15500,-250)(16000,0)(15500,250)
%%%
\path(20000,-2000)(20000,2000)
\blacken\path(19750,-750)(20250,-750)(20000,-1250)(19750,-750)
\blacken\path(19750,750)(20250,750)(20000,1250)(19750,750)
%%%
\put(17000,-5000){$c_{+}(u,p,v,q)$}
\put(20000,-2750){\scriptsize{$v,q$}}
\path(20000,-4000)(20000,-3000)
\whiten\path(19750,-3500)(20250,-3500)(20000,-3000)(19750,-3500)

%%%Vertex a_{-}%%%
\path(-2000,-8000)(2000,-8000)
\blacken\path(-750,-7750)(-750,-8250)(-1250,-8000)(-750,-7750)
\blacken\path(1250,-7750)(1250,-8250)(750,-8000)(1250,-7750)
%%%
\put(-3750,-8000){\scriptsize{$u,p$}}
\path(-5000,-8000)(-4000,-8000)
\whiten\path(-4500,-7750)(-4500,-8250)(-4000,-8000)(-4500,-7750)
%%%
\path(0000,-10000)(0000,-6000)
\blacken\path(-250,-8750)(250,-8750)(0,-9250)(-250,-8750)
\blacken\path(-250,-6750)(250,-6750)(0,-7250)(-250,-6750)
%%%
\put(-3000,-13000){$a_{-}(u,p,v,q)$}
\put(0,-10750){\scriptsize{$v,q$}}
\path(0,-12000)(0,-11000)
\whiten\path(-250,-11500)(250,-11500)(0,-11000)(-250,-11500)

%%%Vertex b_{-}%%%
\path(8000,-8000)(12000,-8000)
\blacken\path(9250,-7750)(9250,-8250)(8750,-8000)(9250,-7750)
\blacken\path(11250,-7750)(11250,-8250)(10750,-8000)(11250,-7750)
%%%
\put(6250,-8000){\scriptsize{$u,p$}}
\path(5000,-8000)(6000,-8000)
\whiten\path(5500,-7750)(5500,-8250)(6000,-8000)(5500,-7750)
%%%
\path(10000,-10000)(10000,-6000)
\blacken\path(9750,-9250)(10250,-9250)(10000,-8750)(9750,-9250)
\blacken\path(9750,-7250)(10250,-7250)(10000,-6750)(9750,-7250)
%%%
\put(7000,-13000){$b_{-}(u,p,v,q)$}
\put(10000,-10750){\scriptsize{$v,q$}}
\path(10000,-12000)(10000,-11000)
\whiten\path(9750,-11500)(10250,-11500)(10000,-11000)(9750,-11500)

%%%Vertex c_{-}%%%
\path(18000,-8000)(22000,-8000)
\blacken\path(19250,-7750)(19250,-8250)(18750,-8000)(19250,-7750)
\blacken\path(20750,-7750)(20750,-8250)(21250,-8000)(20750,-7750)
%%%
\put(16250,-8000){\scriptsize{$u,p$}}
\path(15000,-8000)(16000,-8000)
\whiten\path(15500,-7750)(15500,-8250)(16000,-8000)(15500,-7750)
%%%
\path(20000,-10000)(20000,-6000)
\blacken\path(19750,-9250)(20250,-9250)(20000,-8750)(19750,-9250)
\blacken\path(19750,-6750)(20250,-6750)(20000,-7250)(19750,-6750)
%%%
\put(17000,-13000){$c_{-}(u,p,v,q)$}
\put(20000,-10750){\scriptsize{$v,q$}}
\path(20000,-12000)(20000,-11000)
\whiten\path(19750,-11500)(20250,-11500)(20000,-11000)(19750,-11500)

\end{picture}

\end{minipage}
\end{center}

\caption[Six vertices of the trigonometric Felderhof model]{Six vertices of the trigonometric Felderhof model.}
\label{verttf} 
\end{figure}

With the following result we assert that the Yang-Baxter equation continues to hold, even in the presence of the external fields.

\begin{lemma}
{\rm 
The $R$-matrix (\ref{Rmat-tf}) obeys the Yang-Baxter equation
\begin{align}
R_{ab}(u,p,v,q)
R_{ac}(u,p,w,r)
R_{bc}(v,q,w,r)
=
\label{yb-tf}
R_{bc}(v,q,w,r)
R_{ac}(u,p,w,r)
R_{ab}(u,p,v,q).
\end{align}
This is an identity in ${\rm End}(V_a \otimes V_b \otimes V_c)$, true for all $u,v,w$ and $p,q,r$.
}
\end{lemma}

\begin{figure}[H]
\begin{center}
\begin{minipage}{4.3in}

\setlength{\unitlength}{0.00038cm}
\begin{picture}(20000,9500)(-2000,-3500)

\path(0,0)(5000,5000)
\put(1250,1250){\circle*{200}}
\put(1250,650){\tiny{$i_b$}}
\put(5000,5000){\circle*{200}}
\put(5000,4400){\tiny{$k_b$}}
\put(8750,5000){\circle*{200}}
\put(8750,4400){\tiny{$j_b$}}
%%%
\path(-3000,0)(-2000,0)
\whiten\path(-2250,250)(-2250,-250)(-2000,0)(-2250,250)
\put(-1750,0){\scriptsize{$v,q$}}
%%%
\path(5000,0)(0,5000)
\put(1250,3750){\circle*{200}}
\put(1250,4000){\tiny{$i_a$}}
\put(5000,0){\circle*{200}}
\put(5000,-600){\tiny{$k_a$}}
\put(8750,0){\circle*{200}}
\put(8750,-600){\tiny{$j_a$}}
%%%
\path(-3000,5000)(-2000,5000)
\whiten\path(-2250,5250)(-2250,4750)(-2000,5000)(-2250,5250)
\put(-1750,5000){\scriptsize{$u,p$}}
%%%
\path(5000,5000)(10000,5000)
\path(5000,0)(10000,0)
%%%
\path(7500,-2500)(7500,7500)
\put(7500,-1250){\circle*{200}}
\put(7750,-1250){\tiny{$i_c$}}
\put(7500,2500){\circle*{200}}
\put(7750,2500){\tiny{$k_c$}}
\put(7500,6250){\circle*{200}}
\put(7750,6250){\tiny{$j_c$}}
\put(7500,-3250){\scriptsize{$w,r$}}
\path(7500,-4500)(7500,-3500)
\whiten\path(7250,-3750)(7750,-3750)(7500,-3500)(7250,-3750)

%%%%%%%%%%

\put(11000,2300){$=$}

%%%%%%%%%%
\path(12000,5000)(13000,5000)
\whiten\path(12750,5250)(12750,4750)(13000,5000)(12750,5250)
\put(13250,5000){\scriptsize{$u,p$}}
%%%%%%%%%%
\path(15000,5000)(20000,5000)
\put(16250,5000){\circle*{200}}
\put(16250,4400){\tiny{$i_a$}}
\path(20000,5000)(25000,0)
\put(20000,5000){\circle*{200}}
\put(20000,5250){\tiny{$k_a$}}
\put(23750,1250){\circle*{200}}
\put(23750,1600){\tiny{$j_a$}}
%%%%%%%%%%
\path(12000,0)(13000,0)
\whiten\path(12750,250)(12750,-250)(13000,0)(12750,250)
\put(13250,0){\scriptsize{$v,q$}}
%%%%%%%%%%
\path(15000,0)(20000,0)
\put(16250,0){\circle*{200}}
\put(16250,-600){\tiny{$i_b$}}
\path(20000,0)(25000,5000)
\put(20000,0){\circle*{200}}
\put(20000,-600){\tiny{$k_b$}}
\put(23750,3750){\circle*{200}}
\put(23750,3250){\tiny{$j_b$}}
%%%%%%%%%%
\path(17500,-4500)(17500,-3500)
\whiten\path(17250,-3750)(17750,-3750)(17500,-3500)(17250,-3750)
\put(17500,-3250){\scriptsize{$w,r$}}
%%%%%%%%%%
\path(17500,-2500)(17500,7500)
\put(17500,-1250){\circle*{200}}
\put(17750,-1250){\tiny{$i_c$}}
\put(17500,2500){\circle*{200}}
\put(17750,2500){\tiny{$k_c$}}
\put(17500,6250){\circle*{200}}
\put(17750,6250){\tiny{$j_c$}}

\end{picture}

\end{minipage}
\end{center}

\caption[Yang-Baxter equation for the trigonometric Felderhof model]{Yang-Baxter equation for the trigonometric Felderhof model. Each index $i_a,i_b,i_c,j_a,j_b,j_c$ represents a black arrow that is fixed on both sides of the equation, while the indices $k_a,k_b,k_c$ are summed over all possible arrow configurations.}
\label{ybtf}
\end{figure}
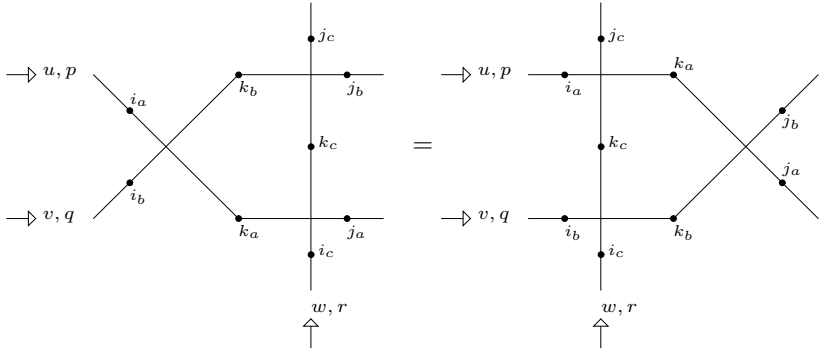

\subsection{Monodromy matrix and intertwining equation}

The monodromy matrix is an ordered product of $R$-matrices, given by
\begin{align}
T_a(u,p,\{w,r\}_M)
=
R_{a1}(u,p,w_1,r_1)
\ldots
R_{aM}(u,p,w_M,r_M),
\end{align}
with the multiplication taken in the space ${\rm End}(V_a)$. We write the contribution from the space ${\rm End}(V_a)$ explicitly, by defining
\begin{align}
T_a(u,p,\{w,r\}_M)
=
\left(
\begin{array}{cc}
A(u,p,\{w,r\}_M) & B(u,p,\{w,r\}_M)
\\
C(u,p,\{w,r\}_M) & D(u,p,\{w,r\}_M)
\end{array}
\right)_a,
\label{tf-mon}
\end{align}
where the matrix entries are all operators acting in $V_1 \otimes \cdots \otimes V_M$. For a graphical representation of these operators, we refer the reader to figure \ref{6vT}. The correspondence is exactly the same, except that the rapidity $u$ is now accompanied by the external field $p$, and each $w_j$ by an $r_j$. Due to the Yang-Baxter equation (\ref{yb-tf}), we obtain the intertwining equation
\begin{align}
&
R_{ab}(u,p,v,q) T_a(u,p,\{w,r\}_M) T_b(v,q,\{w,r\}_M)
\label{tf-int}
=
\\
&
T_b(v,q,\{w,r\}_M) T_a(u,p,\{w,r\}_M) R_{ab}(u,p,v,q).
\nonumber
\end{align}
As usual, this leads to sixteen commutation relations amongst the entries of the monodromy matrix (\ref{tf-mon}). Of these commutation relations, two have particular significance in our later calculations. They are given by
\begin{align}
& [u-v+p+q] B(u,p,\{w,r\}_M) B(v,q,\{w,r\}_M) 
=
\label{bb-tf}
\\
& [v-u+p+q] B(v,q,\{w,r\}_M)B(u,p,\{w,r\}_M),
\nonumber
\\
\nonumber
\\
& [v-u+p+q] C(u,p,\{w,r\}_M) C(v,q,\{w,r\}_M)
=
\label{cc-tf}
\\
& [u-v+p+q] C(v,q,\{w,r\}_M)C(u,p,\{w,r\}_M).
\nonumber
\end{align}

\subsection{Recovering $H$ from the transfer matrix}

Let 
\begin{align}
t(u,p,\{w,r\}_M)
=
A(u,p,\{w,r\}_M)
+
D(u,p,\{w,r\}_M)
\end{align}
denote the transfer matrix of the trigonometric Felderhof model. The Hamiltonian (\ref{ham-tf}) is recovered via the formula
\begin{align}
H
=
[2p]
\frac{\partial}{\partial u} \log t(u,p)
\Big|_{u=0},
\quad{\rm where}\quad
t(u,p) 
=
t(u,p,\{w,r\}_M)
\Big|_{\substack{w_1=\cdots = w_M = 0 \\ r_1 = \cdots = r_M = p}}.
\end{align}
Therefore all eigenvectors of $t(u,p,\{w,r\}_M)$ are also eigenvectors of $H$. Once again, the problem of diagonalizing $H$ is reduced to finding vectors $|\Psi\rangle \in V$ satisfying
\begin{align}
\Big(
A(u,p,\{w,r\}_M) + D(u,p,\{w,r\}_M)
\Big)
|\Psi\rangle
=
\tau_{\Psi}(u,p,\{w,r\}_M) 
|\Psi\rangle
\end{align}
for some constants $\tau_{\Psi}(u,p,\{w,r\}_M)$.

\subsection{Bethe Ansatz for the eigenvectors}

We construct the Bethe eigenvectors of the trigonometric Felderhof model in an analogous fashion to in \S \ref{s-betheans}. The eigenvectors $|\Psi\rangle \in V$ and $\langle \Psi | \in V^{*}$ of $t(u,p,\{w,r\}_M)$ are given by the Ans\"atze
\begin{align}
|\Psi\rangle
&=
B(v_1,q_1,\{w,r\}_M)
\ldots
B(v_N,q_N,\{w,r\}_M)
|\Uparrow_M\rangle,
\label{tf-bethe3}
\\
\langle \Psi |
&=
\langle \Uparrow_M |
C(v_N,q_N,\{w,r\}_M)
\ldots
C(v_1,q_1,\{w,r\}_M),
\label{tf-bethe4}
\end{align}
with $N \leq M$ in both cases. To show that these are genuine eigenvectors, it is necessary to act on them with $t(u,p,\{w,r\}_M)$ using the commutation relations contained in (\ref{tf-int}). We find that (\ref{tf-bethe3}) and (\ref{tf-bethe4}) are eigenstates of the transfer matrix if and only if 
\begin{align}
(-)^N
\prod_{k=1}^{M}
[v_j+q_j-w_k+r_k]
+
\prod_{k=1}^{M}
[v_j+q_j-w_k-r_k]
=
0
\label{tf-bethe2}
\end{align}
for all $1\leq j \leq N$. The identities (\ref{tf-bethe2}) constitute the Bethe equations for the trigonometric Felderhof model and in contrast to those of the XXZ chain (\ref{bethe2}), the dependence on $\{v_1,\ldots, v_N\}$ is decoupled. This is a consequence of the free fermionic nature of the model being studied. Hence we may interpret $\{v_1+q_1,\ldots,v_N+q_N\}$ as being roots, in $x$, of the single equation
\begin{align}
(-)^N
\prod_{k=1}^{M}
[x-w_k+r_k]
+
\prod_{k=1}^{M}
[x-w_k-r_k]
=
0.
\label{tf-bethe5}
\end{align}

\section{Domain wall partition function $Z_N\Big(\{v,q\}_N,\{w,r\}_N\Big)$}

\subsection{Definition of $Z_N(\{v,q\}_N,\{w,r\}_N)$}

The domain wall partition function of the trigonometric Felderhof model has the algebraic definition
\begin{align}
Z_N\Big(
\{v,q\}_N,\{w,r\}_N
\Big)
=
\langle \Downarrow_N|
\lprod_{j=1}^{N}
B(v_j,q_j,\{w,r\}_N)
|\Uparrow_N\rangle.
\label{pf-tf}
\end{align}
This naturally extends the definition of the domain wall partition function (\ref{pf-xxz}) to a model containing external fields. Notice that we must define an ordering of the $B$-operators in (\ref{pf-tf}), since by (\ref{bb-tf}) they do not commute. 

Similarly to before, we represent the domain wall partition function by an $N \times N$ lattice, as shown in figure \ref{da-part}.

\begin{figure}[H]

\begin{center}
\begin{minipage}{4.3in}

\setlength{\unitlength}{0.00038cm}
\begin{picture}(20000,12500)(-10000,-3000)

\path(-2000,0)(10000,0)
\put(-4250,0){\tiny$v_N,q_N$}
\path(-5500,0)(-4500,0)
\whiten\path(-4750,250)(-4750,-250)(-4500,0)(-4750,250)
\blacken\path(-1250,250)(-1250,-250)(-750,0)(-1250,250)
\blacken\path(9250,250)(9250,-250)(8750,0)(9250,250)

\path(-2000,2000)(10000,2000)
\path(-5500,2000)(-4500,2000)
\whiten\path(-4750,2250)(-4750,1750)(-4500,2000)(-4750,2250)
\blacken\path(-1250,2250)(-1250,1750)(-750,2000)(-1250,2250)
\blacken\path(9250,2250)(9250,1750)(8750,2000)(9250,2250)

\path(-2000,4000)(10000,4000)
\path(-5500,4000)(-4500,4000)
\whiten\path(-4750,4250)(-4750,3750)(-4500,4000)(-4750,4250)
\blacken\path(-1250,4250)(-1250,3750)(-750,4000)(-1250,4250)
\blacken\path(9250,4250)(9250,3750)(8750,4000)(9250,4250)

\path(-2000,6000)(10000,6000)
\path(-5500,6000)(-4500,6000)
\whiten\path(-4750,6250)(-4750,5750)(-4500,6000)(-4750,6250)
\blacken\path(-1250,6250)(-1250,5750)(-750,6000)(-1250,6250)
\blacken\path(9250,6250)(9250,5750)(8750,6000)(9250,6250)

\path(-2000,8000)(10000,8000)
\put(-4250,8000){\tiny$v_1,q_1$}
\path(-5500,8000)(-4500,8000)
\whiten\path(-4750,8250)(-4750,7750)(-4500,8000)(-4750,8250)
\blacken\path(-1250,8250)(-1250,7750)(-750,8000)(-1250,8250)
\blacken\path(9250,8250)(9250,7750)(8750,8000)(9250,8250)

%%%%%%%%%%%

\path(0,-2000)(0,10000)
\put(-500,-2750){\tiny$w_1,r_1$}
\path(0,-4000)(0,-3000)
\whiten\path(-250,-3250)(250,-3250)(0,-3000)(-250,-3250)
\blacken\path(-250,-750)(250,-750)(0,-1250)(-250,-750)
\blacken\path(-250,8750)(250,8750)(0,9250)(-250,8750)

\path(2000,-2000)(2000,10000)
\path(2000,-4000)(2000,-3000)
\whiten\path(1750,-3250)(2250,-3250)(2000,-3000)(1750,-3250)
\blacken\path(1750,-750)(2250,-750)(2000,-1250)(1750,-750)
\blacken\path(1750,8750)(2250,8750)(2000,9250)(1750,8750)

\path(4000,-2000)(4000,10000)
\path(4000,-4000)(4000,-3000)
\whiten\path(3750,-3250)(4250,-3250)(4000,-3000)(3750,-3250)
\blacken\path(3750,-750)(4250,-750)(4000,-1250)(3750,-750)
\blacken\path(3750,8750)(4250,8750)(4000,9250)(3750,8750)

\path(6000,-2000)(6000,10000)
\path(6000,-4000)(6000,-3000)
\whiten\path(5750,-3250)(6250,-3250)(6000,-3000)(5750,-3250)
\blacken\path(5750,-750)(6250,-750)(6000,-1250)(5750,-750)
\blacken\path(5750,8750)(6250,8750)(6000,9250)(5750,8750)

\path(8000,-2000)(8000,10000)
\put(7500,-2750){\tiny$w_N,r_N$}
\path(8000,-4000)(8000,-3000)
\whiten\path(7750,-3250)(8250,-3250)(8000,-3000)(7750,-3250)
\blacken\path(7750,-750)(8250,-750)(8000,-1250)(7750,-750)
\blacken\path(7750,8750)(8250,8750)(8000,9250)(7750,8750)

\end{picture}

\end{minipage}
\end{center}

\caption[Domain wall partition function of the trigonometric Felderhof model]{Domain wall partition function of the trigonometric Felderhof model. The top row of arrows corresponds with the state vector $|\Uparrow_N\rangle$. The bottom row of arrows corresponds with the dual state vector $\langle \Downarrow_N|$. Each horizontal lattice line corresponds to multiplication by a $B(v_j,q_j,\{w,r\}_N)$ operator. Notice that the ordering of these lattice lines respects the ordering of $B$-operators defined in (\ref{pf-tf}).} 

\label{da-part}
\end{figure}
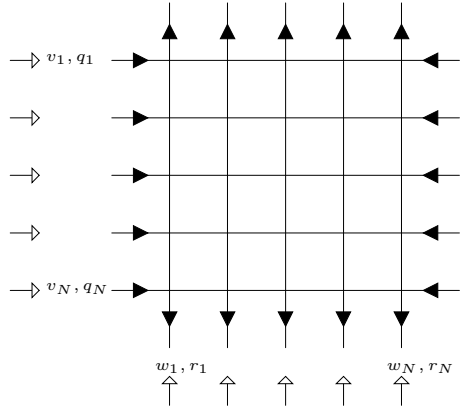

\subsection{Conditions on $Z_N( \{v,q\}_N, \{w,r\}_N )$}

We progress towards calculating the domain wall partition function (\ref{pf-tf}). The procedure begins with the following result from \cite{fwz1}, which establishes a set of Korepin-type conditions on $Z_N(\{v,q\}_N,\{w,r\}_N)$. 

\begin{lemma}
{\rm We adopt the shorthand $Z_N = Z_N(\{v,q\}_N,\{w,r\}_N)$. For all $N \geq 2$ we claim that

\setcounter{conditions}{0}
\begin{conditions}
{\rm 
$Z_N$ is a trigonometric polynomial of degree $N-1$ in the rapidity variable $v_N$.
}
\end{conditions}

\begin{conditions}
{\rm
$Z_N$ has zeros at the points $v_N = v_j+q_j+q_N$, for all $1 \leq j \leq N-1$.
}
\end{conditions}

\begin{conditions}
{\rm 
Setting $v_N=w_N+q_N+r_N$, $Z_N$ satisfies the recursion relation
\begin{align}
Z_N \Big|_{v_N = w_N+q_N+r_N}
=
[2q_N]^{\frac{1}{2}} [2r_N]^{\frac{1}{2}}
\label{pfrec1-tf}
\prod_{j=1}^{N-1}
[w_N-w_j+r_j+r_N]
[v_j-w_N+q_j-r_N]
Z_{N-1},
\end{align}
where $Z_{N-1}$ is the domain wall partition function on a square lattice of size $N-1$.
}
\end{conditions}

In addition, we have the supplementary condition

\begin{conditions}
{\rm 
The partition function on the $1\times 1$ lattice is given by
$
Z_1
=
[2q_1]^{\frac{1}{2}}
[2r_1]^{\frac{1}{2}}
$.
}
\end{conditions}
}
\end{lemma}

\begin{proof}

\setcounter{conditions}{0}
\begin{conditions}
{\rm 
By inserting the set of states $\sum_{n=1}^{N} \sigma_n^{+} |\Downarrow_N\rangle \langle \Downarrow_N| \sigma_n^{-}$ after the first $B$-operator appearing in (\ref{pf-tf}), we obtain the expansion
\begin{align}
Z_N\Big(\{v,q\}_N,\{w,r\}_N\Big)
&=
\sum_{n=1}^{N}
\langle \Downarrow_N|
B(v_N,q_N,\{w,r\}_N)
\sigma_n^{+}
|\Downarrow_N\rangle
\label{tf-peel2}
\\
&
\times
\langle \Downarrow_N|
\sigma_n^{-}
\lprod_{j=1}^{N-1}
B(v_j,q_j,\{w,r\}_N)
|\Uparrow_N\rangle,
\nonumber
\end{align}
in which all dependence on $v_N$ appears in the first factor within the sum. Hence we shall calculate $\langle \Downarrow_N| B(v_N,q_N,\{w,r\}_N) \sigma_n^{+} |\Downarrow_N\rangle$ for all $1\leq n \leq N$, as shown below:

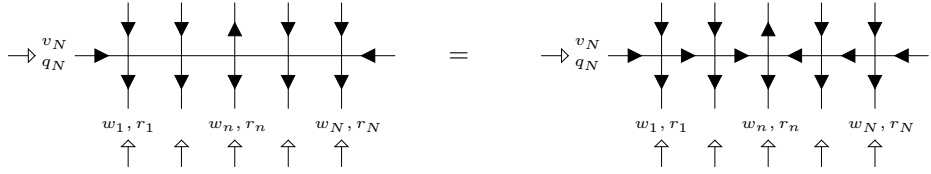
\begin{figure}[H]

\begin{center}
\begin{minipage}{4.3in}

\setlength{\unitlength}{0.00035cm}
\begin{picture}(20000,6000)(-2500,-3000)

\path(-2000,0)(10000,0)
\put(-3250,400){\tiny$v_N$}
\put(-3250,-400){\tiny$q_N$}
\path(-4500,0)(-3500,0)
\whiten\path(-3750,250)(-3750,-250)(-3500,0)(-3750,250)
\blacken\path(-1250,250)(-1250,-250)(-750,0)(-1250,250)
\blacken\path(9250,250)(9250,-250)(8750,0)(9250,250)

%%%%%%%%%%%

\path(0,-2000)(0,2000)
\put(-1000,-2750){\tiny$w_1,r_1$}
\path(0,-4200)(0,-3200)
\whiten\path(-250,-3450)(250,-3450)(0,-3200)(-250,-3450)
\blacken\path(-250,-750)(250,-750)(0,-1250)(-250,-750)
\blacken\path(-250,1250)(250,1250)(0,750)(-250,1250)

\path(2000,-2000)(2000,2000)
\path(2000,-4200)(2000,-3200)
\whiten\path(1750,-3450)(2250,-3450)(2000,-3200)(1750,-3450)
\blacken\path(1750,-750)(2250,-750)(2000,-1250)(1750,-750)
\blacken\path(1750,1250)(2250,1250)(2000,750)(1750,1250)

\path(4000,-2000)(4000,2000)
\put(3000,-2750){\tiny$w_n,r_n$}
\path(4000,-4200)(4000,-3200)
\whiten\path(3750,-3450)(4250,-3450)(4000,-3200)(3750,-3450)
\blacken\path(3750,-750)(4250,-750)(4000,-1250)(3750,-750)
\blacken\path(3750,750)(4250,750)(4000,1250)(3750,750)

\path(6000,-2000)(6000,2000)
\path(6000,-4200)(6000,-3200)
\whiten\path(5750,-3450)(6250,-3450)(6000,-3200)(5750,-3450)
\blacken\path(5750,-750)(6250,-750)(6000,-1250)(5750,-750)
\blacken\path(5750,1250)(6250,1250)(6000,750)(5750,1250)

\path(8000,-2000)(8000,2000)
\put(7000,-2750){\tiny$w_N,r_N$}
\path(8000,-4200)(8000,-3200)
\whiten\path(7750,-3450)(8250,-3450)(8000,-3200)(7750,-3450)
\blacken\path(7750,-750)(8250,-750)(8000,-1250)(7750,-750)
\blacken\path(7750,1250)(8250,1250)(8000,750)(7750,1250)

%%%%%%%

\put(12000,-250){$=$}

%%%%%%%2nd lattice %%%%%%

\path(18000,0)(30000,0)
\put(16750,400){\tiny$v_N$}
\put(16750,-400){\tiny$q_N$}
\path(15500,0)(16500,0)
\whiten\path(16250,250)(16250,-250)(16500,0)(16250,250)
\blacken\path(18750,250)(18750,-250)(19250,0)(18750,250)
\blacken\path(20750,250)(20750,-250)(21250,0)(20750,250)
\blacken\path(22750,250)(22750,-250)(23250,0)(22750,250)
\blacken\path(25250,250)(25250,-250)(24750,0)(25250,250)
\blacken\path(27250,250)(27250,-250)(26750,0)(27250,250)
\blacken\path(29250,250)(29250,-250)(28750,0)(29250,250)

%%%%%%%%%%%

\path(20000,-2000)(20000,2000)
\put(19000,-2750){\tiny$w_1,r_1$}
\path(20000,-4200)(20000,-3200)
\whiten\path(19750,-3450)(20250,-3450)(20000,-3200)(19750,-3450)
\blacken\path(19750,-750)(20250,-750)(20000,-1250)(19750,-750)
\blacken\path(19750,1250)(20250,1250)(20000,750)(19750,1250)

\path(22000,-2000)(22000,2000)
\path(22000,-4200)(22000,-3200)
\whiten\path(21750,-3450)(22250,-3450)(22000,-3200)(21750,-3450)
\blacken\path(21750,-750)(22250,-750)(22000,-1250)(21750,-750)
\blacken\path(21750,1250)(22250,1250)(22000,750)(21750,1250)

\path(24000,-2000)(24000,2000)
\put(23000,-2750){\tiny$w_n,r_n$}
\path(24000,-4200)(24000,-3200)
\whiten\path(23750,-3450)(24250,-3450)(24000,-3200)(23750,-3450)
\blacken\path(23750,-750)(24250,-750)(24000,-1250)(23750,-750)
\blacken\path(23750,750)(24250,750)(24000,1250)(23750,750)

\path(26000,-2000)(26000,2000)
\path(26000,-4200)(26000,-3200)
\whiten\path(25750,-3450)(26250,-3450)(26000,-3200)(25750,-3450)
\blacken\path(25750,-750)(26250,-750)(26000,-1250)(25750,-750)
\blacken\path(25750,1250)(26250,1250)(26000,750)(25750,1250)

\path(28000,-2000)(28000,2000)
\put(27000,-2750){\tiny$w_N,r_N$}
\path(28000,-4200)(28000,-3200)
\whiten\path(27750,-3450)(28250,-3450)(28000,-3200)(27750,-3450)
\blacken\path(27750,-750)(28250,-750)(28000,-1250)(27750,-750)
\blacken\path(27750,1250)(28250,1250)(28000,750)(27750,1250)

\end{picture}

\end{minipage}
\end{center}

\caption[Peeling away the bottom row of the trigonometric Felderhof partition function]{Peeling away the bottom row of the trigonometric Felderhof partition function. The diagram on the left hand side represents the quantity $\langle \Downarrow_N |B(v_N,q_N,\{w,r\}_N) \sigma_n^{+} |\Downarrow_N \rangle$, with the internal black arrows being summed over all configurations. The diagram on the right represents the only surviving configuration.}

\label{tf-zequiv3}
\end{figure}

The right hand side of figure \ref{tf-zequiv3} represents a product of vertices. Replacing each vertex with its corresponding trigonometric weight (see figure 14), we have
\begin{align}
&
\langle \Downarrow_N| B(v_N,q_N,\{w,r\}_N)
\sigma_n^{+} |\Downarrow_N\rangle
=
\label{tf-peel}
\\
&
[2q_N]^{\frac{1}{2}} [2r_n]^{\frac{1}{2}}
\prod_{1 \leq j < n} [v_N-w_j+r_j-q_N]
\prod_{n < j \leq N} [w_j-v_N+q_N+r_j].
\nonumber
\end{align}
Substituting (\ref{tf-peel}) into the expansion (\ref{tf-peel2}) gives
\begin{align}
Z_N
&=
\sum_{n=1}^{N}
[2q_N]^{\frac{1}{2}} [2r_n]^{\frac{1}{2}}
\prod_{1\leq j <n} [v_N-w_j+r_j-q_N]
\prod_{n < j \leq N} [w_j-v_N+q_N+r_j]
\label{pfexp-tf}
\\
&
\times
\langle \Downarrow_N| \sigma_n^{-}
\lprod_{j=1}^{N-1} B(v_j,q_j,\{w,r\}_N)
|\Uparrow_N\rangle.
\nonumber
\end{align}
From (\ref{pfexp-tf}) we see that every term in $Z_N(\{v,q\}_N,\{w,r\}_N)$ contains a product of $N-1$ trigonometric functions with argument $v_N$. Thus $Z_N(\{v,q\}_N,\{w,r\}_N)$ is a trigonometric polynomial of degree $N-1$ in the variable $v_N$.
}
\end{conditions}

\begin{conditions}
{\rm 
We multiply the partition function (\ref{pf-tf}) by $\prod_{j=1}^{N-1}[v_N-v_j+q_j+q_N]$ and repeatedly use the commutation relation
\begin{align}
&
[v_N-v_j+q_j+q_N]
B(v_N,q_N,\{w,r\}_N)
B(v_j,q_j,\{w,r\}_N)
=
\\
&
[v_j-v_N+q_j+q_N]
B(v_j,q_j,\{w,r\}_N)
B(v_N,q_N,\{w,r\}_N),
\nonumber
\end{align}
which is a rewriting of (\ref{bb-tf}), to change the order of the $B$-operators. We obtain 
\begin{align}
&
\prod_{j=1}^{N-1}
[v_N-v_j+q_j+q_N]
Z_N\Big( \{v,q\}_N,\{w,r\}_N\Big)
=
\label{reorder}
\\
&
\prod_{j=1}^{N-1}
[v_j-v_N+q_j+q_N]
\langle \Downarrow_N|
\lprod_{j=1}^{N-1} B(v_j,q_j,\{w,r\}_N)
B(v_N,q_N,\{w,r\}_N)
|\Uparrow_N\rangle.
\nonumber
\end{align}
Graphically, we depict (\ref{reorder}) with the following diagrams:

\begin{figure}[H]

\begin{center}
\begin{minipage}{4.3in}

\setlength{\unitlength}{0.00030cm}
\begin{picture}(20000,14000)(-3000,-4000)

\path(-2000,0)(10000,0)
\put(-3250,350){\scriptsize{$v_N$}}
\put(-3250,-350){\scriptsize{$q_N$}}
\path(-4750,0)(-3750,0)
\whiten\path(-4000,250)(-4000,-250)(-3750,0)(-4000,250)
\blacken\path(-1250,250)(-1250,-250)(-750,0)(-1250,250)

\path(-2000,2000)(12000,2000)
\path(-4750,2000)(-3750,2000)
\whiten\path(-4000,2250)(-4000,1750)(-3750,2000)(-4000,2250)
\blacken\path(-1250,2250)(-1250,1750)(-750,2000)(-1250,2250)
\blacken\path(9250,2250)(9250,1750)(8750,2000)(9250,2250)
\blacken\path(11250,2250)(11250,1750)(10750,2000)(11250,2250)

\path(-2000,4000)(12000,4000)
\path(-4750,4000)(-3750,4000)
\whiten\path(-4000,4250)(-4000,3750)(-3750,4000)(-4000,4250)
\blacken\path(-1250,4250)(-1250,3750)(-750,4000)(-1250,4250)
\blacken\path(9250,4250)(9250,3750)(8750,4000)(9250,4250)
\blacken\path(11250,4250)(11250,3750)(10750,4000)(11250,4250)

\path(-2000,6000)(12000,6000)
\path(-4750,6000)(-3750,6000)
\whiten\path(-4000,6250)(-4000,5750)(-3750,6000)(-4000,6250)
\blacken\path(-1250,6250)(-1250,5750)(-750,6000)(-1250,6250)
\blacken\path(9250,6250)(9250,5750)(8750,6000)(9250,6250)
\blacken\path(11250,6250)(11250,5750)(10750,6000)(11250,6250)

\path(-2000,8000)(12000,8000)
\put(-3250,8350){\scriptsize{$v_1$}}
\put(-3250,7650){\scriptsize{$q_1$}}
\path(-4750,8000)(-3750,8000)
\whiten\path(-4000,8250)(-4000,7750)(-3750,8000)(-4000,8250)
\blacken\path(-1250,8250)(-1250,7750)(-750,8000)(-1250,8250)
\blacken\path(9250,8250)(9250,7750)(8750,8000)(9250,8250)
\blacken\path(11250,8250)(11250,7750)(10750,8000)(11250,8250)

%%%%%%%%%%%

\path(0,-2000)(0,10000)
\put(0,-2650){\scriptsize{$w_1$}}
\put(0,-3350){\scriptsize{$r_1$}}
\path(0,-4750)(0,-3750)
\whiten\path(-250,-4000)(250,-4000)(0,-3750)(-250,-4000)
\blacken\path(-250,-750)(250,-750)(0,-1250)(-250,-750)
\blacken\path(-250,8750)(250,8750)(0,9250)(-250,8750)

\path(2000,-2000)(2000,10000)
\path(2000,-4750)(2000,-3750)
\whiten\path(1750,-4000)(2250,-4000)(2000,-3750)(1750,-4000)
\blacken\path(1750,-750)(2250,-750)(2000,-1250)(1750,-750)
\blacken\path(1750,8750)(2250,8750)(2000,9250)(1750,8750)

\path(4000,-2000)(4000,10000)
\path(4000,-4750)(4000,-3750)
\whiten\path(3750,-4000)(4250,-4000)(4000,-3750)(3750,-4000)
\blacken\path(3750,-750)(4250,-750)(4000,-1250)(3750,-750)
\blacken\path(3750,8750)(4250,8750)(4000,9250)(3750,8750)

\path(6000,-2000)(6000,10000)
\path(6000,-4750)(6000,-3750)
\whiten\path(5750,-4000)(6250,-4000)(6000,-3750)(5750,-4000)
\blacken\path(5750,-750)(6250,-750)(6000,-1250)(5750,-750)
\blacken\path(5750,8750)(6250,8750)(6000,9250)(5750,8750)

\path(8000,-2000)(8000,10000)
\put(8000,-2650){\scriptsize{$w_N$}}
\put(8000,-3350){\scriptsize{$r_N$}}
\path(8000,-4750)(8000,-3750)
\whiten\path(7750,-4000)(8250,-4000)(8000,-3750)(7750,-4000)
\blacken\path(7750,-750)(8250,-750)(8000,-1250)(7750,-750)
\blacken\path(7750,8750)(8250,8750)(8000,9250)(7750,8750)

\path(10000,0)(10000,10000)
\blacken\path(9750,1250)(10250,1250)(10000,750)(9750,1250)
\blacken\path(9750,3250)(10250,3250)(10000,2750)(9750,3250)
\blacken\path(9750,5250)(10250,5250)(10000,4750)(9750,5250)
\blacken\path(9750,7250)(10250,7250)(10000,6750)(9750,7250)
\blacken\path(9750,9250)(10250,9250)(10000,8750)(9750,9250)

%%%%%%%%%%%%%%

\put(14000,3750){$=$}

%%%%%%%%2nd lattice%%%%%%%%%%%

\path(20000,0)(34000,0)
\path(17250,0)(18250,0)
\whiten\path(18000,250)(18000,-250)(18250,0)(18000,250)
\blacken\path(20750,250)(20750,-250)(21250,0)(20750,250)
\blacken\path(22750,250)(22750,-250)(23250,0)(22750,250)
\blacken\path(33250,250)(33250,-250)(32750,0)(33250,250)

\path(20000,2000)(34000,2000)
\path(17250,2000)(18250,2000)
\whiten\path(18000,2250)(18000,1750)(18250,2000)(18000,2250)
\blacken\path(20750,2250)(20750,1750)(21250,2000)(20750,2250)
\blacken\path(22750,2250)(22750,1750)(23250,2000)(22750,2250)
\blacken\path(33250,2250)(33250,1750)(32750,2000)(33250,2250)

\path(20000,4000)(34000,4000)
\path(17250,4000)(18250,4000)
\whiten\path(18000,4250)(18000,3750)(18250,4000)(18000,4250)
\blacken\path(20750,4250)(20750,3750)(21250,4000)(20750,4250)
\blacken\path(22750,4250)(22750,3750)(23250,4000)(22750,4250)
\blacken\path(33250,4250)(33250,3750)(32750,4000)(33250,4250)

\path(20000,6000)(34000,6000)
\put(18750,6350){\scriptsize$v_1$}
\put(18750,5650){\scriptsize$q_1$}
\path(17250,6000)(18250,6000)
\whiten\path(18000,6250)(18000,5750)(18250,6000)(18000,6250)
\blacken\path(20750,6250)(20750,5750)(21250,6000)(20750,6250)
\blacken\path(22750,6250)(22750,5750)(23250,6000)(22750,6250)
\blacken\path(33250,6250)(33250,5750)(32750,6000)(33250,6250)

\path(22000,8000)(34000,8000)
\blacken\path(33250,8250)(33250,7750)(32750,8000)(33250,8250)

%%%%%%%%%%%

\path(22000,-2000)(22000,8000)
\put(22000,-2650){\scriptsize$v_N$}
\put(22000,-3350){\scriptsize$q_N$}
\path(22000,-4750)(22000,-3750)
\whiten\path(21750,-4000)(22250,-4000)(22000,-3750)(21750,-4000)
\blacken\path(21750,-1250)(22250,-1250)(22000,-750)(21750,-1250)
\blacken\path(21750,750)(22250,750)(22000,1250)(21750,750)
\blacken\path(21750,2750)(22250,2750)(22000,3250)(21750,2750)
\blacken\path(21750,4750)(22250,4750)(22000,5250)(21750,4750)
\blacken\path(21750,6750)(22250,6750)(22000,7250)(21750,6750)

\path(24000,-2000)(24000,10000)
\put(24000,-2650){\scriptsize$w_1$}
\put(24000,-3350){\scriptsize$r_1$}
\path(24000,-4750)(24000,-3750)
\whiten\path(23750,-4000)(24250,-4000)(24000,-3750)(23750,-4000)
\blacken\path(23750,-750)(24250,-750)(24000,-1250)(23750,-750)
\blacken\path(23750,8750)(24250,8750)(24000,9250)(23750,8750)

\path(26000,-2000)(26000,10000)
\path(26000,-4750)(26000,-3750)
\whiten\path(25750,-4000)(26250,-4000)(26000,-3750)(25750,-4000)
\blacken\path(25750,-750)(26250,-750)(26000,-1250)(25750,-750)
\blacken\path(25750,8750)(26250,8750)(26000,9250)(25750,8750)

\path(28000,-2000)(28000,10000)
\path(28000,-4750)(28000,-3750)
\whiten\path(27750,-4000)(28250,-4000)(28000,-3750)(27750,-4000)
\blacken\path(27750,-750)(28250,-750)(28000,-1250)(27750,-750)
\blacken\path(27750,8750)(28250,8750)(28000,9250)(27750,8750)

\path(30000,-2000)(30000,10000)
\path(30000,-4750)(30000,-3750)
\whiten\path(29750,-4000)(30250,-4000)(30000,-3750)(29750,-4000)
\blacken\path(29750,-750)(30250,-750)(30000,-1250)(29750,-750)
\blacken\path(29750,8750)(30250,8750)(30000,9250)(29750,8750)

\path(32000,-2000)(32000,10000)
\put(32000,-2650){\scriptsize$w_N$}
\put(32000,-3350){\scriptsize$r_N$}
\path(32000,-4750)(32000,-3750)
\whiten\path(31750,-4000)(32250,-4000)(32000,-3750)(31750,-4000)
\blacken\path(31750,-750)(32250,-750)(32000,-1250)(31750,-750)
\blacken\path(31750,8750)(32250,8750)(32000,9250)(31750,8750)

\end{picture}

\end{minipage}
\end{center}

\caption[Reordering the lattice lines of the trigonometric Felderhof partition function]{Reordering the lattice lines of the trigonometric Felderhof partition function. The diagram on the left is the domain wall partition function multiplied by the string of vertices $\prod_{j=1}^{N-1} a_{-}(v_j,q_j,v_N,q_N)$, and it corresponds with the left hand side of (\ref{reorder}). Each vertex can be threaded through the lattice using the Yang-Baxter equation, which ultimately produces the diagram on the right. This diagram represents the domain wall partition function with its $N^{\rm th}$ row transferred to the top of the lattice, multiplied by the string of vertices $\prod_{j=1}^{N-1} a_{+}(v_j,q_j,v_N,q_N)$. Clearly, this corresponds with the right hand side of (\ref{reorder}).}
\end{figure}
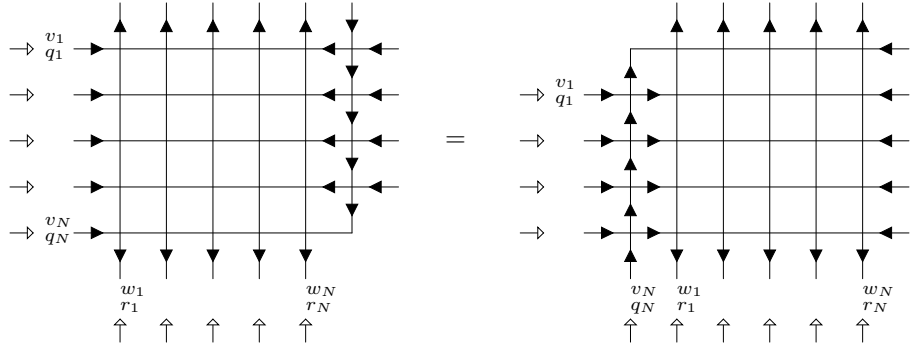

The right hand side of (\ref{reorder}) is a trigonometric polynomial of degree $2N-2$ in $v_N$, with zeros at the points $v_N = v_j + q_j +q_N$ for all $1 \leq j \leq N-1$. Therefore the partition function $Z_N(\{v,q\}_N,\{w,r\}_N)$ must have zeros at the same points.

}
\end{conditions}

\begin{conditions}
{\rm 
We start from the expansion (\ref{pfexp-tf}) of the domain wall partition function, and set $v_N = w_N+q_N+r_N$. This causes all terms in the summation over $1\leq n \leq N$ to collapse to zero except the $n=N$ term, and we obtain
\begin{align}
&
Z_N\Big(\{v,q\}_N,\{w,r\}_N\Big)\Big|_{v_N = w_N + q_N + r_N}
=
\label{pfrec2-tf}
\\
&
[2q_N]^{\frac{1}{2}} [2r_N]^{\frac{1}{2}}
\prod_{j=1}^{N-1} [w_N-w_j+r_j+r_N]
\langle \Downarrow_N| \sigma_N^{-}
\lprod_{j=1}^{N-1} B(v_j,q_j,\{w,r\}_N)
|\Uparrow_N\rangle.
\nonumber
\end{align}
We further simplify the previous expression by using the diagrammatic representation of $\langle \Downarrow_N | \sigma_N^{-}\ \lprod_{j=1}^{N-1} B(v_j,q_j,\{w,r\}_N) |\Uparrow_N\rangle$, shown below:

\begin{figure}[H]

\begin{center}
\begin{minipage}{4.3in}

\setlength{\unitlength}{0.000325cm}
\begin{picture}(20000,12500)(-4000,-3000)

\path(-2000,0)(10000,0)
\put(-4250,400){\tiny$v_{N-1}$}
\put(-4250,-400){\tiny$q_{N-1}$}
\path(-5500,0)(-4500,0)
\whiten\path(-4750,250)(-4750,-250)(-4500,0)(-4750,250)
\blacken\path(-1250,250)(-1250,-250)(-750,0)(-1250,250)
\blacken\path(9250,250)(9250,-250)(8750,0)(9250,250)

\path(-2000,2000)(10000,2000)
\path(-5500,2000)(-4500,2000)
\whiten\path(-4750,2250)(-4750,1750)(-4500,2000)(-4750,2250)
\blacken\path(-1250,2250)(-1250,1750)(-750,2000)(-1250,2250)
\blacken\path(9250,2250)(9250,1750)(8750,2000)(9250,2250)

\path(-2000,4000)(10000,4000)
\path(-5500,4000)(-4500,4000)
\whiten\path(-4750,4250)(-4750,3750)(-4500,4000)(-4750,4250)
\blacken\path(-1250,4250)(-1250,3750)(-750,4000)(-1250,4250)
\blacken\path(9250,4250)(9250,3750)(8750,4000)(9250,4250)

\path(-2000,6000)(10000,6000)
\put(-4250,6400){\tiny$v_1$}
\put(-4250,5600){\tiny$q_1$}
\path(-5500,6000)(-4500,6000)
\whiten\path(-4750,6250)(-4750,5750)(-4500,6000)(-4750,6250)
\blacken\path(-1250,6250)(-1250,5750)(-750,6000)(-1250,6250)
\blacken\path(9250,6250)(9250,5750)(8750,6000)(9250,6250)

%%%%%%%%%%%

\path(0,-2000)(0,8000)
\put(-1000,-2750){\tiny$w_1,r_1$}
\path(0,-4200)(0,-3200)
\whiten\path(-250,-3450)(250,-3450)(0,-3200)(-250,-3450)
\blacken\path(-250,-750)(250,-750)(0,-1250)(-250,-750)
\blacken\path(-250,6750)(250,6750)(0,7250)(-250,6750)

\path(2000,-2000)(2000,8000)
\path(2000,-4200)(2000,-3200)
\whiten\path(1750,-3450)(2250,-3450)(2000,-3200)(1750,-3450)
\blacken\path(1750,-750)(2250,-750)(2000,-1250)(1750,-750)
\blacken\path(1750,6750)(2250,6750)(2000,7250)(1750,6750)

\path(4000,-2000)(4000,8000)
\path(4000,-4200)(4000,-3200)
\whiten\path(3750,-3450)(4250,-3450)(4000,-3200)(3750,-3450)
\blacken\path(3750,-750)(4250,-750)(4000,-1250)(3750,-750)
\blacken\path(3750,6750)(4250,6750)(4000,7250)(3750,6750)

\path(6000,-2000)(6000,8000)
\path(6000,-4200)(6000,-3200)
\whiten\path(5750,-3450)(6250,-3450)(6000,-3200)(5750,-3450)
\blacken\path(5750,-750)(6250,-750)(6000,-1250)(5750,-750)
\blacken\path(5750,6750)(6250,6750)(6000,7250)(5750,6750)

\path(8000,-2000)(8000,8000)
\put(7000,-2750){\tiny$w_N,r_N$}
\path(8000,-4200)(8000,-3200)
\whiten\path(7750,-3450)(8250,-3450)(8000,-3200)(7750,-3450)
\blacken\path(7750,-1250)(8250,-1250)(8000,-750)(7750,-1250)
\blacken\path(7750,6750)(8250,6750)(8000,7250)(7750,6750)

%%%%%%%

\put(12000,2750){$=$}

%%%%%%%2nd lattice %%%%%%

\path(18000,0)(30000,0)
\put(15750,400){\tiny$v_{N-1}$}
\put(15750,-400){\tiny$q_{N-1}$}
\path(14500,0)(15500,0)
\whiten\path(15250,250)(15250,-250)(15500,0)(15250,250)
\blacken\path(18750,250)(18750,-250)(19250,0)(18750,250)
\blacken\path(27250,250)(27250,-250)(26750,0)(27250,250)
\blacken\path(29250,250)(29250,-250)(28750,0)(29250,250)

\path(18000,2000)(30000,2000)
\path(14500,2000)(15500,2000)
\whiten\path(15250,2250)(15250,1750)(15500,2000)(15250,2250)
\blacken\path(18750,2250)(18750,1750)(19250,2000)(18750,2250)
\blacken\path(27250,2250)(27250,1750)(26750,2000)(27250,2250)
\blacken\path(29250,2250)(29250,1750)(28750,2000)(29250,2250)

\path(18000,4000)(30000,4000)
\path(14500,4000)(15500,4000)
\whiten\path(15250,4250)(15250,3750)(15500,4000)(15250,4250)
\blacken\path(18750,4250)(18750,3750)(19250,4000)(18750,4250)
\blacken\path(27250,4250)(27250,3750)(26750,4000)(27250,4250)
\blacken\path(29250,4250)(29250,3750)(28750,4000)(29250,4250)

\path(18000,6000)(30000,6000)
\put(15750,6400){\tiny$v_1$}
\put(15750,5600){\tiny$q_1$}
\path(14500,6000)(15500,6000)
\whiten\path(15250,6250)(15250,5750)(15500,6000)(15250,6250)
\blacken\path(18750,6250)(18750,5750)(19250,6000)(18750,6250)
\blacken\path(27250,6250)(27250,5750)(26750,6000)(27250,6250)
\blacken\path(29250,6250)(29250,5750)(28750,6000)(29250,6250)

%%%%%%%%%%%

\path(20000,-2000)(20000,8000)
\put(19000,-2750){\tiny$w_1,r_1$}
\path(20000,-4200)(20000,-3200)
\whiten\path(19750,-3450)(20250,-3450)(20000,-3200)(19750,-3450)
\blacken\path(19750,-750)(20250,-750)(20000,-1250)(19750,-750)
\blacken\path(19750,6750)(20250,6750)(20000,7250)(19750,6750)

\path(22000,-2000)(22000,8000)
\path(22000,-4200)(22000,-3200)
\whiten\path(21750,-3450)(22250,-3450)(22000,-3200)(21750,-3450)
\blacken\path(21750,-750)(22250,-750)(22000,-1250)(21750,-750)
\blacken\path(21750,6750)(22250,6750)(22000,7250)(21750,6750)

\path(24000,-2000)(24000,8000)
\path(24000,-4200)(24000,-3200)
\whiten\path(23750,-3450)(24250,-3450)(24000,-3200)(23750,-3450)
\blacken\path(23750,-750)(24250,-750)(24000,-1250)(23750,-750)
\blacken\path(23750,6750)(24250,6750)(24000,7250)(23750,6750)

\path(26000,-2000)(26000,8000)
\path(26000,-4200)(26000,-3200)
\whiten\path(25750,-3450)(26250,-3450)(26000,-3200)(25750,-3450)
\blacken\path(25750,-750)(26250,-750)(26000,-1250)(25750,-750)
\blacken\path(25750,6750)(26250,6750)(26000,7250)(25750,6750)

\path(28000,-2000)(28000,8000)
\put(27000,-2750){\tiny$w_N,r_N$}
\path(28000,-4200)(28000,-3200)
\whiten\path(27750,-3450)(28250,-3450)(28000,-3200)(27750,-3450)
\blacken\path(27750,-1250)(28250,-1250)(28000,-750)(27750,-1250)
\blacken\path(27750,750)(28250,750)(28000,1250)(27750,750)
\blacken\path(27750,2750)(28250,2750)(28000,3250)(27750,2750)
\blacken\path(27750,4750)(28250,4750)(28000,5250)(27750,4750)
\blacken\path(27750,6750)(28250,6750)(28000,7250)(27750,6750)

\end{picture}

\end{minipage}
\end{center}

\caption[Peeling the right-most column of the trigonometric Felderhof partition function]{Peeling the right-most column of the trigonometric Felderhof partition function. The diagram on the left hand side represents the quantity $\langle \Downarrow_N | \sigma_N^{-}\ \lprod_{j=1}^{N-1} B(v_j,q_j,\{w,r\}_N) |\Uparrow_N\rangle$, with the internal black arrows being summed over all configurations. The diagram on the right contains all surviving configurations.}

\label{tf-zequiv5}
\end{figure}

The right hand side of figure \ref{tf-zequiv5} represents the $(N-1)\times (N-1)$ partition function, multiplied by a column of vertices. Replacing these vertices with their trigonometric weights, we have
\begin{align}
\langle \Downarrow_N| \sigma_N^{-}
\lprod_{j=1}^{N-1} B(v_j,q_j,\{w,r\}_N)
|\Uparrow_N\rangle
=
\label{tf-peel3}
\prod_{j=1}^{N-1} [v_j-w_N+q_j-r_N]
Z_{N-1}.
\end{align}
Substituting (\ref{tf-peel3}) into (\ref{pfrec2-tf}) we recover the required recursion relation (\ref{pfrec1-tf}). 

}
\end{conditions}

\begin{conditions}
{\rm
Specializing the definition (\ref{pf-tf}) to the case $N=1$ gives
\begin{align}
Z_1(v_1,q_1,w_1,r_1)
&=
\langle \Downarrow_1|
B(v_1,q_1,\{w,r\}_1)
|\Uparrow_1\rangle,
\\
&=
\uparrow_{a_1}^{*} \otimes \downarrow_{1}^{*}
R_{a_1 1}(v_1,q_1,w_1,r_1)
\uparrow_1 \otimes \downarrow_{a_1}
=
[2q_1]^{\frac{1}{2}} [2r_1]^{\frac{1}{2}},
\nonumber
\end{align}
as required. Alternatively, the $1\times 1$ partition function is the top-right vertex in figure \ref{verttf}, whose weight is equal to $[2q_1]^{\frac{1}{2}} [2r_1]^{\frac{1}{2}}$.
}
\end{conditions}
\end{proof}

\subsection{Factorized expression for $Z_{N}(\{v,q\}_{N},\{w,r\}_{N})$}

The conditions {\bf 1}--{\bf 4} are strong constraints. Not only do they specify $Z_{N}(\{v,q\}_{N},\{w,r\}_{N})$ uniquely, they lead to its direct evaluation, as we demonstrate below.

\begin{lemma}
{\rm 
The domain wall partition function has the factorized expression
\begin{align}
&
Z_N\Big(\{v,q\}_N,\{w,r\}_N\Big)
=
\label{lem-tf}
\prod_{j=1}^{N}
[2q_j]^{\frac{1}{2}} [2r_j]^{\frac{1}{2}}
\prod_{1 \leq j < k \leq N}
[v_j-v_k+q_j+q_k] [w_k-w_j+r_j+r_k].
\end{align}
The result (\ref{lem-tf}) was first obtained in \cite{cfwz} using a complicated recursion relation. A more straightforward proof, based on solving the conditions {\bf 1}--{\bf 4}, subsequently appeared in \cite{fwz1}. It is the latter proof which we present below. 
}
\end{lemma} 

\begin{proof}
From condition {\bf 1} and {\bf 2} on $Z_{N}(\{v,q\}_{N},\{w,r\}_{N})$ we know that it must have the form
\begin{align}
Z_{N}\Big(\{v,q\}_{N},\{w,r\}_{N}\Big)
=
\mathcal{C}\Big(\{v\}_{N-1},\{q\}_N,\{w,r\}_N\Big)
\prod_{j=1}^{N-1}
[v_j-v_N+q_j+q_N],
\label{p1-tf}
\end{align}
where $\mathcal{C}$ does not depend on $v_N$, but depends on all other variables. Evaluating (\ref{p1-tf}) at $v_N = w_N+q_N+r_N$ and comparing with condition {\bf 3} on $Z_{N}$, we obtain 
\begin{align}
Z_N\Big|_{v_N = w_N+q_N+r_N}
&=
\mathcal{C} \Big(\{v\}_{N-1},\{q\}_N,\{w,r\}_N\Big)
\prod_{j=1}^{N-1}
[v_j-w_N+q_j-r_N],
\\
&=
[2q_N]^{\frac{1}{2}} [2r_N]^{\frac{1}{2}}
\prod_{j=1}^{N-1}
[w_N-w_j+r_j+r_N]
[v_j-w_N+q_j-r_N]
Z_{N-1},
\nonumber
\end{align}
from which we extract the equation
\begin{align}
\mathcal{C}
=
[2q_N]^{\frac{1}{2}} [2r_N]^{\frac{1}{2}}
\prod_{j=1}^{N-1}
[w_N-w_j+r_j+r_N]
Z_{N-1}\Big(\{v,q\}_{N-1},\{w,r\}_{N-1}\Big).
\end{align}
Substituting this expression for $\mathcal{C}$ into (\ref{p1-tf}), we obtain the recurrence 
\begin{align}
&
Z_N\Big(\{v,q\}_N,\{w,r\}_N\Big)
=
[2q_N]^{\frac{1}{2}} [2r_N]^{\frac{1}{2}}
\times
\\
&
\prod_{j=1}^{N-1}
[v_j-v_N+q_j+q_N]
[w_N-w_j+r_j+r_N]
Z_{N-1}\Big(\{v,q\}_{N-1},\{w,r\}_{N-1}\Big),
\nonumber
\end{align}
whose basis is given by condition {\bf 4}. This recurrence is trivially solved to produce the formula (\ref{lem-tf}).
\end{proof}

\section{Scalar products $S_n\Big(\{u,p\}_n,\{v,q\}_N,\{w,r\}_M\Big)$}

\subsection{Definition of $S_n(\{u,p\}_n,\{v,q\}_N,\{w,r\}_M)$}

Define three sets of rapidity variables $\{u\}_n = \{u_1,\ldots,u_n\},\{v\}_N = \{v_1,\ldots,v_N\},\{w\}_M = \{w_1,\ldots,w_M\}$ and their corresponding sets of external fields $\{p\}_n = \{p_1,\ldots,p_n\},\{q\}_N = \{q_1,\ldots,q_N\},\{r\}_M = \{r_1,\ldots,r_M\}$. The cardinalities of these sets are assumed to satisfy $0 \leq n \leq N$ and $1 \leq N \leq M$. For $n=0$ we define
\begin{align}
S_0\Big(\{v,q\}_N,\{w,r\}_M\Big)
=
\langle \Downarrow_{N/M} |
\lprod_{k=1}^{N}
B(v_k,q_k,\{w,r\}_M)
|\Uparrow_M\rangle.
\label{tf-sp1}
\end{align}
Similarly to in \S 4 we will find that $S_0$ is equal to the trigonometric Felderhof partition function $Z_N$, up to an overall normalization. Next, for all $1\leq n \leq N-1$ we define
\begin{align}
&
S_n\Big(\{u,p\}_n,\{v,q\}_N,\{w,r\}_M\Big)
=
\label{tf-sp2}
\langle \Downarrow_{\widetilde{N}/M} |
\lprod_{j=1}^{n} C(u_j,p_j,\{w,r\}_M)
\lprod_{k=1}^{N} B(v_k,q_k,\{w,r\}_M)
|\Uparrow_M\rangle
\end{align}
with $\widetilde{N} = N-n$. Finally, in the case $n=N$ we fix
\begin{align}
S_N\Big(
\{u,p\}_N,\{v,q\}_N,\{w,r\}_M
\Big)
=
\label{tf-sp3}
\langle \Uparrow_M|
\lprod_{j=1}^{N}
C(u_j,p_j,\{w,r\}_M)
\lprod_{k=1}^{N}
B(v_k,q_k,\{w,r\}_M)
|\Uparrow_M \rangle.
\end{align}
The scalar products (\ref{tf-sp1})--(\ref{tf-sp3}) are the trigonometric Felderhof analogues of those defined in \S \ref{xxz-sp-def}. They have identical graphical representations to those described in \S \ref{xxz-sp-graph}, except that every rapidity variable is now accompanied by an appropriate external field. In the following subsection we give a set of conditions on these scalar products, using similar techniques to those developed earlier.

\subsection{Conditions on $S_n(\{u,p\}_n,\{v,q\}_N,\{w,r\}_M)$}

\begin{lemma}
{\rm 
For all $1\leq n \leq N$ we claim that

\setcounter{conditions}{0}
\begin{conditions}
{\rm 
$S_n$ is invariant under the simultaneous permutation of variables
$\{w_j,r_j\} \leftrightarrow \{w_k,r_k\}$ for all $j,k \in \{\widetilde{N}+1,\ldots,M\}$.
}
\end{conditions}

\begin{conditions}
{\rm
$S_n$ is a trigonometric polynomial of degree $M-1$ in $u_n$, with zeros occurring at the points $u_n =p_n+w_j+r_j$, for all $1\leq j \leq \widetilde{N}$.
}
\end{conditions}

\begin{conditions}
{\rm
Setting $u_n+p_n = w_{\widetilde{N}+1}+r_{\widetilde{N}+1}$, $S_n$ satisfies the recursion relation
\begin{align}
S_n 
\Big|_{u_n+p_n=w_{\widetilde{N}+1}+r_{\widetilde{N}+1}}
&=
[2p_n]^{\frac{1}{2}} [2r_{\widetilde{N}+1}]^{\frac{1}{2}}
\prod_{1\leq j < \widetilde{N}+1}
[w_j-w_{\widetilde{N}+1}+r_j-r_{\widetilde{N}+1}+2p_n]
\label{Srec0-tf}
\\
&
\times
\prod_{\widetilde{N}+1 < j \leq M}
[w_{\widetilde{N}+1}-w_j+r_j+r_{\widetilde{N}+1}]
S_{n-1},
\nonumber
\end{align}
where we have abbreviated $S_{n-1} = S_{n-1}(\{u,p\}_{n-1},\{v,q\}_N,\{w,r\}_M)$.
}
\end{conditions}

In addition, we have the supplementary condition

\begin{conditions}
{\rm $S_0$ and $Z_N$ are related via the equation
\begin{align}
S_0\Big(\{v,q\}_N,\{w,r\}_M\Big)
=
\prod_{j=1}^{N}
\prod_{k=N+1}^{M}
[v_j-w_k+q_j-r_k]
Z_N\Big(\{v,q\}_N,\{w,r\}_N\Big).
\label{S0-tf}
\end{align}

}
\end{conditions}
}
\end{lemma}

\begin{proof}
The proof of properties {\bf 1}--{\bf 4} is analogous to the proof of lemma 9. There, we presented an algebraic proof of the properties. Here, we outline a less technical graphical proof.

\setcounter{conditions}{0}
\begin{conditions}
{\rm
For any $\widetilde{N}+1 < j \leq M$, multiplying $S_n(\{u,p\}_n,\{v,q\}_N,\{w,r\}_M)$ by the weight $a_{+}(w_{j},r_{j},w_{j-1},r_{j-1})$ is equivalent to attaching an $a_{+}$ vertex at the base of the lattice, as shown in figure \ref{multbya}.

\begin{figure}[H]

\begin{center}
\begin{minipage}{4.3in}

\setlength{\unitlength}{0.0003cm}
\begin{picture}(20000,22000)(-9000,-11000)

\put(-8000,-4000){\fbox{$u,p$}}

\path(-2000,-6000)(20000,-6000)
\put(-3250,-6000){\tiny{$n$}}
\path(-4750,-6000)(-3750,-6000)
\whiten\path(-4000,-5750)(-4000,-6250)(-3750,-6000)(-4000,-5750)
\blacken\path(-750,-5750)(-750,-6250)(-1250,-6000)(-750,-5750)
\blacken\path(18750,-5750)(18750,-6250)(19250,-6000)(18750,-5750)

\path(-2000,-4000)(20000,-4000)
\path(-4750,-4000)(-3750,-4000)
\whiten\path(-4000,-3750)(-4000,-4250)(-3750,-4000)(-4000,-3750)
\blacken\path(-750,-3750)(-750,-4250)(-1250,-4000)(-750,-3750)
\blacken\path(18750,-3750)(18750,-4250)(19250,-4000)(18750,-3750)

\path(-2000,-2000)(20000,-2000)
\put(-3250,-2000){\tiny{1}}
\path(-4750,-2000)(-3750,-2000)
\whiten\path(-4000,-1750)(-4000,-2250)(-3750,-2000)(-4000,-1750)
\blacken\path(-750,-1750)(-750,-2250)(-1250,-2000)(-750,-1750)
\blacken\path(18750,-1750)(18750,-2250)(19250,-2000)(18750,-1750)

\put(-8000,5000){\fbox{$v,q$}}

\path(-2000,0)(20000,0)
\put(-3250,0){\tiny{$N$}}
\path(-4750,0)(-3750,0)
\whiten\path(-4000,250)(-4000,-250)(-3750,0)(-4000,250)
\blacken\path(-1250,250)(-1250,-250)(-750,0)(-1250,250)
\blacken\path(19250,250)(19250,-250)(18750,0)(19250,250)

\path(-2000,2000)(20000,2000)
\path(-4750,2000)(-3750,2000)
\whiten\path(-4000,2250)(-4000,1750)(-3750,2000)(-4000,2250)
\blacken\path(-1250,2250)(-1250,1750)(-750,2000)(-1250,2250)
\blacken\path(19250,2250)(19250,1750)(18750,2000)(19250,2250)

\path(-2000,4000)(20000,4000)
\path(-4750,4000)(-3750,4000)
\whiten\path(-4000,4250)(-4000,3750)(-3750,4000)(-4000,4250)
\blacken\path(-1250,4250)(-1250,3750)(-750,4000)(-1250,4250)
\blacken\path(19250,4250)(19250,3750)(18750,4000)(19250,4250)

\path(-2000,6000)(20000,6000)
\path(-4750,6000)(-3750,6000)
\whiten\path(-4000,6250)(-4000,5750)(-3750,6000)(-4000,6250)
\blacken\path(-1250,6250)(-1250,5750)(-750,6000)(-1250,6250)
\blacken\path(19250,6250)(19250,5750)(18750,6000)(19250,6250)

\path(-2000,8000)(20000,8000)
\path(-4750,8000)(-3750,8000)
\whiten\path(-4000,8250)(-4000,7750)(-3750,8000)(-4000,8250)
\blacken\path(-1250,8250)(-1250,7750)(-750,8000)(-1250,8250)
\blacken\path(19250,8250)(19250,7750)(18750,8000)(19250,8250)

\path(-2000,10000)(20000,10000)
\put(-3250,10000){\tiny{$1$}}
\path(-4750,10000)(-3750,10000)
\whiten\path(-4000,10250)(-4000,9750)(-3750,10000)(-4000,10250)
\blacken\path(-1250,10250)(-1250,9750)(-750,10000)(-1250,10250)
\blacken\path(19250,10250)(19250,9750)(18750,10000)(19250,10250)

%%%%%%%%%%%

\path(0,-8000)(0,12000)
\put(-200,-10750){\tiny{1}}
\path(0,-12500)(0,-11500)
\whiten\path(-250,-11750)(250,-11750)(0,-11500)(-250,-11750)
\blacken\path(-250,-6750)(250,-6750)(0,-7250)(-250,-6750)
\blacken\path(-250,10750)(250,10750)(0,11250)(-250,10750)

\path(2000,-8000)(2000,12000)
\path(2000,-12500)(2000,-11500)
\whiten\path(1750,-11750)(2250,-11750)(2000,-11500)(1750,-11750)
\blacken\path(1750,-6750)(2250,-6750)(2000,-7250)(1750,-6750)
\blacken\path(1750,10750)(2250,10750)(2000,11250)(1750,10750)

\path(4000,-8000)(4000,12000)
\put(3800,-10750){\tiny{$\widetilde{N}$}}
\path(4000,-12500)(4000,-11500)
\whiten\path(3750,-11750)(4250,-11750)(4000,-11500)(3750,-11750)
\blacken\path(3750,-6750)(4250,-6750)(4000,-7250)(3750,-6750)
\blacken\path(3750,10750)(4250,10750)(4000,11250)(3750,10750)

\path(6000,-8000)(6000,12000)
\put(5800,-10750){\tiny{$\widetilde{N}+1$}}
\path(6000,-12500)(6000,-11500)
\whiten\path(5750,-11750)(6250,-11750)(6000,-11500)(5750,-11750)
\blacken\path(5750,-7250)(6250,-7250)(6000,-6750)(5750,-7250)
\blacken\path(5750,10750)(6250,10750)(6000,11250)(5750,10750)

\path(8000,-8000)(8000,12000)
\path(8000,-12500)(8000,-11500)
\whiten\path(7750,-11750)(8250,-11750)(8000,-11500)(7750,-11750)
\blacken\path(7750,-7250)(8250,-7250)(8000,-6750)(7750,-7250)
\blacken\path(7750,10750)(8250,10750)(8000,11250)(7750,10750)

\path(10000,-8000)(10000,12000)
\path(10000,-12500)(10000,-11500)
\whiten\path(9750,-11750)(10250,-11750)(10000,-11500)(9750,-11750)
\blacken\path(9750,-7250)(10250,-7250)(10000,-6750)(9750,-7250)
\blacken\path(9750,10750)(10250,10750)(10000,11250)(9750,10750)

\path(12000,-12500)(12000,-11500)
\whiten\path(11750,-11750)(12250,-11750)(12000,-11500)(11750,-11750)
\put(13800,-10750){\tiny{$j-1$}}
\path(12000,-7250)(12000,12000)
\path(12000,-7250)(14000,-9250)
\blacken\path(13750,-9750)(14250,-9750)(14000,-9250)(13750,-9750)
\shade\path(11750,-7250)(12250,-7250)(12000,-6750)(11750,-7250)
\blacken\path(11750,10750)(12250,10750)(12000,11250)(11750,10750)

\path(14000,-12500)(14000,-11500)
\whiten\path(13750,-11750)(14250,-11750)(14000,-11500)(13750,-11750)
\put(11800,-10750){\tiny{$j$}}
\path(14000,-7250)(14000,12000)
\path(14000,-7250)(12000,-9250)
\blacken\path(11750,-9750)(12250,-9750)(12000,-9250)(11750,-9750)
\shade\path(13750,-7250)(14250,-7250)(14000,-6750)(13750,-7250)
\blacken\path(13750,10750)(14250,10750)(14000,11250)(13750,10750)

\path(16000,-8000)(16000,12000)
\path(16000,-12500)(16000,-11500)
\whiten\path(15750,-11750)(16250,-11750)(16000,-11500)(15750,-11750)
\blacken\path(15750,-7250)(16250,-7250)(16000,-6750)(15750,-7250)
\blacken\path(15750,10750)(16250,10750)(16000,11250)(15750,10750)

\path(18000,-8000)(18000,12000)
\put(17800,-10750){\tiny{$M$}}
\path(18000,-12500)(18000,-11500)
\whiten\path(17750,-11750)(18250,-11750)(18000,-11500)(17750,-11750)
\blacken\path(17750,-7250)(18250,-7250)(18000,-6750)(17750,-7250)
\blacken\path(17750,10750)(18250,10750)(18000,11250)(17750,10750)

\put(20000,-9000){\fbox{$w,r$}}

\end{picture}

\end{minipage}
\end{center}

\caption[Attaching an $a_{+}(w_j,r_j,w_{j-1},r_{j-1})$ vertex to the $S_n$ lattice]{Attaching an $a_{+}(w_j,r_j,w_{j-1},r_{j-1})$ vertex to the $S_n$ lattice. The points marked with grey arrows are considered to be summed over all arrow configurations, but the only non-zero configuration is the one shown.}

\label{multbya}
\end{figure}
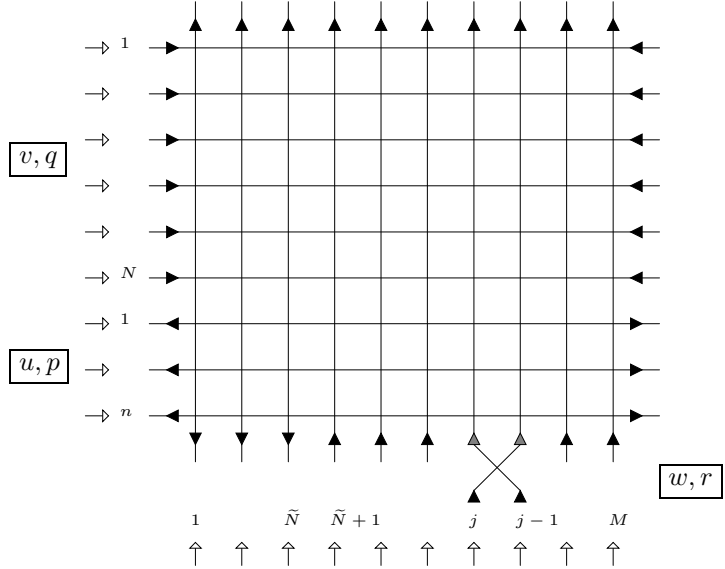

The attached vertex can be translated vertically through the lattice using the graphical version of the Yang-Baxter equation, as given by figure \ref{ybtf}. It ultimately emerges from the top of the lattice, still as an $a_{+}(w_{j},r_{j},w_{j-1},r_{j-1})$ vertex, and the $(j-1)^{\rm th}$ and $j^{\rm th}$ lattice columns are swapped in the process. The result of this procedure is shown in figure \ref{extracta}.

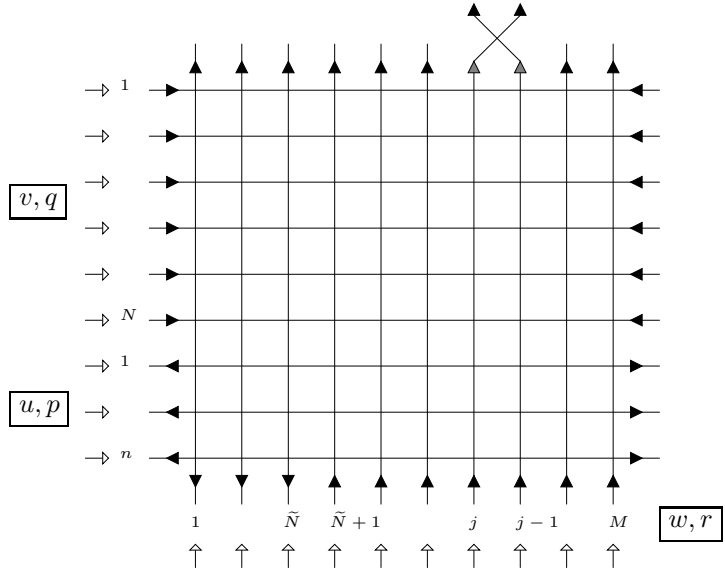
\begin{figure}[H]

\begin{center}
\begin{minipage}{4.3in}

\setlength{\unitlength}{0.00030cm}
\begin{picture}(20000,22000)(-9000,-10000)

\put(-8000,-4000){\fbox{$u,p$}}

\path(-2000,-6000)(20000,-6000)
\put(-3250,-6000){\tiny{$n$}}
\path(-4750,-6000)(-3750,-6000)
\whiten\path(-4000,-5750)(-4000,-6250)(-3750,-6000)(-4000,-5750)
\blacken\path(-750,-5750)(-750,-6250)(-1250,-6000)(-750,-5750)
\blacken\path(18750,-5750)(18750,-6250)(19250,-6000)(18750,-5750)

\path(-2000,-4000)(20000,-4000)
\path(-4750,-4000)(-3750,-4000)
\whiten\path(-4000,-3750)(-4000,-4250)(-3750,-4000)(-4000,-3750)
\blacken\path(-750,-3750)(-750,-4250)(-1250,-4000)(-750,-3750)
\blacken\path(18750,-3750)(18750,-4250)(19250,-4000)(18750,-3750)

\path(-2000,-2000)(20000,-2000)
\put(-3250,-2000){\tiny{1}}
\path(-4750,-2000)(-3750,-2000)
\whiten\path(-4000,-1750)(-4000,-2250)(-3750,-2000)(-4000,-1750)
\blacken\path(-750,-1750)(-750,-2250)(-1250,-2000)(-750,-1750)
\blacken\path(18750,-1750)(18750,-2250)(19250,-2000)(18750,-1750)

\put(-8000,5000){\fbox{$v,q$}}

\path(-2000,0)(20000,0)
\put(-3250,0){\tiny{$N$}}
\path(-4750,0)(-3750,0)
\whiten\path(-4000,250)(-4000,-250)(-3750,0)(-4000,250)
\blacken\path(-1250,250)(-1250,-250)(-750,0)(-1250,250)
\blacken\path(19250,250)(19250,-250)(18750,0)(19250,250)

\path(-2000,2000)(20000,2000)
\path(-4750,2000)(-3750,2000)
\whiten\path(-4000,2250)(-4000,1750)(-3750,2000)(-4000,2250)
\blacken\path(-1250,2250)(-1250,1750)(-750,2000)(-1250,2250)
\blacken\path(19250,2250)(19250,1750)(18750,2000)(19250,2250)

\path(-2000,4000)(20000,4000)
\path(-4750,4000)(-3750,4000)
\whiten\path(-4000,4250)(-4000,3750)(-3750,4000)(-4000,4250)
\blacken\path(-1250,4250)(-1250,3750)(-750,4000)(-1250,4250)
\blacken\path(19250,4250)(19250,3750)(18750,4000)(19250,4250)

\path(-2000,6000)(20000,6000)
\path(-4750,6000)(-3750,6000)
\whiten\path(-4000,6250)(-4000,5750)(-3750,6000)(-4000,6250)
\blacken\path(-1250,6250)(-1250,5750)(-750,6000)(-1250,6250)
\blacken\path(19250,6250)(19250,5750)(18750,6000)(19250,6250)

\path(-2000,8000)(20000,8000)
\path(-4750,8000)(-3750,8000)
\whiten\path(-4000,8250)(-4000,7750)(-3750,8000)(-4000,8250)
\blacken\path(-1250,8250)(-1250,7750)(-750,8000)(-1250,8250)
\blacken\path(19250,8250)(19250,7750)(18750,8000)(19250,8250)

\path(-2000,10000)(20000,10000)
\put(-3250,10000){\tiny{1}}
\path(-4750,10000)(-3750,10000)
\whiten\path(-4000,10250)(-4000,9750)(-3750,10000)(-4000,10250)
\blacken\path(-1250,10250)(-1250,9750)(-750,10000)(-1250,10250)
\blacken\path(19250,10250)(19250,9750)(18750,10000)(19250,10250)

%%%%%%%%%%%

\path(0,-8000)(0,12000)
\put(-200,-9000){\tiny{1}}
\path(0,-10750)(0,-9750)
\whiten\path(-250,-10000)(250,-10000)(0,-9750)(-250,-10000)
\blacken\path(-250,-6750)(250,-6750)(0,-7250)(-250,-6750)
\blacken\path(-250,10750)(250,10750)(0,11250)(-250,10750)

\path(2000,-8000)(2000,12000)
\path(2000,-10750)(2000,-9750)
\whiten\path(1750,-10000)(2250,-10000)(2000,-9750)(1750,-10000)
\blacken\path(1750,-6750)(2250,-6750)(2000,-7250)(1750,-6750)
\blacken\path(1750,10750)(2250,10750)(2000,11250)(1750,10750)

\path(4000,-8000)(4000,12000)
\put(3800,-9000){\tiny{$\widetilde{N}$}}
\path(4000,-10750)(4000,-9750)
\whiten\path(3750,-10000)(4250,-10000)(4000,-9750)(3750,-10000)
\blacken\path(3750,-6750)(4250,-6750)(4000,-7250)(3750,-6750)
\blacken\path(3750,10750)(4250,10750)(4000,11250)(3750,10750)

\path(6000,-8000)(6000,12000)
\put(5800,-9000){\tiny{$\widetilde{N}+1$}}
\path(6000,-10750)(6000,-9750)
\whiten\path(5750,-10000)(6250,-10000)(6000,-9750)(5750,-10000)
\blacken\path(5750,-7250)(6250,-7250)(6000,-6750)(5750,-7250)
\blacken\path(5750,10750)(6250,10750)(6000,11250)(5750,10750)

\path(8000,-8000)(8000,12000)
\path(8000,-10750)(8000,-9750)
\whiten\path(7750,-10000)(8250,-10000)(8000,-9750)(7750,-10000)
\blacken\path(7750,-7250)(8250,-7250)(8000,-6750)(7750,-7250)
\blacken\path(7750,10750)(8250,10750)(8000,11250)(7750,10750)

\path(10000,-8000)(10000,12000)
\path(10000,-10750)(10000,-9750)
\whiten\path(9750,-10000)(10250,-10000)(10000,-9750)(9750,-10000)
\blacken\path(9750,-7250)(10250,-7250)(10000,-6750)(9750,-7250)
\blacken\path(9750,10750)(10250,10750)(10000,11250)(9750,10750)

\path(12000,-8000)(12000,10750)
\put(11800,-9000){\tiny{$j$}}
\path(12000,-10750)(12000,-9750)
\whiten\path(11750,-10000)(12250,-10000)(12000,-9750)(11750,-10000)
\blacken\path(11750,-7250)(12250,-7250)(12000,-6750)(11750,-7250)
\shade\path(11750,10750)(12250,10750)(12000,11250)(11750,10750)
\path(12000,11250)(14000,13250)
\blacken\path(13750,13250)(14250,13250)(14000,13750)(13750,13250)

\path(14000,-8000)(14000,10750)
\put(13800,-9000){\tiny{$j-1$}}
\path(14000,-10750)(14000,-9750)
\whiten\path(13750,-10000)(14250,-10000)(14000,-9750)(13750,-10000)
\blacken\path(13750,-7250)(14250,-7250)(14000,-6750)(13750,-7250)
\shade\path(13750,10750)(14250,10750)(14000,11250)(13750,10750)
\path(14000,11250)(12000,13250)
\blacken\path(11750,13250)(12250,13250)(12000,13750)(11750,13250)

\path(16000,-8000)(16000,12000)
\path(16000,-10750)(16000,-9750)
\whiten\path(15750,-10000)(16250,-10000)(16000,-9750)(15750,-10000)
\blacken\path(15750,-7250)(16250,-7250)(16000,-6750)(15750,-7250)
\blacken\path(15750,10750)(16250,10750)(16000,11250)(15750,10750)

\path(18000,-8000)(18000,12000)
\put(17800,-9000){\tiny{$M$}}
\path(18000,-10750)(18000,-9750)
\whiten\path(17750,-10000)(18250,-10000)(18000,-9750)(17750,-10000)
\blacken\path(17750,-7250)(18250,-7250)(18000,-6750)(17750,-7250)
\blacken\path(17750,10750)(18250,10750)(18000,11250)(17750,10750)

\put(20000,-9000){\fbox{$w,r$}}

\end{picture}

\end{minipage}
\end{center}

\caption[Extracting the $a_{+}(w_j,r_j,w_{j-1},r_{j-1})$ vertex from the $S_n$ lattice]{Extracting the $a_{+}(w_j,r_j,w_{j-1},r_{j-1})$ vertex from the $S_n$ lattice. Once again, the grey arrows indicate the only surviving configuration in the summation at those points.}

\label{extracta}
\end{figure}

Cancelling the common factor $a_{+}(w_j,r_j,w_{j-1},r_{j-1})$ from figures \ref{multbya} and \ref{extracta}, we see that $S_n$ is invariant under swapping the $(j-1)^{\rm th}$ and $j^{\rm th}$ lattice columns, for all $\widetilde{N}+1 < j \leq M$. An arbitrary permutation of the lattice columns is just a composition of such swaps. Therefore $S_n$ is invariant under permuting its $j^{\rm th}$ and $k^{\rm th}$ columns, for all $\widetilde{N}+1 \leq j,k \leq M$. 
}
\end{conditions}

\begin{conditions}
{\rm
Consider the graphical representation of the scalar product $S_n$, as given below:

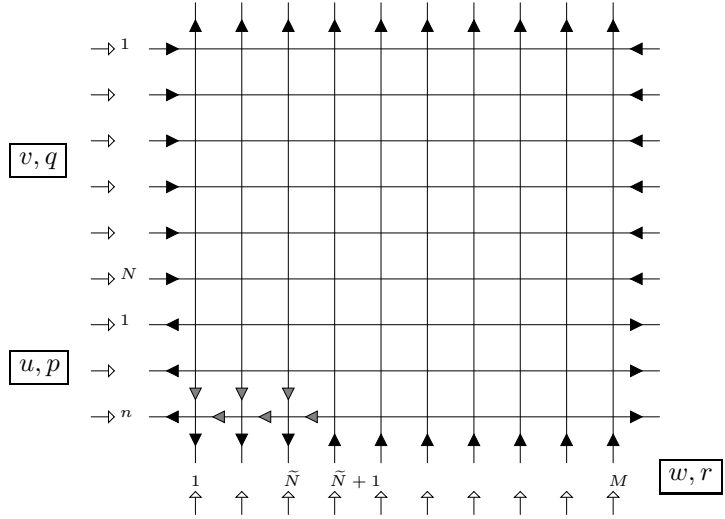
\begin{figure}[H]

\begin{center}
\begin{minipage}{4.3in}

\setlength{\unitlength}{0.0003cm}
\begin{picture}(20000,20000)(-9000,-9000)

\put(-8000,-4000){\fbox{$u,p$}}

\path(-2000,-6000)(20000,-6000)
\put(-3250,-6000){\tiny$n$}
\path(-4500,-6000)(-3500,-6000)
\whiten\path(-3750,-5750)(-3750,-6250)(-3500,-6000)(-3750,-5750)
\blacken\path(-750,-5750)(-750,-6250)(-1250,-6000)(-750,-5750)
\shade\path(1250,-5750)(1250,-6250)(750,-6000)(1250,-5750)
\shade\path(3250,-5750)(3250,-6250)(2750,-6000)(3250,-5750)
\shade\path(5250,-5750)(5250,-6250)(4750,-6000)(5250,-5750)
\blacken\path(18750,-5750)(18750,-6250)(19250,-6000)(18750,-5750)

\path(-2000,-4000)(20000,-4000)
\path(-4500,-4000)(-3500,-4000)
\whiten\path(-3750,-3750)(-3750,-4250)(-3500,-4000)(-3750,-3750)
\blacken\path(-750,-3750)(-750,-4250)(-1250,-4000)(-750,-3750)
\blacken\path(18750,-3750)(18750,-4250)(19250,-4000)(18750,-3750)

\path(-2000,-2000)(20000,-2000)
\put(-3250,-2000){\tiny$1$}
\path(-4500,-2000)(-3500,-2000)
\whiten\path(-3750,-1750)(-3750,-2250)(-3500,-2000)(-3750,-1750)
\blacken\path(-750,-1750)(-750,-2250)(-1250,-2000)(-750,-1750)
\blacken\path(18750,-1750)(18750,-2250)(19250,-2000)(18750,-1750)

\put(-8000,5000){\fbox{$v,q$}}

\path(-2000,0)(20000,0)
\put(-3250,0){\tiny$N$}
\path(-4500,0)(-3500,0)
\whiten\path(-3750,250)(-3750,-250)(-3500,0)(-3750,250)
\blacken\path(-1250,250)(-1250,-250)(-750,0)(-1250,250)
\blacken\path(19250,250)(19250,-250)(18750,0)(19250,250)

\path(-2000,2000)(20000,2000)
\path(-4500,2000)(-3500,2000)
\whiten\path(-3750,2250)(-3750,1750)(-3500,2000)(-3750,2250)
\blacken\path(-1250,2250)(-1250,1750)(-750,2000)(-1250,2250)
\blacken\path(19250,2250)(19250,1750)(18750,2000)(19250,2250)

\path(-2000,4000)(20000,4000)
\path(-4500,4000)(-3500,4000)
\whiten\path(-3750,4250)(-3750,3750)(-3500,4000)(-3750,4250)
\blacken\path(-1250,4250)(-1250,3750)(-750,4000)(-1250,4250)
\blacken\path(19250,4250)(19250,3750)(18750,4000)(19250,4250)

\path(-2000,6000)(20000,6000)
\path(-4500,6000)(-3500,6000)
\whiten\path(-3750,6250)(-3750,5750)(-3500,6000)(-3750,6250)
\blacken\path(-1250,6250)(-1250,5750)(-750,6000)(-1250,6250)
\blacken\path(19250,6250)(19250,5750)(18750,6000)(19250,6250)

\path(-2000,8000)(20000,8000)
\path(-4500,8000)(-3500,8000)
\whiten\path(-3750,8250)(-3750,7750)(-3500,8000)(-3750,8250)
\blacken\path(-1250,8250)(-1250,7750)(-750,8000)(-1250,8250)
\blacken\path(19250,8250)(19250,7750)(18750,8000)(19250,8250)

\path(-2000,10000)(20000,10000)
\put(-3250,10000){\tiny$1$}
\path(-4500,10000)(-3500,10000)
\whiten\path(-3750,10250)(-3750,9750)(-3500,10000)(-3750,10250)
\blacken\path(-1250,10250)(-1250,9750)(-750,10000)(-1250,10250)
\blacken\path(19250,10250)(19250,9750)(18750,10000)(19250,10250)

%%%%%%%%%%%

\path(0,-8000)(0,12000)
\put(-200,-9000){\tiny$1$}
\path(0,-10250)(0,-9250)
\whiten\path(-250,-9500)(250,-9500)(0,-9250)(-250,-9500)
\blacken\path(-250,-6750)(250,-6750)(0,-7250)(-250,-6750)
\shade\path(-250,-4750)(250,-4750)(0,-5250)(-250,-4750)
\blacken\path(-250,10750)(250,10750)(0,11250)(-250,10750)

\path(2000,-8000)(2000,12000)
\path(2000,-10250)(2000,-9250)
\whiten\path(1750,-9500)(2250,-9500)(2000,-9250)(1750,-9500)
\blacken\path(1750,-6750)(2250,-6750)(2000,-7250)(1750,-6750)
\shade\path(1750,-4750)(2250,-4750)(2000,-5250)(1750,-4750)
\blacken\path(1750,10750)(2250,10750)(2000,11250)(1750,10750)

\path(4000,-8000)(4000,12000)
\put(3800,-9000){\tiny$\widetilde{N}$}
\path(4000,-10250)(4000,-9250)
\whiten\path(3750,-9500)(4250,-9500)(4000,-9250)(3750,-9500)
\blacken\path(3750,-6750)(4250,-6750)(4000,-7250)(3750,-6750)
\shade\path(3750,-4750)(4250,-4750)(4000,-5250)(3750,-4750)
\blacken\path(3750,10750)(4250,10750)(4000,11250)(3750,10750)

\path(6000,-8000)(6000,12000)
\put(5800,-9000){\tiny$\widetilde{N}+1$}
\path(6000,-10250)(6000,-9250)
\whiten\path(5750,-9500)(6250,-9500)(6000,-9250)(5750,-9500)
\blacken\path(5750,-7250)(6250,-7250)(6000,-6750)(5750,-7250)
\blacken\path(5750,10750)(6250,10750)(6000,11250)(5750,10750)

\path(8000,-8000)(8000,12000)
\path(8000,-10250)(8000,-9250)
\whiten\path(7750,-9500)(8250,-9500)(8000,-9250)(7750,-9500)
\blacken\path(7750,-7250)(8250,-7250)(8000,-6750)(7750,-7250)
\blacken\path(7750,10750)(8250,10750)(8000,11250)(7750,10750)

\path(10000,-8000)(10000,12000)
\path(10000,-10250)(10000,-9250)
\whiten\path(9750,-9500)(10250,-9500)(10000,-9250)(9750,-9500)
\blacken\path(9750,-7250)(10250,-7250)(10000,-6750)(9750,-7250)
\blacken\path(9750,10750)(10250,10750)(10000,11250)(9750,10750)

\path(12000,-8000)(12000,12000)
\path(12000,-10250)(12000,-9250)
\whiten\path(11750,-9500)(12250,-9500)(12000,-9250)(11750,-9500)
\blacken\path(11750,-7250)(12250,-7250)(12000,-6750)(11750,-7250)
\blacken\path(11750,10750)(12250,10750)(12000,11250)(11750,10750)

\path(14000,-8000)(14000,12000)
\path(14000,-10250)(14000,-9250)
\whiten\path(13750,-9500)(14250,-9500)(14000,-9250)(13750,-9500)
\blacken\path(13750,-7250)(14250,-7250)(14000,-6750)(13750,-7250)
\blacken\path(13750,10750)(14250,10750)(14000,11250)(13750,10750)

\path(16000,-8000)(16000,12000)
\path(16000,-10250)(16000,-9250)
\whiten\path(15750,-9500)(16250,-9500)(16000,-9250)(15750,-9500)
\blacken\path(15750,-7250)(16250,-7250)(16000,-6750)(15750,-7250)
\blacken\path(15750,10750)(16250,10750)(16000,11250)(15750,10750)

\path(18000,-8000)(18000,12000)
\put(17800,-9000){\tiny$M$}
\path(18000,-10250)(18000,-9250)
\whiten\path(17750,-9500)(18250,-9500)(18000,-9250)(17750,-9500)
\blacken\path(17750,-7250)(18250,-7250)(18000,-6750)(17750,-7250)
\blacken\path(17750,10750)(18250,10750)(18000,11250)(17750,10750)

\put(20000,-8750){\fbox{$w,r$}}

\end{picture}

\end{minipage}
\end{center}

\caption[Lattice representation of $S_n$, with frozen vertices included]{Lattice representation of $S_n$, with frozen vertices included. The grey arrows indicate points which are summed over all configurations, but whose only non-zero configuration is the one shown.}

\label{snfroze}
\end{figure} 

We examine the final row of this lattice, through which the variable $u_n$ flows. Every non-zero configuration of this row contains a $c_{-}(u_n,p_n,w_j,r_j)$ vertex, which by (\ref{c-tf}) does not depend on $u_n$, and $M-1$ other vertices which are trigonometric polynomials of degree 1 in $u_n$. It follows that $S_n$ is a trigonometric polynomial of degree $M-1$ in $u_n$. 

Furthermore, all surviving configurations of the final row contain the $\widetilde{N}$ vertices as shown in figure \ref{snfroze}. Consequentially, $S_n$ contains the factor
\begin{align}
\prod_{j=1}^{\widetilde{N}} a_{-}(u_n,p_n,w_j,r_j)
=
\prod_{j=1}^{\widetilde{N}} [w_j-u_n+p_n+r_j],
\end{align} 
which gives rise to zeros at $u_n = p_n+w_j+r_j$ for all $1\leq j \leq \widetilde{N}$.
  
}
\end{conditions}

\begin{conditions}
{\rm
Consider the vertex at the intersection of the $u_n$ and $w_{\widetilde{N}+1}$ lines in figure \ref{snfroze}. In any given lattice configuration, this can be of type $b_{-}(u_n,p_n,w_{\widetilde{N}+1},r_{\widetilde{N}+1})$ or $c_{-}(u_n,p_n,w_{\widetilde{N}+1},r_{\widetilde{N}+1})$. Setting $u_n+p_n = w_{\widetilde{N}+1}+r_{\widetilde{N}+1}$ cancels all terms containing $b_{-}(u_n,p_n,w_{\widetilde{N}+1},r_{\widetilde{N}+1})$, and freezes the entire final row of the lattice to the configuration below:

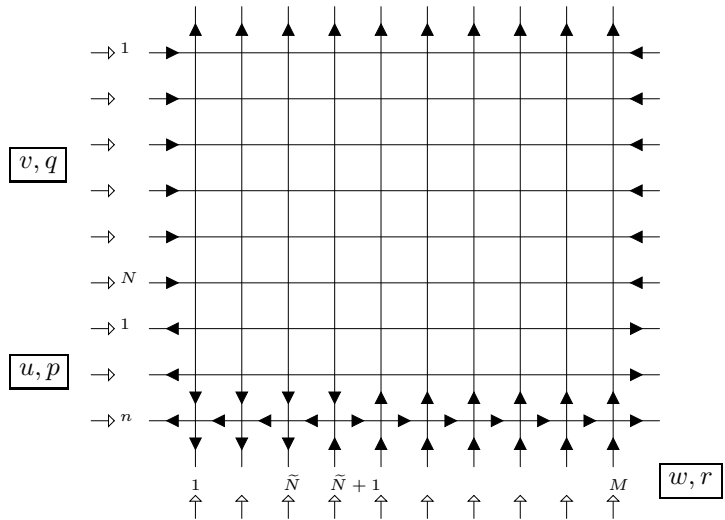
\begin{figure}[H]

\begin{center}
\begin{minipage}{4.3in}

\setlength{\unitlength}{0.0003cm}
\begin{picture}(20000,20000)(-9000,-9000)

\put(-8000,-4000){\fbox{$u,p$}}

\path(-2000,-6000)(20000,-6000)
\put(-3250,-6000){\tiny$n$}
\path(-4500,-6000)(-3500,-6000)
\whiten\path(-3750,-5750)(-3750,-6250)(-3500,-6000)(-3750,-5750)
\blacken\path(-750,-5750)(-750,-6250)(-1250,-6000)(-750,-5750)
\blacken\path(1250,-5750)(1250,-6250)(750,-6000)(1250,-5750)
\blacken\path(3250,-5750)(3250,-6250)(2750,-6000)(3250,-5750)
\blacken\path(5250,-5750)(5250,-6250)(4750,-6000)(5250,-5750)
\blacken\path(6750,-5750)(6750,-6250)(7250,-6000)(6750,-5750)
\blacken\path(8750,-5750)(8750,-6250)(9250,-6000)(8750,-5750)
\blacken\path(10750,-5750)(10750,-6250)(11250,-6000)(10750,-5750)
\blacken\path(12750,-5750)(12750,-6250)(13250,-6000)(12750,-5750)
\blacken\path(14750,-5750)(14750,-6250)(15250,-6000)(14750,-5750)
\blacken\path(16750,-5750)(16750,-6250)(17250,-6000)(16750,-5750)
\blacken\path(18750,-5750)(18750,-6250)(19250,-6000)(18750,-5750)

\path(-2000,-4000)(20000,-4000)
\path(-4500,-4000)(-3500,-4000)
\whiten\path(-3750,-3750)(-3750,-4250)(-3500,-4000)(-3750,-3750)
\blacken\path(-750,-3750)(-750,-4250)(-1250,-4000)(-750,-3750)
\blacken\path(18750,-3750)(18750,-4250)(19250,-4000)(18750,-3750)

\path(-2000,-2000)(20000,-2000)
\put(-3250,-2000){\tiny$1$}
\path(-4500,-2000)(-3500,-2000)
\whiten\path(-3750,-1750)(-3750,-2250)(-3500,-2000)(-3750,-1750)
\blacken\path(-750,-1750)(-750,-2250)(-1250,-2000)(-750,-1750)
\blacken\path(18750,-1750)(18750,-2250)(19250,-2000)(18750,-1750)

\put(-8000,5000){\fbox{$v,q$}}

\path(-2000,0)(20000,0)
\put(-3250,0){\tiny$N$}
\path(-4500,0)(-3500,0)
\whiten\path(-3750,250)(-3750,-250)(-3500,0)(-3750,250)
\blacken\path(-1250,250)(-1250,-250)(-750,0)(-1250,250)
\blacken\path(19250,250)(19250,-250)(18750,0)(19250,250)

\path(-2000,2000)(20000,2000)
\path(-4500,2000)(-3500,2000)
\whiten\path(-3750,2250)(-3750,1750)(-3500,2000)(-3750,2250)
\blacken\path(-1250,2250)(-1250,1750)(-750,2000)(-1250,2250)
\blacken\path(19250,2250)(19250,1750)(18750,2000)(19250,2250)

\path(-2000,4000)(20000,4000)
\path(-4500,4000)(-3500,4000)
\whiten\path(-3750,4250)(-3750,3750)(-3500,4000)(-3750,4250)
\blacken\path(-1250,4250)(-1250,3750)(-750,4000)(-1250,4250)
\blacken\path(19250,4250)(19250,3750)(18750,4000)(19250,4250)

\path(-2000,6000)(20000,6000)
\path(-4500,6000)(-3500,6000)
\whiten\path(-3750,6250)(-3750,5750)(-3500,6000)(-3750,6250)
\blacken\path(-1250,6250)(-1250,5750)(-750,6000)(-1250,6250)
\blacken\path(19250,6250)(19250,5750)(18750,6000)(19250,6250)

\path(-2000,8000)(20000,8000)
\path(-4500,8000)(-3500,8000)
\whiten\path(-3750,8250)(-3750,7750)(-3500,8000)(-3750,8250)
\blacken\path(-1250,8250)(-1250,7750)(-750,8000)(-1250,8250)
\blacken\path(19250,8250)(19250,7750)(18750,8000)(19250,8250)

\path(-2000,10000)(20000,10000)
\put(-3250,10000){\tiny$1$}
\path(-4500,10000)(-3500,10000)
\whiten\path(-3750,10250)(-3750,9750)(-3500,10000)(-3750,10250)
\blacken\path(-1250,10250)(-1250,9750)(-750,10000)(-1250,10250)
\blacken\path(19250,10250)(19250,9750)(18750,10000)(19250,10250)

%%%%%%%%%%%

\path(0,-8000)(0,12000)
\put(-200,-9000){\tiny$1$}
\path(0,-10250)(0,-9250)
\whiten\path(-250,-9500)(250,-9500)(0,-9250)(-250,-9500)
\blacken\path(-250,-6750)(250,-6750)(0,-7250)(-250,-6750)
\blacken\path(-250,-4750)(250,-4750)(0,-5250)(-250,-4750)
\blacken\path(-250,10750)(250,10750)(0,11250)(-250,10750)

\path(2000,-8000)(2000,12000)
\path(2000,-10250)(2000,-9250)
\whiten\path(1750,-9500)(2250,-9500)(2000,-9250)(1750,-9500)
\blacken\path(1750,-6750)(2250,-6750)(2000,-7250)(1750,-6750)
\blacken\path(1750,-4750)(2250,-4750)(2000,-5250)(1750,-4750)
\blacken\path(1750,10750)(2250,10750)(2000,11250)(1750,10750)

\path(4000,-8000)(4000,12000)
\put(3800,-9000){\tiny$\widetilde{N}$}
\path(4000,-10250)(4000,-9250)
\whiten\path(3750,-9500)(4250,-9500)(4000,-9250)(3750,-9500)
\blacken\path(3750,-6750)(4250,-6750)(4000,-7250)(3750,-6750)
\blacken\path(3750,-4750)(4250,-4750)(4000,-5250)(3750,-4750)
\blacken\path(3750,10750)(4250,10750)(4000,11250)(3750,10750)

\path(6000,-8000)(6000,12000)
\put(5800,-9000){\tiny$\widetilde{N}+1$}
\path(6000,-10250)(6000,-9250)
\whiten\path(5750,-9500)(6250,-9500)(6000,-9250)(5750,-9500)
\blacken\path(5750,-7250)(6250,-7250)(6000,-6750)(5750,-7250)
\blacken\path(5750,-4750)(6250,-4750)(6000,-5250)(5750,-4750)
\blacken\path(5750,10750)(6250,10750)(6000,11250)(5750,10750)

\path(8000,-8000)(8000,12000)
\path(8000,-10250)(8000,-9250)
\whiten\path(7750,-9500)(8250,-9500)(8000,-9250)(7750,-9500)
\blacken\path(7750,-7250)(8250,-7250)(8000,-6750)(7750,-7250)
\blacken\path(7750,-5250)(8250,-5250)(8000,-4750)(7750,-5250)
\blacken\path(7750,10750)(8250,10750)(8000,11250)(7750,10750)

\path(10000,-8000)(10000,12000)
\path(10000,-10250)(10000,-9250)
\whiten\path(9750,-9500)(10250,-9500)(10000,-9250)(9750,-9500)
\blacken\path(9750,-7250)(10250,-7250)(10000,-6750)(9750,-7250)
\blacken\path(9750,-5250)(10250,-5250)(10000,-4750)(9750,-5250)
\blacken\path(9750,10750)(10250,10750)(10000,11250)(9750,10750)

\path(12000,-8000)(12000,12000)
\path(12000,-10250)(12000,-9250)
\whiten\path(11750,-9500)(12250,-9500)(12000,-9250)(11750,-9500)
\blacken\path(11750,-7250)(12250,-7250)(12000,-6750)(11750,-7250)
\blacken\path(11750,-5250)(12250,-5250)(12000,-4750)(11750,-5250)
\blacken\path(11750,10750)(12250,10750)(12000,11250)(11750,10750)

\path(14000,-8000)(14000,12000)
\path(14000,-10250)(14000,-9250)
\whiten\path(13750,-9500)(14250,-9500)(14000,-9250)(13750,-9500)
\blacken\path(13750,-7250)(14250,-7250)(14000,-6750)(13750,-7250)
\blacken\path(13750,-5250)(14250,-5250)(14000,-4750)(13750,-5250)
\blacken\path(13750,10750)(14250,10750)(14000,11250)(13750,10750)

\path(16000,-8000)(16000,12000)
\path(16000,-10250)(16000,-9250)
\whiten\path(15750,-9500)(16250,-9500)(16000,-9250)(15750,-9500)
\blacken\path(15750,-7250)(16250,-7250)(16000,-6750)(15750,-7250)
\blacken\path(15750,-5250)(16250,-5250)(16000,-4750)(15750,-5250)
\blacken\path(15750,10750)(16250,10750)(16000,11250)(15750,10750)

\path(18000,-8000)(18000,12000)
\put(17800,-9000){\tiny$M$}
\path(18000,-10250)(18000,-9250)
\whiten\path(17750,-9500)(18250,-9500)(18000,-9250)(17750,-9500)
\blacken\path(17750,-7250)(18250,-7250)(18000,-6750)(17750,-7250)
\blacken\path(17750,-5250)(18250,-5250)(18000,-4750)(17750,-5250)
\blacken\path(17750,10750)(18250,10750)(18000,11250)(17750,10750)

\put(20000,-8750){\fbox{$w,r$}}

\end{picture}

\end{minipage}
\end{center}

\caption[Freezing the entire last row of the $S_n$ lattice]{Freezing the entire last row of the $S_n$ lattice. The last row of vertices produces the prefactor in (\ref{Srec0-tf}), while the remainder of the lattice represents $S_{n-1}$.}

\label{snfroze2}
\end{figure}

From the diagram we see that setting $u_n+p_n = w_{\widetilde{N}+1}+r_{\widetilde{N}+1}$ reduces $S_n$ to $S_{n-1}$, up to a multiplicative factor. This factor is evaluated by matching each vertex in the final row with its trigonometric weight, giving  
\begin{align}
&
S_n\Big|_{u_n+p_n = w_{\widetilde{N}+1}+r_{\widetilde{N}+1}}
=
\prod_{1\leq j < \widetilde{N}+1}
a_{-}\Big(w_{\widetilde{N}+1}+r_{\widetilde{N}+1}-p_n,p_n,w_j,r_j\Big)
\times
\label{Srec-tf}
\\
&
c_{-}\Big(w_{\widetilde{N}+1}+r_{\widetilde{N}+1}-p_n,p_n,w_{\widetilde{N}+1},r_{\widetilde{N}+1}\Big)
\prod_{\widetilde{N}+1 < j \leq M}
a_{+}\Big(w_{\widetilde{N}+1}+r_{\widetilde{N}+1}-p_n,p_n,w_j,r_j\Big)
S_{n-1}.
\nonumber
\end{align}
Using the explicit formulae (\ref{a-tf}) and (\ref{c-tf}) for the functions appearing in (\ref{Srec-tf}), we obtain the required recursion relation (\ref{Srec0-tf}).
}
\end{conditions}

\begin{conditions}
{\rm
The scalar product $S_0$ is represented by the lattice below:

\begin{figure}[H]

\begin{center}
\begin{minipage}{4.3in}

\setlength{\unitlength}{0.0003cm}
\begin{picture}(20000,15000)(-9000,-3000)

\put(-8000,5000){\fbox{$v,q$}}

\path(-2000,0)(20000,0)
\put(-3250,0){\tiny$N$}
\path(-4500,0)(-3500,0)
\whiten\path(-3750,250)(-3750,-250)(-3500,0)(-3750,250)
\blacken\path(-1250,250)(-1250,-250)(-750,0)(-1250,250)
\blacken\path(11250,250)(11250,-250)(10750,0)(11250,250)
\blacken\path(13250,250)(13250,-250)(12750,0)(13250,250)
\blacken\path(15250,250)(15250,-250)(14750,0)(15250,250)
\blacken\path(17250,250)(17250,-250)(16750,0)(17250,250)
\blacken\path(19250,250)(19250,-250)(18750,0)(19250,250)

\path(-2000,2000)(20000,2000)
\path(-4500,2000)(-3500,2000)
\whiten\path(-3750,2250)(-3750,1750)(-3500,2000)(-3750,2250)
\blacken\path(-1250,2250)(-1250,1750)(-750,2000)(-1250,2250)
\blacken\path(11250,2250)(11250,1750)(10750,2000)(11250,2250)
\blacken\path(13250,2250)(13250,1750)(12750,2000)(13250,2250)
\blacken\path(15250,2250)(15250,1750)(14750,2000)(15250,2250)
\blacken\path(17250,2250)(17250,1750)(16750,2000)(17250,2250)
\blacken\path(19250,2250)(19250,1750)(18750,2000)(19250,2250)

\path(-2000,4000)(20000,4000)
\path(-4500,4000)(-3500,4000)
\whiten\path(-3750,4250)(-3750,3750)(-3500,4000)(-3750,4250)
\blacken\path(-1250,4250)(-1250,3750)(-750,4000)(-1250,4250)
\blacken\path(11250,4250)(11250,3750)(10750,4000)(11250,4250)
\blacken\path(13250,4250)(13250,3750)(12750,4000)(13250,4250)
\blacken\path(15250,4250)(15250,3750)(14750,4000)(15250,4250)
\blacken\path(17250,4250)(17250,3750)(16750,4000)(17250,4250)
\blacken\path(19250,4250)(19250,3750)(18750,4000)(19250,4250)

\path(-2000,6000)(20000,6000)
\path(-4500,6000)(-3500,6000)
\whiten\path(-3750,6250)(-3750,5750)(-3500,6000)(-3750,6250)
\blacken\path(-1250,6250)(-1250,5750)(-750,6000)(-1250,6250)
\blacken\path(11250,6250)(11250,5750)(10750,6000)(11250,6250)
\blacken\path(13250,6250)(13250,5750)(12750,6000)(13250,6250)
\blacken\path(15250,6250)(15250,5750)(14750,6000)(15250,6250)
\blacken\path(17250,6250)(17250,5750)(16750,6000)(17250,6250)
\blacken\path(19250,6250)(19250,5750)(18750,6000)(19250,6250)

\path(-2000,8000)(20000,8000)
\path(-4500,8000)(-3500,8000)
\whiten\path(-3750,8250)(-3750,7750)(-3500,8000)(-3750,8250)
\blacken\path(-1250,8250)(-1250,7750)(-750,8000)(-1250,8250)
\blacken\path(11250,8250)(11250,7750)(10750,8000)(11250,8250)
\blacken\path(13250,8250)(13250,7750)(12750,8000)(13250,8250)
\blacken\path(15250,8250)(15250,7750)(14750,8000)(15250,8250)
\blacken\path(17250,8250)(17250,7750)(16750,8000)(17250,8250)
\blacken\path(19250,8250)(19250,7750)(18750,8000)(19250,8250)

\path(-2000,10000)(20000,10000)
\put(-3250,10000){\tiny$1$}
\path(-4500,10000)(-3500,10000)
\whiten\path(-3750,10250)(-3750,9750)(-3500,10000)(-3750,10250)
\blacken\path(-1250,10250)(-1250,9750)(-750,10000)(-1250,10250)
\blacken\path(11250,10250)(11250,9750)(10750,10000)(11250,10250)
\blacken\path(13250,10250)(13250,9750)(12750,10000)(13250,10250)
\blacken\path(15250,10250)(15250,9750)(14750,10000)(15250,10250)
\blacken\path(17250,10250)(17250,9750)(16750,10000)(17250,10250)
\blacken\path(19250,10250)(19250,9750)(18750,10000)(19250,10250)

%%%%%%%%%%%

\path(0,-2000)(0,12000)
\put(-200,-2750){\tiny$1$}
\path(0,-4250)(0,-3250)
\whiten\path(-250,-3500)(250,-3500)(0,-3250)(-250,-3500)
\blacken\path(-250,-750)(250,-750)(0,-1250)(-250,-750)
\blacken\path(-250,10750)(250,10750)(0,11250)(-250,10750)

\path(2000,-2000)(2000,12000)
\path(2000,-4250)(2000,-3250)
\whiten\path(1750,-3500)(2250,-3500)(2000,-3250)(1750,-3500)
\blacken\path(1750,-750)(2250,-750)(2000,-1250)(1750,-750)
\blacken\path(1750,10750)(2250,10750)(2000,11250)(1750,10750)

\path(4000,-2000)(4000,12000)
\path(4000,-4250)(4000,-3250)
\whiten\path(3750,-3500)(4250,-3500)(4000,-3250)(3750,-3500)
\blacken\path(3750,-750)(4250,-750)(4000,-1250)(3750,-750)
\blacken\path(3750,10750)(4250,10750)(4000,11250)(3750,10750)

\path(6000,-2000)(6000,12000)
\path(6000,-4250)(6000,-3250)
\whiten\path(5750,-3500)(6250,-3500)(6000,-3250)(5750,-3500)
\blacken\path(5750,-750)(6250,-750)(6000,-1250)(5750,-750)
\blacken\path(5750,10750)(6250,10750)(6000,11250)(5750,10750)

\path(8000,-2000)(8000,12000)
\path(8000,-4250)(8000,-3250)
\whiten\path(7750,-3500)(8250,-3500)(8000,-3250)(7750,-3500)
\blacken\path(7750,-750)(8250,-750)(8000,-1250)(7750,-750)
\blacken\path(7750,10750)(8250,10750)(8000,11250)(7750,10750)

\path(10000,-2000)(10000,12000)
\put(9800,-2750){\tiny$N$}
\path(10000,-4250)(10000,-3250)
\whiten\path(9750,-3500)(10250,-3500)(10000,-3250)(9750,-3500)
\blacken\path(9750,-750)(10250,-750)(10000,-1250)(9750,-750)
\blacken\path(9750,10750)(10250,10750)(10000,11250)(9750,10750)

\path(12000,-2000)(12000,12000)
\put(11800,-2750){\tiny$N+1$}
\path(12000,-4250)(12000,-3250)
\whiten\path(11750,-3500)(12250,-3500)(12000,-3250)(11750,-3500)
\blacken\path(11750,-1250)(12250,-1250)(12000,-750)(11750,-1250)
\blacken\path(11750,750)(12250,750)(12000,1250)(11750,750)
\blacken\path(11750,2750)(12250,2750)(12000,3250)(11750,2750)
\blacken\path(11750,4750)(12250,4750)(12000,5250)(11750,4750)
\blacken\path(11750,6750)(12250,6750)(12000,7250)(11750,6750)
\blacken\path(11750,8750)(12250,8750)(12000,9250)(11750,8750)
\blacken\path(11750,10750)(12250,10750)(12000,11250)(11750,10750)

\path(14000,-2000)(14000,12000)
\path(14000,-4250)(14000,-3250)
\whiten\path(13750,-3500)(14250,-3500)(14000,-3250)(13750,-3500)
\blacken\path(13750,-1250)(14250,-1250)(14000,-750)(13750,-1250)
\blacken\path(13750,750)(14250,750)(14000,1250)(13750,750)
\blacken\path(13750,2750)(14250,2750)(14000,3250)(13750,2750)
\blacken\path(13750,4750)(14250,4750)(14000,5250)(13750,4750)
\blacken\path(13750,6750)(14250,6750)(14000,7250)(13750,6750)
\blacken\path(13750,8750)(14250,8750)(14000,9250)(13750,8750)
\blacken\path(13750,10750)(14250,10750)(14000,11250)(13750,10750)

\path(16000,-2000)(16000,12000)
\path(16000,-4250)(16000,-3250)
\whiten\path(15750,-3500)(16250,-3500)(16000,-3250)(15750,-3500)
\blacken\path(15750,-1250)(16250,-1250)(16000,-750)(15750,-1250)
\blacken\path(15750,750)(16250,750)(16000,1250)(15750,750)
\blacken\path(15750,2750)(16250,2750)(16000,3250)(15750,2750)
\blacken\path(15750,4750)(16250,4750)(16000,5250)(15750,4750)
\blacken\path(15750,6750)(16250,6750)(16000,7250)(15750,6750)
\blacken\path(15750,8750)(16250,8750)(16000,9250)(15750,8750)
\blacken\path(15750,10750)(16250,10750)(16000,11250)(15750,10750)

\path(18000,-2000)(18000,12000)
\put(17800,-2750){\tiny$M$}
\path(18000,-4250)(18000,-3250)
\whiten\path(17750,-3500)(18250,-3500)(18000,-3250)(17750,-3500)
\blacken\path(17750,-1250)(18250,-1250)(18000,-750)(17750,-1250)
\blacken\path(17750,750)(18250,750)(18000,1250)(17750,750)
\blacken\path(17750,2750)(18250,2750)(18000,3250)(17750,2750)
\blacken\path(17750,4750)(18250,4750)(18000,5250)(17750,4750)
\blacken\path(17750,6750)(18250,6750)(18000,7250)(17750,6750)
\blacken\path(17750,8750)(18250,8750)(18000,9250)(17750,8750)
\blacken\path(17750,10750)(18250,10750)(18000,11250)(17750,10750)

\put(20000,-2750){\fbox{$w,r$}}

\end{picture}

\end{minipage}
\end{center}

\caption[Frozen vertices within $S_0$]{Frozen vertices within $S_0$. The final $M-N$ columns of vertices produce the prefactor in (\ref{S0-tf}), while the remainder of the lattice represents $Z_N$.}
\end{figure}
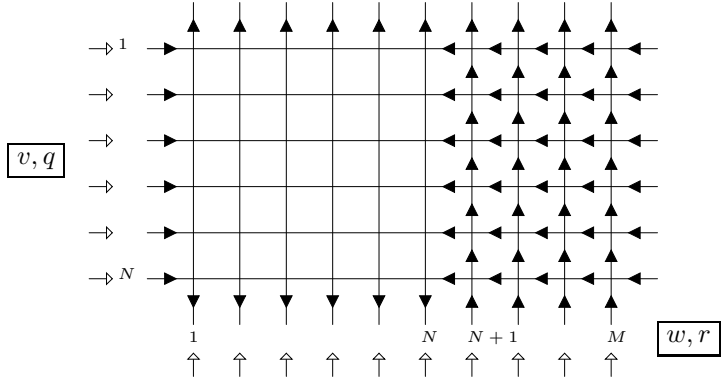

The vertices in the last $M-N$ columns of a lattice configuration must be of the form $b_{-}(v_j,q_j,w_k,r_k)$, or else the configuration vanishes. Peeling away this block of frozen vertices, we find that $S_0$ is equal to $Z_N$ up to the overall factor $\prod_{j=1}^{N}\prod_{k=N+1}^{M} b_{-}(v_j,q_j,w_k,r_k)$.
}
\end{conditions}

\end{proof}

\subsection{Factorized expression for $S_n(\{u,p\}_n,\{v,q\}_N,\{w,r\}_M)$}

\begin{lemma}
{\rm Let $v_j+q_j$ satisfy the equation (\ref{tf-bethe2}) for all $1\leq j \leq N$. In the presence of this constraint, the scalar product $S_n$ has the factorized expression
\begin{align}
&
S_n\Big( \{u,p\}_n, \{v,q\}_N, \{w,r\}_M \Big)
=
\prod_{j=1}^{n}
[2p_j]^{\frac{1}{2}}
\prod_{j=1}^{N}
[2q_j]^{\frac{1}{2}}
\prod_{j=1}^{\widetilde{N}}
[2r_j]^{\frac{1}{2}}
\times
\label{Scalc-tf}
\\
 &
\prod_{1 \leq j<k \leq n}
[u_k-u_j+p_j+p_k]
\prod_{1\leq j<k \leq N}
[v_j-v_k+q_j+q_k]
\prod_{1\leq j<k \leq \widetilde{N}}
[w_k-w_j+r_j+r_k]
\times
\nonumber
\\
 &
\prod_{j=1}^{n}
\prod_{k=1}^{\widetilde{N}}
[w_k-u_j+p_j+r_k]
\prod_{j=1}^{N}
\prod_{k=\widetilde{N}+1}^{M}
[v_j-w_k+q_j-r_k]
\times
\nonumber
\\
&
\prod_{j=1}^{n}
\prod_{k=1}^{N}
\frac{1}{\displaystyle{
[u_j+p_j-v_k-q_k]
}}
\prod_{j=1}^{n}
\left(
(-)^N
\prod_{k=1}^{M}
[u_j+p_j-w_k+r_k]
+
\prod_{k=1}^{M}
[u_j+p_j-w_k-r_k]
\right).
\nonumber
\end{align}

}
\end{lemma}

\begin{proof}

We begin by stating that the conditions {\bf 1}--{\bf 4} are uniquely determining.\footnote{This is proved along very similar lines to lemma 7.} Hence we need only verify that (\ref{Scalc-tf}) satisfies properties {\bf 1}--{\bf 4}.

\setcounter{conditions}{0}
\begin{conditions}
{\rm
By studying (\ref{Scalc-tf}) we see that $S_n$ has dependence on $\{w_{\widetilde{N}+1},\ldots,w_M\}$ and $\{r_{\widetilde{N}+1},\ldots,r_M\}$ only through the terms
\begin{align*}
\prod_{j=1}^{N}
\prod_{k=\widetilde{N}+1}^{M}
[v_j-w_k+q_j-r_k],
\ 
\prod_{j=1}^{n}
\left(
(-)^N
\prod_{k=1}^{M}
[u_j+p_j-w_k+r_k]
+
\prod_{k=1}^{M}
[u_j+p_j-w_k-r_k]
\right)
\end{align*}
which are clearly both invariant under the permutation $\{w_j,r_j\} \leftrightarrow \{w_k,r_k\}$ for all $j,k \in \{\widetilde{N}+1,\ldots,M\}$. 

}
\end{conditions}

\begin{conditions}
{\rm
Since $v_j+q_j$ is a root of the equation (\ref{tf-bethe2}) for all $1\leq j \leq N$, it follows that
\begin{align*}
\prod_{j=1}^{n}
\prod_{k=1}^{N}
\frac{1}{[u_j+p_j-v_k-q_k]}
\prod_{j=1}^{n}
\left(
(-)^N
\prod_{k=1}^{M}
[u_j+p_j-w_k+r_k]
+
\prod_{k=1}^{M}
[u_j+p_j-w_k-r_k]
\right)
\end{align*}
is a trigonometric polynomial in $u_n$ of degree $M-N$. The remaining terms in (\ref{Scalc-tf}) comprise a trigonometric polynomial of degree $N-1$ in $u_n$. Therefore the entire expression (\ref{Scalc-tf}) is a trigonometric polynomial of degree $M-1$ in $u_n$. In addition, the required factor $\prod_{j=1}^{\widetilde{N}}[w_j-u_n+p_n+r_j]$ is present in (\ref{Scalc-tf}).

}
\end{conditions}

\begin{conditions}
{\rm
The recursion relation (\ref{Srec0-tf}) is proved by setting $u_n+p_n=w_{\widetilde{N}+1}+r_{\widetilde{N}+1}$ in (\ref{Scalc-tf}) and rearranging the factors in the resulting equation. Since this procedure is tedious but elementary in nature, we shall omit the details. 
}
\end{conditions}

\begin{conditions}
{\rm
Setting $n=0$ in (\ref{Scalc-tf}) gives
\begin{align}
&
S_0\Big(\{v,q\}_N,\{w,r\}_M\Big)
=
\prod_{j=1}^{N} \prod_{k=N+1}^{M}
[v_j-w_k+q_j-r_k]
\times
\label{S0proof-tf}
\\
&
\prod_{j=1}^{N} [2q_j]^{\frac{1}{2}} [2r_j]^{\frac{1}{2}}
\prod_{1\leq j < k \leq N} [v_j-v_k+q_j+q_k] [w_k-w_j+r_j+r_k].
\nonumber
\end{align}
Comparing equation (\ref{S0proof-tf}) with the factorized expression (\ref{lem-tf}) for the domain wall partition function, we verify (\ref{S0-tf}).
}
\end{conditions}

\end{proof}

\subsection{Evaluation of $S_N(\{u,p\}_N,\{v,q\}_N,\{w,r\}_M)$}

For completeness, we write the $n=N$ case of equation (\ref{Scalc-tf}) explicitly. We have
\begin{align}
&
S_N
=
\prod_{j=1}^{N}
[2p_j]^{\frac{1}{2}} [2q_j]^{\frac{1}{2}}
\prod_{1\leq j<k \leq N}
[u_k-u_j+p_j+p_k]
[v_j-v_k+q_j+q_k]
\times
\\
&
\frac{\displaystyle{
\prod_{j=1}^{N} \prod_{k=1}^{M}
[v_j-w_k+q_j-r_k]
}}
{\displaystyle{
\prod_{j,k=1}^{N}
[u_j+p_j-v_k-q_k]
}}
\prod_{j=1}^{N}
\left(
(-)^N \prod_{k=1}^{M} [u_j+p_j-w_k+r_k]
+
\prod_{k=1}^{M} [u_j+p_j-w_k-r_k]
\right).
\nonumber
\end{align}
Let us consider the zeros of this expression in the variable $u_N+p_N$. Studying the first line of the right hand side, we see that $N-1$ of the zeros are of the same type as in the domain wall partition function. The remaining zeros, in the second line, coincide with the $M-N$ roots of the Bethe equation (\ref{tf-bethe5}) which are {\it different} from $\{v_1+q_1,\ldots,v_N+q_N\}$.

\end{document}